\documentclass[11pt]{article}
\usepackage{makeidx}  
\usepackage{fullpage}
\usepackage{times}
\usepackage{float}

\usepackage{subfigure}
\usepackage{color}
\usepackage{url}
\usepackage{graphicx}
\usepackage{amsmath}
\usepackage{amssymb, mathrsfs, bbm, comment}

\usepackage{tikz}
\usepackage{ifthen}
\usetikzlibrary{arrows,shapes}

\usetikzlibrary{decorations.pathmorphing}

\definecolor{green2}{RGB}{34,139,34}

\newcommand{\E}{\mathcal{E}}

\newcommand{\A}{\mathbf{A}}
\newcommand{\B}{\mathbf{B}}

\newcommand{\qed}{\mbox{}\hspace*{\fill}\nolinebreak\mbox{$\rule{0.6em}{0.6em}$}}

\newcommand{\expect}{{\bf \mbox{\bf E}}}
\newcommand{\prob}{{\bf \mbox{\bf Pr}}}

\definecolor{gray}{rgb}{0.5,0.5,0.5}
\newcommand{\e}{{\epsilon}}

\newcommand{\YES}{{\bf YES~}}
\newcommand{\NO}{{\bf NO~}}

\newcommand{\ctwo}{512}

\newtheorem{lemma}{Lemma}[section]
\newtheorem{theorem}[lemma]{Theorem}

\newtheorem{claim}[lemma]{Claim}

\newtheorem{definition}[lemma]{Definition}
\newtheorem{remark}[lemma]{Remark}
\newtheorem{rem}{Remark}[lemma]

\newcommand{\Int}{\text{Int}_M(v)}
\newcommand{\Ext}{\text{Ext}_M(v)}
\newcommand{\Bound}{\partial_M(v)}

\newcommand{\eps}{\varepsilon}
\newcommand{\bound}[1][C]{\operatorname{\mathbf{U}_{#1, s^*}}}

\newenvironment{proof}{{\bf Proof:}}{$\qed$\par}

\newenvironment{proofof}[1]{\noindent{\bf Proof of #1:}}{$\qed$\par}

\usepackage{hyperref}
\hypersetup{
bookmarksnumbered
}

{\makeatletter
 \gdef\xxxmark{%
   \expandafter\ifx\csname @mpargs\endcsname\relax 
     \expandafter\ifx\csname @captype\endcsname\relax 
       \marginpar{xxx}
     \else
       xxx 
     \fi
   \else
     xxx 
   \fi}
 \gdef\xxx{\@ifnextchar[\xxx@lab\xxx@nolab}
 \long\gdef\xxx@lab[#1]#2{{\bf [\xxxmark #2 ---{\sc #1}]}}
 \long\gdef\xxx@nolab#1{{\bf [\xxxmark #1]}}
 \long\gdef\xxx@lab[#1]#2{}\long\gdef\xxx@nolab#1{}%
}

\makeatletter
\newcommand\footnoteref[1]{\protected@xdef\@thefnmark{\ref{#1}}\@footnotemark}
\makeatother

\newcommand{\wt}{\widetilde}

\newcommand{\wh}[1]{\widehat{#1}}

\newcommand{\IHP}{{\bf IHP}}
\newcommand{\DIHP}{{\bf DIHP}}

\begin{document}
\title{An Optimal Space Lower Bound for Approximating MAX-CUT}
\author{Michael Kapralov\\EPFL\\michael.kapralov@epfl.ch \and Dmitry Krachun\\University of Geneva\\dima@krachun.ru}

\maketitle

\begin{abstract}
We consider the problem of estimating the value of MAX-CUT in a graph in the streaming model of computation.  At one extreme, there is a trivial $2$-approximation for this problem that uses only $O(\log n)$ space, namely, count the number of edges and output half of this value as the estimate for the size of the MAX-CUT. On the other extreme, for any fixed $\eps > 0$, if one allows $\tilde{O}(n)$ space, a $(1+\eps)$-approximate solution to the MAX-CUT value can be obtained by storing an $\tilde{O}(n)$-size sparsifier that essentially preserves  MAX-CUT value. 

Our main result is that any (randomized) single pass streaming algorithm that breaks the $2$-approximation barrier requires $\Omega(n)$-space, thus resolving the space complexity of any non-trivial approximations of the MAX-CUT value to within polylogarithmic factors in the single pass streaming model. We achieve the result by presenting a tight analysis of the Implicit Hidden Partition Problem introduced by Kapralov et al.[SODA'17] for an arbitrarily large number of players. In this problem a number of players receive random matchings of $\Omega(n)$ size together with random bits on the edges, and their task is to determine whether the bits correspond to parities of some hidden bipartition, or are just uniformly random. 

Unlike all previous Fourier analytic communication lower bounds, our analysis does not directly use bounds on the $\ell_2$ norm of Fourier coefficients of a typical message at any given weight level that follow from hypercontractivity.  Instead, we use the fact that graphs received by players are sparse (matchings) to  obtain strong upper bounds on the $\ell_1$ norm of the Fourier coefficients of the messages of individual players using their special structure, and then argue, using the convolution theorem, that similar strong bounds on the  $\ell_1$ norm are essentially preserved (up to an exponential loss in the number of players) once messages of different players are combined. We feel that our main technique is likely of independent interest.
 \end{abstract}

\newcommand{\bool}{\{0,1\}}

\renewcommand{\P}{{\cal P}}
\newcommand{\D}{{\cal D}}

\newcommand{\I}{\mathcal{I}}

\setcounter{page}{0}

\thispagestyle{empty}
\newpage

\tableofcontents

\newpage

\section{Introduction}

In the MAX-CUT problem an undirected graph is given as input, and the goal is to find a bipartition of the vertices of this graph (or, equivalently, a {\em cut}) that maximizes the number of edges that cross the bipartition. The size of a MAX-CUT is the number of edges that cross the optimal bipartition.
In this paper, we study the space complexity of approximating the MAX-CUT size in an
undirected graph in the streaming model of computation. Our main result is a strong lower bound
(optimal to within polylogarithmic factors) on the space required for a
non-trivial approximation to MAX-CUT size. 

Specifically we consider a space bounded algorithm that is presented with a
stream of edges of a graph on known vertex set $[n] \triangleq \{1,\ldots,n\}$ and 
is required to output an $\alpha$-approximation to the size of the maximum cut
in the graph. 
An algorithm that simply counts $m$, the number of edges in the graph, and reports $m/2$
requires space $O(\log n)$ and produces a $2$-approximation to the size of a
maximum cut since the MAX-CUT size is at most $m$ and at least $m/2$ in any undirected graph. 
On the other extreme, for any $\epsilon > 0$, it is 
possible to maintain a cut-preserving sparsifier of the graph using $\tilde{O}(n)$ space that allows one to recover 
a $(1+\epsilon)$-approximation to the maximum cut value -- in fact, one can recover the actual vertex partition as well in this case.
Till recently it was open as to whether such good approximations could be
obtained with polylogarithmic space. This question was resolved in the negative
by \cite{KKS15, kogan-krauthgamer14}. In particular,~\cite{KKS15} showed that any
$(2-\epsilon)$-approximation algorithm requires $\tilde{\Omega}(\sqrt{n})$ space, and~\cite{KapralovKSV17} ruled out the possibility of an approximation scheme in $o(n)$ space. This however left open the possibility
that a non-trivial approximation (i.e. better than the trivial $2$-approximation described above) can be achieved in $o(n)$ space. In this work, we settle this problem by showing
that no sublinear space algorithm can achieve a strictly better than $2$-approximation to the size of the maximum cut:

\begin{theorem}\label{thm:main}
For every $\e\in (0, 99/100)$, every randomized single-pass streaming algorithm that yields a $(2-\e)$-approximation to MAX-CUT size must use 
$n/(1/\e)^{O(1/\e^2)}=\Omega(n)$ space, where $n$ denotes the number of vertices in the graph.
\end{theorem}

Our main technical contribution is a nearly  optimal lower bound on the communication complexity of the {\bf Implicit Hidden Partition} problem introduced in~\cite{KapralovKSV17}. The implicit hidden partition problem is a multiple-player communication game
 where many players are
given labellings of sparse subsets of edges and must determine if they are consistent with a
bipartition of vertices.  This setting is in contrast to the vanilla ``hidden partition problem'' used in several previous works on MAX-CUT lower bounds, where one player is given a cut in 
a complete graph on $n$ vertices and the other player is given a labelling of a linear number of edges of the 
graph where the labelling supposedly indicates whether those edge cross the cut, and 
the two players have to decide if the labelling is consistent with the given
cut. Gavinsky et al~\cite{GKKRW07} gave a tight $\Omega(\sqrt{n})$ lower bound on the
one-way communication complexity of this problem, and this was used 
by \cite{KKS15} for instance to give a $\tilde{\Omega}(\sqrt{n})$ lower bound on the space needed to find
a $(2-\epsilon)$-approximation of MAX-CUT. Since in the {\em implicit} hidden partition problem no player has an explicit knowledge of the bipartition, this problem
plausibly has a linear ($\Omega(n)$) lower bound on the communication
complexity and this is what we prove. 

The main technical tool underlying our analysis is a novel and general way of using the {\bf convolution theorem} in Fourier analysis to analyze information conveyed by the combined messages of multiple players (which corresponds to the intersection of their individual messages). While this idea has recently been used for proving lower bounds on the streaming complexity of MAX-CUT~\cite{KapralovKSV17} and the sketching complexity of subgraph counting~\cite{KallaugherKP18}, in both of these works the convolution theorem is applied in a rather restricted setting. Specifically, the structure of the Fourier transforms of the messages of the players is such that convolution simply amounts to multiplication in Fourier domain, i.e. only a single nonzero contributes to the corresponding sum of Fourier coefficients. In our setting convolving the Fourier transforms of the players' messages leads to contributions across different levels of the weight spectrum, and analyzing such processes requires a new technique. Our main insight is the idea of controlling the $\ell_1$ norm of the Fourier transform of the intersection of the players' messages as opposed to the $\ell_2$ norm, bounds on which follow more naturally as a consequence of hypercontractivity (note that $\ell_2$ bounds on various levels of the Fourier spectrum that follow from the hypercontractive inequality underlie the analysis of the Boolean Hidden Matching problem of~\cite{GKKRW07}, as well as recent works on streaming and sketching lower bounds through Fourier analysis~\cite{KKS15, kogan-krauthgamer14, KapralovKSV17, KallaugherKP18}). Conceptually, the idea of controlling the $\ell_1$ norm stems from the fact that since individual players receive parity information of some hidden vector $X$ across edges of a sparse graph (a matching), strong upper bounds on the $\ell_1$ norm of the Fourier transform of the corresponding player's message follow (due to sparsity of the graph), and these $\ell_1$ bounds remain approximately preserved when the players' functions are multiplied (i.e. when the players' messages are combined).

\medskip
\noindent
{\bf Related work:} 
The streaming model of computation, formally introduced in the seminal work of~\cite{AlonMS96} and motivated by applications in processing massive datasets, is an extremely well-studied model for designing sublinear space algorithms.  
The past decade has seen an extensive body of work on understanding the space complexity of fundamental graph problems in the streaming model; see, for instance, the survey by McGregor~\cite{McGregorSurvey14}.
It is now known that many fundamental problems admit streaming algorithms that only require $\tilde O(n)$ space (i.e. they do not need space to load the edge set of the graph into memory) --  e.g. sparsifiers~\cite{anh-guha, kl11, agm_pods, KapralovLMMS14}, spanning trees~\cite{agm}, matchings~\cite{guha-ahn-1, guha-ahn-3, gkk:streaming-soda12, kap13, GO12, HRVZ13, Konrad15, AssadiKLY15}, spanners~\cite{agm_pods, kw14}.  Very recently it has been shown that it is sometimes possible to approximate the {\em cost} of the solution without even having enough space to load the {\em vertex set} of the graph into memory (e.g.~\cite{KKS14, onak15, McGregor15}).  Our work contributes to the study of streaming algorithms, by providing a tight impossibility result for non-trivially approximating MAX-CUT value in $o(n)$ space.

\medskip
\noindent

\paragraph{Organization} In section~\ref{sec:comm-prob} we introduce our communication problem (\DIHP), give a reduction from \DIHP~ to MAX-CUT, state our main technical theorem (Theorem~\ref{thm:main-comm}) on the communication complexity of \DIHP\ in this section and prove Theorem~\ref{thm:main} assuming this result.
The rest of the paper is devoted to proving Theorem~\ref{thm:main-comm}. Section~\ref{sec:prelims} presents preliminaries on Fourier analysis and basic combinatorics of random matchings. Section~\ref{sec:defs} presents our basic setup for proving Theorem~\ref{thm:main-comm}, introduces a crucial definition of $(C, s^*)$-bounded sets, and proves some of their basic properties. Section~\ref{sec:outline-app} presents a technical overview of the analysis: we start with a simple component growing protocol for \DIHP~that serves as a good model for our lower bound proof, give a direct analysis for this protocol, state the main technical components of our general lower bound proof and use the simple protocol to illustrate some of the steps. The rest of the paper is devoted to formally proving Theorem~\ref{thm:main-comm}. Section~\ref{sec:main-result} proves Theorem~\ref{thm:main-comm} assuming a key technical lemma, whose proof is given in Section~\ref{sec:inductionstep}.

\section{Communication problem and hard distribution}\label{sec:comm-prob}

In this section we introduce a $T$-player ``sequential'' communication problem and state our lower bound for
this problem. We first describe the general model in which this problem is presented.

We consider a communication problem where $T$ players receive, sequentially, public inputs $M_t$ and private
inputs $w_t$, for $t \in [T]$. Their goal is to compute some joint function
$F(M_1,\ldots,M_T;w_1,\ldots,w_T)$. At stage $t$, the $t$th player announces its message $S_t =
r_t(M_1,\ldots,M_t; S_1,\ldots,S_{t-1}; w_t)$ for some function $r_t$. The message $S_T$ is defined to be the output of the protocol.
The complexity of the protocol is the maximum length of the message $\{S_t\}_{t \in [T]}$. We consider the
setting where the inputs are drawn from some distribution $\mu$ and the error of the protocol is the
probability that its output does not equal $F(M_1,\ldots, M_T, w_1, \ldots, w_T)$. We now describe the specific communication
problem we consider in this work. For a protocol $\Pi$ we let $|\Pi|$ denote the maximum size of messages posted by players.

\paragraph{Implicit Hidden Partition (\IHP) Problem of~\cite{KapralovKSV17}.} 
We define a parametrized class of problems \IHP$(n,\alpha,T)$ for
positive integers $n$ and $T$ and real $\alpha \in (0,1/2)$ as follows:
\IHP$(n,\alpha,T)$ is a $T$-player problem with public inputs $M_t \in
\{0,1\}^{\alpha n \times n}$ being incidence matrices of matchings (so their
rows sum to $2$ and columns sum to at most $1$), and the private inputs are
$w_t \in \{0,1\}^{\alpha n}$. Their goal is to decide the Boolean function
$F(M_1,\ldots,M_T;w_1,\ldots,w_T)$ which is \YES if and only if there exists
$x \in \{0,1\}^n$ such that $w_t = M_t x (\text{mod}~2)$ for all $t \in [T]$, and \NO
otherwise.

By associating $[n]$ with the vertices of a graph $G$, we may think of $x$
as a partition (cut) of the graph $G$ whereas the edges are the set
$\cup_{t \in [T]} M_t$. The condition $w_t = M_t x$ enforces that an edge
crosses the cut if and only if it is labelled $1$ in $w_t$. Thus the
communication problem corresponds to asking if there is a cut consistent
with the labelling of edges which are partitioned into matchings and
presented as such.

\paragraph{Distributional Implicit Hidden Partition (\DIHP) Problem.} We will work with the following distributional version of \IHP.
We define a distribution $\D^Y$ supported on \YES instances which is obtained by
sampling $X^* \in \{0,1\}^n$ uniformly, sampling $M_t$'s independently and
uniformly from the set of $\alpha n$ sized matchings on $[n]$ and setting
$w_t = M_t X^*$. The distribution $\D^N$ supported mostly on \NO instances is
obtained by sampling $M_t$'s as above, and $w_t$'s independently and
uniformly from $\{0,1\}^{\alpha n}$. The final distribution $\D$ is
simply $\D = \frac12(\D^Y +\D^N)$. We use \DIHP\ to denote the
distributional \IHP\ problem where instances are drawn from $\D$.

\newcommand{\thmMainComm}{There exists a constant $C_0>0$ such that for every $\e\in (0, 1)$ and every sufficiently large $n$ every protocol $\Pi$ for \DIHP$(n, \alpha, T)$ with $\alpha = 10^{-11}$, $T=\ctwo/(\alpha \e^2)$ that succeeds with probability at least $2/3$ satisfies $|\Pi|\geq n/(C_0/\e)^{C_0/\e^2}$. 
}
The following theorem is the main technical contribution of the paper:
\begin{theorem}\label{thm:main-comm}
\thmMainComm
\end{theorem}

\begin{remark}
We note that there exists a protocol $\Pi$ with $|\Pi|=n/c^T$ which solves \DIHP$(n,\alpha,T)$ for constant $\alpha$, where $c=c(\alpha)>1$, which makes our lower bound close to tight (up to the dependence on $1/\e$ in the base of the exponent, which we think can likely be removed by a more careful, but somewhat more complex, analysis). The protocol works as follows. The first player posts bits on the first $s/(1+\alpha)^T$ edges of the matching $M_1$. Then each player posts bits on all edges incident to at least one of the edges which have been revealed previously. In other words, the players keep growing connected components in the graph whose edges they reveal the bits on. It can be shown that with high probability each player will post at most $s$ bits and a cycle will be found thus solving the problem. See Section \ref{sec:simple-protocol} for more details.
\end{remark}
\paragraph{Reduction from \DIHP~to MAX-CUT.} We now give a reduction from \DIHP~to MAX-CUT (see Theorem~\ref{thm:reduction} below). Our main result, Theorem~\ref{thm:main}, then follows by putting together Theorem~\ref{thm:reduction} and the yet to be proved Theorem~\ref{thm:main-comm}. The rest of the paper is then devoted to proving Theorem~\ref{thm:main-comm}.

We start with the reduction. We show that for every $\e > 0$, if there exists a single-pass streaming algorithm {\bf ALG} that uses space $s = s(n)$ and approximates the MAX-CUT value to within a factor of $(2 - \e)$ with probability at least $9/10$, then instances of  \DIHP$(n, \alpha, T)$ with $\alpha=10^{-11}$ and $T=\ctwo/(\alpha \e^2)$ can be distinguished with probability at least $2/3$ by a protocol that uses messages of size at most $s$. This reduction combined with Theorem~\ref{thm:main-comm}, establishes our main result, namely, every single-pass streaming algorithm needs $\Omega(n)$ space to approximate MAX-CUT value to within a factor of $(2-\e)$ for every fixed $\e > 0$.

The main idea of the reduction is to map instances of \DIHP~to instances of MAX-CUT by only considering those edges where the corresponding entry of the $w_t$ vector is $1$. This in turn induces a distribution over \YES and \NO instances for MAX-CUT as below: 

\begin{description}
\item[YES]  Let $G'_t=(V, E'_t)$ be the (bipartite) graph obtained by including those edges in $E_t$ that cross the chosen (hidden) bipartition. Let $E':=E'_1\cup E'_2\cup \ldots \cup E'_{T}$, and let $G'(V,E')$ denote the final graph.

\item[NO]   Let $G'_t=(V, E'_t)$ be the graph obtained by including each edge in $E_t$ independently with probability $1/2$. Let $E':=E'_1\cup E'_2\cup \ldots \cup E'_T$, and let $G'(V,E')$ denote the final graph.

\end{description}

We denote the input distribution defined above by $\D^{Y}_{graph}$ (\YES case) and $\D^{N}_{graph}$ (\NO case) respectively. Let $\D_{graph}=\frac1{2}\D^{Y}_{graph}+\frac1{2}\D^{N}_{graph}$. We note that the graphs generated by our distribution $\D_{graph}$ are in general multigraphs. We note that the expected number of repeated edges is only $O(1/\e^2)$, and edge multiplicities are bounded by $2$ with high probability.
  
Using this distribution we get
\begin{theorem}[Reduction from \DIHP~to MAX-CUT]\label{thm:reduction}
For every $\e, \alpha \in (n^{-1/10}, 1)$ and $T=\ctwo/(\alpha\e^2)$, the following conditions hold for sufficiently large $n$.  Let {\bf ALG} denote a (possibly randomized) single-pass streaming algorithm for approximating MAX-CUT value in (multi)graphs to within a factor of $(2 - \e)$ using space $s=s(n)$ on graphs on $n$ nodes with probability at least $9/10$. Then {\bf ALG} can be used to obtain a deterministic protocol $\Pi$ for the \DIHP$(n, \alpha, T)$ with $|\Pi|\leq s$ that succeeds with probability at least $2/3$ over the randomness of the input. This holds even if {\bf ALG} is only required to work on multigraphs that contain at most $O(1/\e^2)$ repeated edges, and edge multiplicity is bounded by $2$.
\end{theorem}

The proof relies on the following Lemma, which establishes that there is almost a factor of $2$ gap between MAX-CUT value in $\D_{graph}^Y$ and $\D_{graph}^N$:
\newcommand{\lemmagap}{
For every $\e, \alpha\in (n^{-1/10},1), \alpha<1/4$ and $T=512/(\alpha\e^2)$  if $G'_T=(V, E'_T), |V|=n, |E'|=m$ be generated according to the process above, then for sufficiently large $n$ there exists $m_0=m_0(n,\alpha,T)$ such that in the \YES case the MAX-CUT value is at least $m_0$, and in the \NO case the MAX-CUT value is at most $m_0/(2-\e)$, both with probability at least $1-1/\sqrt{n}$.
}

\begin{lemma}\label{lm:gap}
\lemmagap
\end{lemma}

The proof of the lemma uses the following version of Chernoff bounds.

\newcommand{\lemmaChernoff}{
Let $X=\sum_{i=1}^n X_i$, where $X_i$ are Bernoulli 0/1 random variables satisfying, for every $k\in[n]$, $\mathbb{E}[X_k|X_1,\ldots,X_{k-1}]\leq p$ for some $p\in(0,1)$. Let $\mu=np$. Then for all $\Delta>0$ \[ \prob[X\geq \mu+\Delta]\leq \exp\left(-\frac{\Delta^2}{2\mu+2\Delta}\right). \]
}
\begin{lemma}\label{lm:chernoff}
\lemmaChernoff
\end{lemma}
The (somewhat standard) proof is given in Appendix~\ref{tools}.  We also need 
\begin{lemma}\label{lm:cut}
Let $G$ be a miltigraph with $n$ vertices and $m$ edges (counted with multiplicities) in which each edge has multiplicity at most $k$. Let $S\subset [n]$ be a uniformly random subset of vertices and $X$ be the number of edges crossing $(S, \bar{S})$. Then for any $\delta>0$ we have 
\[
\prob[X<m/2\cdot (1-\delta)]\leq \frac{k}{\delta^2 m}.
\] 
\end{lemma}
The proof is a simple application of Chebyshev's inequality, and is presented in Appendix~\ref{tools}. We now give

\noindent{ {\bf Proof outline of Lemma~\ref{lm:gap}}
We let $m_0=\frac{\alpha n T}{2}\cdot (1-\delta)$ with $\delta=\e/100$. In the \YES case the graph is bipartite so the value of MAX-CUT is equal to the number of edges in the graph. Since in the \YES case we only keep those edges of the matchings which cross a fixed random bipartition, and in the union of matchings every edge has multiplicity at most $T$, Lemma~\ref{lm:cut} ensures that the probability that the number of edges if smaller than $m_0$ is at most 
\[
\frac{T}{\delta^2 \alpha n T}=\frac{1}{\delta^2\alpha n}\leq 1/\sqrt{n},
\]
since $\e, \alpha>n^{-1/10}$ by assumption of the lemma.

We now consider the \NO case. Since every edge of the matchings is kept with probability $1/2$ independently of the others, by Lemma~\ref{lm:chernoff} with probability at least $1-\exp{\left(-\frac{\delta^2\alpha n T }{4(1+\delta)}\right)}$ the number $m$ of edges in the graph will be at most $\frac{\alpha n T}{2}\cdot (1+\delta)$. One then shows using Lemma~\ref{lm:chernoff} that for every cut $(S, \bar{S})$, where $S \subseteq V$, the probability that significantly more than half of the edges of the graph cross the cut is smaller than $2^{-2n}$. Taking a union bound over all $2^n$ possible cuts completes the proof. The detailed proof is provided in Appendix~\ref{tools}.
$\qed$}

Equipped with Lemma~\ref{lm:gap}, we can now give a proof of the reduction:

\begin{proofof}{Theorem~\ref{thm:reduction}}
By Lemma~\ref{lm:gap} in the \YES case the MAX-CUT value is at least $m_0$, and in the \NO case the MAX-CUT value is at most $m_0/(2-\e)$, both with probability at least $1-1/\sqrt{n}$.

Thus with high probability, the MAX-CUT value in a \YES instance of \DIHP\ is at least $(2 - \e)$ times the MAX-CUT value in a \NO instance of \DIHP. To complete the reduction, it now suffices to show that the algorithm {\bf ALG} can be simulated by a \DIHP\ protocol with message complexity at most $s$. This simulation can be done as follows: player $t$ upon receiving its input $(M_t, w_t)$, runs {\bf ALG} using the state posted by player $(t-1)$ on the set of edges in $M_t$ where the $w_t$ value is $1$. The state of the algorithm {\bf ALG} at the end is then posted by player $t$ on the board. The last player then outputs $\YES$ if {\bf ALG} outputs \YES and \NO otherwise.  Note that since we are evaluating the resulting protocol with respect to an input distribution, it is possible to make the protocol deterministic by fixing the randomness of {\bf ALG} appropriately.
\end{proofof}

Given  Theorem~\ref{thm:reduction}~and assuming Theorem~\ref{thm:main-comm}, our main theorem follows easily.

\begin{proofof}{Theorem~\ref{thm:main}}
The proof follows by putting together Theorem~\ref{thm:reduction} and Theorem~\ref{thm:main-comm}. 
\end{proofof}


\section{Preliminaries}\label{sec:prelims}

In this section we first review Fourier analysis on the boolean hypercube and give a version of the KKL bound that will be important in our analysis(Section~\ref{sec:prelims-fourier}), give basic facts about the total variation distance between distributions (Section~\ref{sec:prelims-tvd})  and state basic bounds on uniformly random matchings that we will use (Section~\ref{sec:prelims-random-matchings}). 

\subsection{Fourier analysis on the boolean hypercube}\label{sec:prelims-fourier}

Let $p: \bool^n\to \mathbb{R}$ be a real valued function defined on the boolean hypercube. We use the following normalization in the Fourier transform:
$$
\wh{p}(v)=\frac1{2^n}\sum_{x\in \bool^n} p(x)\cdot (-1)^{x\cdot v}.
$$
With this normalization the inverse transform is given by
$$
p(x)=\sum_{v\in \bool^n} \wh{p}(v)\cdot (-1)^{x\cdot v}.
$$

For a pair of functions $p, q: \bool^n\to \mathbb{R}$ the convolution of $p$ and $q$, denoted by $p*q$ is defined as 
$$
(p*q)(v)=\sum_{x\in \bool^n} p(x) q(x+v).
$$

We will use the relation between multiplication of functions in time domain and convolution in frequency domain to analyze the Fourier spectrum of $h_t=f_1\cdot f_2\cdot \ldots \cdot f_t$ (recall that $h_t$ is the indicator of $\B_t$ as per Definition~\ref{def:at}). With our normalization of the Fourier transform the convolution identity is
\begin{equation}\label{eq:convolution}
\begin{split}
\wh{(p\cdot q)}(v)&=\sum_{x\in \bool^n} \wh{p}(x)\wh{q}(x+v).\\
\end{split}
\end{equation}
Thus, for each $t=1,\ldots, T$ we have that 
$$
\wh{h}_t=\wh{f}_1*\ldots*\wh{f}_t.
$$
This identity will form the basis of our proof.

We will also need Parseval's equality, which with our normalization takes form
\begin{equation}\label{eq:parseval}
\begin{split}
||\wh{p}||^2&=\sum_{v\in \bool^n} \wh{p}(v)^2=\left(\frac1{2^n}\sum_{x\in \bool^n} p(x)\cdot (-1)^{x\cdot v}\right)^2=\frac1{2^n} \sum_{x\in \bool^n} p(x)^2=\frac1{2^n}||p||^2\\
\end{split}
\end{equation}

\begin{remark}\label{rm:indicator-l2}
If $f(x): \bool^n \to \bool$ is an indicator of a set $\A\subseteq \bool^n$, we have $||f||^2=|\A|$, so that $||\wh{f}||^2=\frac{|\A|}{2^n}$.
\end{remark}

\begin{definition}\label{def:tilde}
For a function $h:\bool^n\to \bool$ that is the indicator of a set $\A\subseteq \bool^n$ we write $\wt{h}$ to denote the Fourier transform of $h$ scaled by $2^n/|\A|$. Specifically, for every $v\in \bool^n$ we have
\begin{equation*}
\wt{h}(v)=\frac{2^n}{|\A|}\wh{h}(v)=\frac1{|\A|}\sum_{x\in \bool^n} h(x)\cdot (-1)^{x\cdot v}=\mathbb{E}_{x\in \A}[(-1)^{x\cdot v}].
\end{equation*}
\end{definition}
The following two important Lemmas will be proved in Appendix \ref{app-F}.
\begin{lemma}\label{lm:L1xorKKL}
Let $A\subset \{0,1\}^m$ be a set of cardinality at least $2^{m-d}$ with indicator function $f$. Then for every $y\in \{0,1\}^m$ and every $q\leq d$ one has 
\[
\sum_{\substack{x\in \{0,1\}^m \\ |x\oplus y|=q }}\left|\widetilde{f}(x)\right| \leq \sqrt{\binom{m}{q} \left(\frac{4d}{q}\right)^{q}},
\]
where for $x\in \{0,1\}^m$ by $|x|$ we denote the number of ones in $x$.
\end{lemma}

\begin{lemma}\label{lm:matching-fourier}
Let $M_t\in \bool^{\alpha n\times n}$ be the incidence matrix of a matching $M$,  where the rows correspond to edges $e$ of $M$ ($M_{eu}=1$ if $e$ is incident on $u$ and $0$ otherwise). Let $g:\bool^{\alpha n} \to \bool^s$ for some $s>0$. Let $a\in \bool^s$ and let 
$q:\bool^{\alpha n}\to \bool$ be the indicator of the set $\A_{reduced}:=\{z\in \bool^{\alpha n}: g(z)=a\}$.
Further, let $f:\bool^n\to \bool$ denote the indicator of the set 
$$
\A_{full}:=\{x\in \bool^n: g(Mx)=a\}.
$$
Then for any $v\in \bool^n$
\begin{equation}
\hat f(v)=\left\lbrace
\begin{array}{lc}
0,&\text{~if~$v$~cannot be perfectly matched via edges of~}M\\
\hat q(w),&w\text{~the perfect matching of $v$ using edges of~}M\text{~o.w.~}\\
\end{array}
\right.
\end{equation}
Furthermore, the perfect matching of $v$, when it exists, is unique. The second condition above is equivalent to the existence of $w\in \{0, 1\}^{\alpha n}=\bool^{M}$ such that $v=M^Tw$. Thus, Fourier coefficients of $f$ are indexed by sets of edges of $M$. Note that nonzero weight $k$ coefficients of $\hat q$ are in one to one correspondence with nonzero weight $2k$ coefficients of $\hat f$, i.e. the only nonzero Fourier coefficients of $\wh{f}$ are of the form $\wh{f}(M^Tw)=\wh{q}(w)$ for some $w\in \bool^M$.
\end{lemma}

\subsection{Total variation distance}\label{sec:prelims-tvd}

We define the notion of total variation distance between two distributions and state some of its useful properties. For a random variable $X$ taking values on a finite sample space $\Omega$ we let $p_X(\omega), \omega\in \Omega$ denote the pdf of $X$. For a subset $A\subseteq \Omega$ we use the notation $p_X(A):=\sum_{\omega\in A} p_X(\omega)$.
We will use the total variation distance $||\cdot ||_{tvd}$ between two distributions:

\begin{definition}[Total variation distance]
Let $\mu, \nu$ be two probability measures on a finite space $\Omega$. The total variation distance between $\mu$ and $\nu$ is given by
$V(\mu, \nu)=\max_{\Omega'\subseteq \Omega} (\mu(\Omega')-\nu(\Omega'))=\frac1{2}\sum_{\omega\in \Omega} |\mu(\omega)-\nu(\omega)|$.
\end{definition}
\begin{definition}
Let $X,Y$ be two random variables taking values on a finite domain $\Omega$. We denote the pdfs of $X$ and $Y$ by $p_X$ and $p_Y$ respectively. The total variation distance between $X$ and $Y$ is defined to be the total variation distance between $p_X$ and $p_Y$. We will write $||X-Y||_{tvd}$ to denote the total variation distance between $X$ and $Y$.
\end{definition}

We will need the following claim.
\begin{claim}[Claim~6.5 of~\cite{KKS15}]\label{cl:1}
Let $X, Y$ be two random variables. Let $W$ be independent of $(X, Y)$. Then for any function $f$ one has 
$||f(X, W)-f(Y, W)||_{tvd}\leq ||X-Y||_{tvd}$.
\end{claim}

\subsection{Basic combinatorics of random matchings}\label{sec:prelims-random-matchings}

Here we define several quantities related to random matchings and state Lemmas concerning random matchings. Proofs are given in Appendix \ref{app:combi}.

\begin{definition}
For every $\alpha\in (0, 1)$, every integer $\ell$ between $0$ and $n/2$ we let $p_\alpha(\ell,n)$ denote the probability, over the choice of a uniformly random matching $M$ of size $\alpha n$, that a fixed set of $2\ell$ points is perfectly matched by $M$ (i.e., $M$ restricted to this set is a perfect matching). 
\end{definition}

\begin{definition}\label{def:qkin}
Let $A$ be a set of cardinality $2k$. Then $q(k,i,n)$ is the probability that a random matching of size $\alpha n$ matches exactly $2i$ out of $2k$ points of $A$ and there are no edges of the matching connecting points from $A$ to points from $A^c$.
\end{definition}
\begin{definition}\label{def:qkibn}
Let $A$ be a set of cardinality $2k$. Then $q(k,i, b, n)$ is the probability that a random matching of size $\alpha n$ matches exactly $2i$ out of $2k$ points of $A$ and there are exactly $b$ edges of the matching connecting points from $A$ to points from $A^c$.
\end{definition}
\begin{definition}
For a set $A$ and a matching $M$ we say that an edge $e$ of $M$ is \emph{inner} if it connects points from $A$, \emph{boundary} if it connects a point from $A$ and a point from $A^c$, \emph{external} if it connects points from $A^c$.
\end{definition}

\newcommand{\lmpln}{For every integer $n$ and every $0\leq \ell\leq n/2$ 
\[
p(\ell,n)=\binom{\alpha n}{\ell}\binom{n}{2\ell}^{-1}. 
\] 
}
\begin{lemma}\label{lm:pln}
\lmpln
\end{lemma}

\newcommand{\lmqkin}{For every integer $n$ and every $0\leq i\leq k\leq n/2$
\[
q(k,i,n)=\binom{\alpha n}{i}\binom{n-2\alpha n}{2(k-i)}\binom{n}{2k}^{-1}.
\] 
}
\begin{lemma}\label{lm:qkin}
\lmqkin
\end{lemma}

\newcommand{\lmqkibn}{For all non-negative integers $n,i,k,b$ satisfying $k\leq n/2$ and $2i+b\leq 2k$
\[
q(k,i,b, n)=\binom{\alpha n}{i}\binom{\alpha n - i}{b}2^{b}\binom{n-2\alpha n}{2(k-i)-b}\binom{n}{2k}^{-1}.
\]
}
\begin{lemma}\label{lm:qkibn}
\lmqkibn
\end{lemma}

\newcommand{\lmqkibninequality}{Suppose we have $k<n/10$, $\alpha < 1/100$, and $2i+b\leq 2k$ then 
\[
q(k,i,b,n)\leq q(k,i,n)20^{-b}4^{k-i}.
\]
}

\begin{lemma}\label{lm:qkibninequality}
\lmqkibninequality
\end{lemma}

\section{The basic setup, $(C, s^*)$-bounded sets and their properties}\label{sec:defs}

In this section we introduce our basic setup for analyzing the \DIHP~problem, define the notion of $(C, s^*)$-bounded sets and introduce some of their basic properties. 

\subsection{The basic setup}
For a random variable $J_t$ we write $J_{1:t}$ to denote the tuple $(J_1,\ldots, J_t)$.
In this notation, the inputs to \IHP($n, \alpha, T)$ are denoted by $M_{1:T}$ 
and $w_{1:T}$.
Recall also that $S_t = r_t(M_{1:t}; S_{1:t-1}; w_t)$ denotes the message
posted by the $t$th player. We use $s$ to denote an upper bound on the size
of the messages.
The goal of our analysis is to show that if $s\ll n$ then the total variation distance between the distribution of messages $S_t$ and matchings $M_t$ posted on the blackboard at time $T$ in the \YES case and in the \NO case is small.

More specifically, let $S^Y_{1:T}$ denote the random variables corresponding to the messages posted by the players when the input $(M_{1:T},w_{1:T})$ is drawn from $\D^Y$, and let
$S^N_{1:T}$ denote the corresponding sequence when the input is drawn from
$\D^N$.
Our goal is to show that the total variation distance between
$(M_{1:T}, S^Y_{1:T})$ and $(M_{1:T}, S^N_{1:T})$, is vanishingly small. As we show in section~\ref{sec:main-result} 
it suffices to consider the \YES case only.
In Lemma~\ref{lm:main-tvd} we show that it suffices to
show that with high probability for each $t=1,\ldots, T$ the distribution of $M_t X^*$ 
is close to uniform in $\bool^{M_t}$. 
We now outline the techniques that we develop to prove this claim.

Our analysis relies on Fourier analytic techniques for reasoning about the distribution of $M_t X^*$.  
Conditioning on messages posted up to time $t$ makes $X^*$ 
uniformly random over a certain subset of the binary cube. We will analyze this subset of the hypercube, or, rather, the Fourier transform of its indicator function, and show that if communication is small, the distribution of $X^*$ conditional on typical history is such that $M_t X^*$ is close to uniformly random in total variation distance. 

We first define notation that lets us reason about the distribution of $X^*$. The knowledge that the players acquire about $X^*$ is represented by some set $\B_t\subseteq \bool^n$ with $X^*$
being distributed uniformly on $\B_t$. 
These sets $\B_t$ are constructed iteratively, by accumulating the information about $X^*$ that each player's message conveys - specifically there exist sets $\A_t$ (defined below) such that $\B_t=\A_1\cap \A_2\cap \ldots \cap \A_t$. Recall that we only consider the \YES case, so the state variables are superscripted accordingly.

\begin{definition}[Sets $\A_t, \B_t$ and their indicator functions $f_t, h_t$]\label{def:at}
Fix $\alpha\in (0, 1)$ and integers $n, T\geq 1$ and $t \in [T]$. 
Consider a \YES instance $(M_{1:T},w_{1:T})$ of \DIHP($n, \alpha, T$)
with $X^*$ being the (random) hidden partition (so $w_t = M_t X^*$).
Recall $S^Y_t = r_t(M_{1:t};S_{1:t-1},M_tX^*)$.
We define $\A_{reduced, t}\subseteq \bool^{M_t}$ be the set of possible values of $w_t=M_t X^*$ that lead to message $S^Y_t$ and $\A_t$ to denote the values of $X^* \in \bool^n$ that correspond to $\A_{reduced, t}$. Formally, letting $g_t(\cdot) :=r_t(M_{1:t}, S^Y_{1:t-1}, \cdot): \bool^{M_t} \to \bool^s$, we have
$$
\A_{reduced, t}=g_t^{-1}(S^Y_t)\subseteq \bool^{M_t}, 
$$
\begin{equation}\label{eq:at-def}
\mbox{and ~~} \A_t=\{x\in \bool^n: M_tx\in \A_{reduced, t}\}.
\end{equation}
Let $f_t:\bool^n\to \bool$ denote indicator functions of the sets $\A_t, t=1,\ldots, T$. For each $t$ let $h_t=f_1\cdot\ldots\cdot f_t$, so that $h_t$ is the indicator of $\B_t:=\A_1\cap \A_2\cap \ldots \cap \A_t$. We let $\B_0:=\bool^n$ for convenience. 
\end{definition}

\begin{claim}
For $t=1,\ldots, T$ let $p_t:\bool^n\to [0, 1]$ denote the following distribution:
$$
p_t(x)=\prob[M_tx\in \A_{t, reduced}].
$$
Further, let $p(x)=\prod_{t=1}^T p_t(x)$ for all $x\in \bool^n$.

The conditional distribution of $X^*$ given messages $\A_1,\ldots, \A_t$ is exactly given by $p(x)/||p||_1$ as above.
\end{claim}
\begin{proof}
Since $X^*$ is uniformly random in $\bool^n$ initially, we have 
\begin{equation*}
\begin{split}
&\prob[X^*=x| M_tX^*\in \A_{t, reduced}, t=1,\ldots, T]\\
&=\frac{\prob[X^*=x\wedge  M_tX^*\in \A_{t, reduced}, t=1,\ldots, T]}{\prob[M_tX^*\in \A_{t, reduced}, t=1,\ldots, T]}\\
&=\frac{\prob[M_tX^*\in \A_{t, reduced}, t=1,\ldots, T|X^*=x]\prob[X^*=x]}{\prob[M_tX^*\in \A_{t, reduced}, t=1,\ldots, T]}\\
&=\frac{\prob[X^*=x]}{\prob[M_tX^*\in \A_{t, reduced}, t=1,\ldots, T]}\cdot \prod_{t=1}^T \prob[M_tx\in \A_{t, reduced}, t=1,\ldots, T] \\
&=\frac{2^{-n}}{\prob[M_tX^*\in \A_{t, reduced}, t=1,\ldots, T]}\cdot \prod_{t=1}^T \prob[M_tx\in \A_{t, reduced}, t=1,\ldots, T] \\
&=p(x)/||p||_1\\
\end{split}
\end{equation*}

\end{proof}

\subsection{$(C, s^*)$-bounded sets and their properties}
The following definition is crucial for our analysis:

\begin{definition}[$(C, s^*)$-bounded set]\label{def:bounded}
Let $\B\subset \{0,1\}^n$ with indicator function $h$. We say that $\B$ (or $h$) is {\em $(C, s^*)$-bounded} if 
\begin{itemize}
\item For all $\ell\leq s^*$ we have
\[
\sum_{\substack{v\in\{0,1\}^n \\ |v|=2\ell}}\left|\wt{\mathbbm{1}_\B}(v)\right|\leq \left(\frac{C\sqrt{s^*n}}{\ell}\right)^\ell;
\]
\item For all $s^* < \ell < \frac{n}{C^2}$ we have
\[
\sum_{\substack{v\in\{0,1\}^n \\ |v|=2\ell}}\left|\wt{\mathbbm{1}_\B}(v)\right|\leq \left(\frac{C^2 n}{\ell}\right)^{\ell/2}.
\]
\end{itemize}
\end{definition}
Defining the function $\bound(\ell)$ by
\begin{equation}\label{eq:bound-def}
\bound(\ell)=
\begin{cases}
1 \qquad \qquad \qquad \ell=0; \\
\left(\frac{C\sqrt{s^* n}}{\ell}\right)^\ell \qquad \ell\in [1: s^*]; \\
\left(\frac{C^2 n}{\ell}\right)^{\ell/2} \qquad \ell>s^*,
\end{cases}
\end{equation}
we are able to simplify notation, ensuring that an indicator function $h$ is $(C, s^*)$ bounded if and only if for all $\ell<n/C^2$ we have 
\[
\sum_{\substack{v\in \{0,1\}^n \\ |v|=2\ell}} \left|\wt{h}(v)\right|\leq \bound(\ell).
\]

\begin{remark}
For intuition it is useful to compare our bounds throughout the paper to the bounds on the weight of Fourier coefficients of the simple adaptive component growing algorithm described in Section~\ref{sec:simple-protocol} when the latter is given a somewhat larger space budget, specifically $s':=\sqrt{s^* n}$ amount of space.
\end{remark}

The following lemma provides the base case of our analysis:
\begin{lemma}\label{lm:one-message-bounded}
For every $s^*$, every matching $M$ on $[n]$, every  $\A\subseteq  \{0, 1\}^M$ such that $|\A|/2^{|M|}\geq 2^{-s^*}$ the set $\B:=\{x\in \{0, 1\}^n: Mx\in \A\}$ is $(3, s^*)$-bounded. 
\end{lemma}
\begin{proof}
By Lemma \ref{lm:matching-fourier} Fourier coefficients of $\mathbbm{1}_{\B}$ can be written in terms of those of $\mathbbm{1}_{\A}$, namely, we have 
\[
\widehat{\mathbbm{1}_{\B}}(M^Tw)=\widehat{\mathbbm{1}_{\A}}(w),
\]
for every $w\in\bool^M$ and $\widehat{\mathbbm{1}_{\B}}(v)=0$ if $v$ is not of the form $M^Tw$. Also note that $|M^Tw|=2\cdot|w|$ and $|\B|=2^{n-\alpha n}\cdot |\A|$. The latter implies that we also have for every $w\in\bool^M$
\[
\wt{\mathbbm{1}_{\B}}(M^Tw)=\frac{2^n}{|\B|}\cdot \widehat{\mathbbm{1}_{\B}}(M^Tw)=
\frac{2^{\alpha n}}{|\A|}\cdot\widehat{\mathbbm{1}_{\A}}(w)=\wt{\mathbbm{1}_{\A}}(w).
\]
We then apply Lemma \ref{lm:L1xorKKL} to infer
\begin{equation}\label{eq:bound-from-KKL}
\sum_{\substack{v\in\{0,1\}^n \\ |v|=2\ell}}\left|\wt{\mathbbm{1}_\B}(v)\right| = \sum_{\substack{w\in\{0,1\}^M \\ |w|=\ell}}\left|\wt{\mathbbm{1}_\A}(w)\right|  \leq \sqrt{\binom{\alpha n}{\ell}\left(\frac{4 s^*}{\ell}\right)^\ell}\leq \left(\frac{3\sqrt{s^*n}}{\ell}\right)^\ell.
\end{equation}
For high weights ($\ell>s^*$) we use a trivial bound saying that sum of squares of normalized Fourier coefficients at any level is at most $2^{s^*}$. We also use the fact that by Lemma \ref{lm:matching-fourier} the function $\mathbbm{1}_\B$ has at most $\binom{\alpha n}{\ell}$ non-zero Fourier coefficients at level $2\ell$.
\[
\sum_{\substack{v\in\{0,1\}^n \\ |v|=2\ell}}\left|\wt{\mathbbm{1}_\B}(v)\right|\leq \sqrt{\binom{\alpha n}{\ell}2^{s^*}}\leq \left(\frac{2en}{\ell}\right)^{\ell/2}.
\]
\end{proof}
\begin{remark}[Intuition for the choice of $\bound(\ell)$]
Lemma \ref{lm:one-message-bounded} basically states that the set of possible $x\in\bool^n$ consistent with a message of one player is $(3,s^*)$ bounded. This Lemma is never explicitly used in the proof of Theorem \ref{thm:main-comm}, however, it provides a good intuition for why the definition of $(C,s)$ boundedness should be as it is. Indeed, the fact that in the lemma above Fourier coefficients of $\mathbbm{1}_{\B}$ at level $2\ell$ correspond to Fourier coefficients of $\mathbbm{1}_{\A}$ at level $\ell$ allows to write  
\begin{equation}\label{eq:L2-bound}
\sum_{\substack{v\in\{0,1\}^n \\ |v|=2\ell}}\wt{\mathbbm{1}_\B}^2(v)\leq \left(\frac{4 s^*}{\ell}\right)^\ell.
\end{equation}
In order to convert this into a bound on the $L_1$ norm we need to know how many non-zero summands we have on the left hand size. For one matching the correspondence with Fourier coefficients of $\mathbbm{1}_{A}$ readily implies that we have at most $\binom{\alpha n}{\ell}$ non-zero summands. When several messages of players are combined this is no longer true and the number of non-zero summands can get as large as $\binom{n}{2\ell}$, however, it turns out that $L_1$ norm behaves as if there were order $\binom{n}{\ell}$ non-zero summands.  
\end{remark}

\begin{lemma}\label{lm:pdfs-closeness}
For every $C>100$, $\delta\in (n^{-1}, 1/2)$ and $\alpha\in (0, 1/100)$ the following condition holds if $n$ is sufficiently large. Let $\B\subset \{0,1\}^n$ be $(C, s^*)$-bounded for $s^*\in [10\ln{(n+1)},\, \delta^4 n/C^2]$ as per Definition~\ref{def:bounded}, $|\B|/2^n\geq 2^{-s^*}$, and let $h:\{0,1\}^n\rightarrow\{0,1\}$ be the indicator of $\B$. Let $M$ be a uniformly random matching on $[n]$ of size $\alpha n$. Then with probability at least $1 - \delta$ over the choice of $M$ for any non-negative function $q$ on $\{0,1\}^M$ one has 
\[
1-\delta\leq \frac{\mathbb{E}_{x\sim \operatorname{Uniform}(\B)}\left[q(Mx)\right]}{\mathbbm{E}_{z\sim \operatorname{Uniform}\left(\{0,1\}^M\right)}\left[q(z)\right]}\leq 1+\delta.
\]
\end{lemma}
The proof is given in Appendix~\ref{app:A}.
\begin{rem}
We note that the proof of Lemma~\ref{lm:pdfs-closeness} does not use the bound on the spectrum of $h$ at levels above $s^*$ that follow from $(C, s^*)$-bounded property.
\end{rem}
\begin{rem}
Using this lemma we can now work with the uniform measure on $\{0,1\}^{M_t}$ instead of the one given by $M_th_{t-1}$. For instance, we may assume that each player receives such bits on $M_t$ that the corresponding function $f_t$ and the part of the cube $A_t$ satisfy $|A_t|>2^{n-s^*}$ (as the opposite happens with probability at most $2^{-s}$).
\end{rem}

\section{Technical overview of our analysis}\label{sec:outline-app}

In this section we define a simple protocol for \DIHP($n, \alpha, T$) and analyze its Fourier spectrum. We show that our protocol does not solve \DIHP($n, \alpha, T$) for constant $\alpha\in (0, 1)$ and $T\geq C_1/\alpha$ with a sufficiently large constant $C_1>1$, unless the communication budget is at least $n/A^T$ for some constant $A>1$. We also non-rigorously explain why this protocol solves the problem if the communication budget is at least $n/c^T$ for some constant $c>1$. Our main theorem (Theorem~\ref{thm:main-comm}) in particular implies that this protocol is close to optimal. More importantly, however, the proof of Theorem~\ref{thm:main-comm} can be viewed as analyzing an $\ell_1$-relaxation of the simple protocol. 

\subsection{A component growing protocol for \DIHP}\label{sec:simple-protocol}

In this section we consider very simple communication protocols, where players choose a subset of the edges of the matching that they receive, post the bits on that subset on the board, and ignore the bits on the other edges (note that the subset of edges should be no larger than the communication bound $s$).  We note that in order to achieve a constant advantage over random guessing for \DIHP~it suffices to ensure that the set of edges whose labels are posted on the board {\bf form a cycle}. Indeed, in the \YES case the sum of labels over the edges of any cycle is zero, while in the \NO case it is a uniformly random number in $\bool$. Thus, one can distinguish between the two cases with constant probability.

\paragraph{A very simple non-adaptive protocol.} The question is of course which subset of edges the players choose. A very basic approach would be to post bits on edges both of whose endpoints have indices between $1$ and $\sqrt{s/n}\cdot n$. Note that the expected number of such edges in a given matching is  $O(s)$, which is consistent with the communication budget per player. This protocol is non-adaptive in the sense that the edges that the $t$-th player posts the bits for are independent of the matchings of the other players. It is easy to see that this protocol will not find a cycle as long as $s\ll n/T^2$. Indeed, the graph induced on the first $\sqrt{s/n}\cdot n$ vertices is quite similar to an Erd\H{o}s-R\'{e}nyi graph with expected degree $\sqrt{s/n}\cdot T\ll 1$, and hence the graph will contain no cycle with high (constant) probability. Note that this behaviour seems to suggest that Theorem~\ref{thm:main} could be strengthened significantly, ensuring that the dependence of space on $T$, the number of matchings, is polynomial as opposed to exponential. However, it turns out that the simple {\em non-adaptive} protocol above is not a good model for the problem. We now introduce a more interesting, but still quite simple, protocol, which serves as a good model for our general analysis.

\paragraph{The component growing protocol.} Let $s$ be the per player communication budget. In order to show that even  {\bf  adaptive} protocols cannot solve \DIHP~ unless  $s\gg n/C^T$ for a constant $C>1$, for every such protocol $\Pi$ we define a strictly stronger protocol $\Pi'$ that has a larger communication budget than $\Pi$ and is more convenient to analyze. This new protocol $\Pi'$ can be thought of as having budget larger than $s$,  is still unable to solve the problem if $s\ll n/A^T$ for some large $A>1$. The protocol $\Pi'$ is defined as follows. The first player simply posts the bits according to the protocol $\Pi$. We let $F_1$ denote the forest created in this way, with edges labeled by numbers in $\bool$. For every $t=2,\ldots, T$, the $t$-th player posts the bits on edges of $M_t$ that have at least one endpoint in a connected component in $F_{t-1}$ (these bits should be thought of as free), as well as the bits on a set of $s$ edges of the matching $M_t$ that do not intersect any component in $F_{t-1}$ so that any edge revealed by $\Pi$ is also revealed by $\Pi'$. Let $F_t$ denote the forest obtained in this way.  If at least one of the edges of $M_t$ closes a cycle (i.e. edges of $F_{t-1}\cap M_t$ form a cycle), the $t$-th player computes the sum of bits on the cycle, and outputs \YES if that sum is zero, and \NO otherwise. Note that if the game is in the \YES case, the players always say \YES. If the game is in the \NO case, then if a cycle is found, the players say \NO with probability $1/2$. Thus, the players  obtain advantage $1/4$ over random guessing for \DIHP. We note that every player posts $s$ bits, and some number of bits (those that intersect existing components in $F_{t-1}$) are revealed for free, so the communication cost of this protocol is at least $s$ bits per player. We show below that this simple protocol does not find a cycle, and hence does not solve \DIHP($n, \alpha, T$) with high constant probability unless $s$ is close to $n$.

At the same time, we note that adaptive protocols are quite powerful: they can solve \DIHP($n, \alpha, T$) as long as their communication budget $s$ exceeds $n/c^T$ for some $c>1$. An example such protocol works as follows. The first player posts $s/2$ bits on arbitrary edges of the matching $M_1$. The subsequent players then maintain a collection of connected components $C_1,\ldots C_k$ formed by the edges whose bits are posted (we call such edges \emph{revealed}). At each step a new player reveals edges incident to at least one of the connected components. We then remove the smallest components from the list so that the total size of all components remains at most $s/2$. Note that this ensures that every player posts at most $s$ bits. After a new player reveals edges, most of the components will increase its size by a factor of $(1+\alpha)$ with high probability. Thus, at each step the average component size will be multiplied by $(1+\alpha)$ whereas the total size of all components will stay close to $s$. Note that deleting the smallest components cannot decrease the avarage component size. After $T-1$ steps we will get approximately $s/c^T$ connected components each of size $c^T$. One can see that if $s\cdot c^T \gg n$ then with high probability one of the edges of $M_t$ will have both endpoints in the same connected component and thus we will solve the problem.
 
  In the rest of this section we first show (Section~\ref{sec:combinatorial-analysis}) directly that even the powerful protocol $\Pi'$ (with free edges) defined above still cannot solve \DIHP($n, \alpha, T$) with constant $\alpha \in (0, 1)$ unless $s\gg n/C^{T}$ for some absolute constant $C>1$. This analysis is quite simple, but does not generalize to arbitrary protocols. We then illustrate our analysis of general protocols by instantiating the relevant parts of the analysis for the simple protocol above and proving that it does not solve \DIHP($n, \alpha, T$) unless $s\gg n/T^T$ (note the slightly weaker bound) by analyzing its Fourier spectrum in Section~\ref{sec:analysis-overview}.

\subsection{Combinatorial analysis of the component growing protocol}\label{sec:combinatorial-analysis}
Our analysis uses a potential function defined on the forest $F$ maintained by the protocol. The potential function is simply the sum of squares of sizes of (nontrivial) connected components of $F$. We refer to this quantity as the weight of $F$, or $||F||$:
\begin{definition}
For a forest $F$  with non-trivial(size more than $1$) connected components $C_1,\ldots,C_k$ we define its weight by
$\|F\|:=\sum_{i=1}^k|C_i|^2$.
\end{definition}

For $t=1,\ldots, T$ let $F_t$ denote the forest computed by players $1,2,\ldots, t$. Recall that in order to distinguish between the \YES and \NO cases using our simple protocol the players must ensure that at least one of the matchings $M_t$ contains an edge with both endpoints inside one of the connected components of $F_{t-1}$. Let $C_1,\ldots, C_k$ denote the collection of connected components in $F_{t-1}$, and note that for a uniformly random edge $e=(u, v)$ one has 
\begin{equation}\label{eq:one-edge-cycle}
\prob[\exists j=1,\ldots, k\text{~s.t.~}u, v\in C_j]\leq \sum_{j=1}^k |C_j|^2/n^2, 
\end{equation}
and hence, taking a union bound over all edges of $M_t$, we get
$$
\prob[\exists j=1,\ldots, k, e=(u, v)\in M_t\text{~s.t.~}u, v\in C_j]\leq |M_t|\cdot \sum_{j=1}^k |C_j|^2/n^2\leq \sum_{j=1}^k |C_j|^2/n=\|F_{t-1}\|/n.
$$
The latter expression does not take into account the fact that as edges are added one after another, the set of components may increase, but if one adds edges of the matching $M_t$, a similar bound follows. In particular, one can show that the players succeed with probability at most 
$$
\sum_{t=1}^T O(\|F_t\|/n)=O(T)\cdot (\|F_T\|/n),
$$
since $\|F_t\|$ is a non-decreasing function of $t$. It thus suffices to prove that $\| F_T\|\ll n/T$. We now outline a simple analysis that shows that if $F_t$ is the forest computed by the players in the simple protocol above, then as long as $s\ll n/A^T$ for a sufficiently large $A>1$, then after $t$ steps one has 
\begin{equation}\label{eq:193uwgweg}
\|F_t\|\leq s\cdot B^t
\end{equation}
for some constant $B>1$. Thus, if the players start with $s\ll n/A^T$ space with $A\geq 100B$, say, then one has $\|F_T\|/n\ll 100^{-T}$, and thus the players do not succeed distinguishing between the \YES and \NO cases with any significant advantage over random guessing.

\paragraph{Analyzing the growth of $\|F_t\|$.} Fix $t\in [T]$, and  let $C_1,C_2, \dots,C_k$ be the nontrivial(size more than $1$) connected components of $F_{t-1}$. Let $e_1,e_2,\ldots,e_{\alpha n}$ be edges of $M_t$. We now analyze the expected increase of the weight $\|F_t\|$ of $F_t$ relative to the weight $\|F_{t-1}\|$ of $F_{t-1}$.

In order to compare $\|F_t\|$ to $\|F_{t-1}\|$, we define $|M_t|$ intermediate forests, where the $j$-th such forest is the forest that results from adding edges from the set $e_1,\ldots, e_j$ to $F_{t-1}$. Specifically, let $F_t^{0}:=F_{t-1}$, and for every $j=1,\ldots, |M_t|$ let 
$$
F_t^j:=\left\lbrace
\begin{array}{ll}
F_t^{j-1}\cup \{e_j\}&\text{if~}e_j\text{~intersects with a nontrivial component in~}F_t^{j-1}\\
F_t^{j-1}&\text{o.w.}
\end{array}
\right.
$$
Recall that besides keeping edges of $M_t$ that intersect connected components of $F_{t-1}$, the $t$-th player also posts the bits on an arbitrarily chosen subset of edges of $M_t$ on the board. Let $E^*$ denote this set of additional edges (not incident on any nontrivial component in $F_{t-1}$ whose bits the $t$-th player posts). This means that $F_t\subseteq F_t^{|M_t|}\cup E^*$. Since edges in $E^*$ are not incident to any components in $F_t^{|M_t|}$, their addition simply increases the weight of the forest by $4s$ (since every component is of size $2$). We now upper bound $\|F_t^{|M_t|}\|$.

For every $i=1,\ldots, |M_t|$ we upper bound 
$$
\mathbb{E}_{e_i}\left[\|F_t^{i}\|\middle| F_t^{i-1}, e_{[1:i-1]}\right].
$$
Conditioned on edges $e_1,e_2,\ldots,e_{i-1}$, the edge $e_i$ is a uniformly random edge not sharing end-points with $e_1,e_2,\ldots,e_{i-1}$. Thus, for every subset of vertices $C\subseteq [n]$ one has

\begin{equation}\label{eq:probability-boundary}
\prob_{e_i}[|C\cap e_i|=1|e_{[1:i-1]}]\leq 4|C|/n,
\end{equation}
since $|M_t|=\alpha n$ and $\alpha<1/4$ by assumption. Note that if for a component $C$ in $F_t^{i-1}$ one has $|C\cap e_i|=1$ and the other endpoint of $e_i$ does not belong to any nontrivial component in $F_t^{i-1}$, then the increase in the weight of the forest is $(|C|+1)^2-|C|^2=2|C|+1$. We call such edges $e_i$ {\bf boundary edges}. 

Similarly, for every pair of  subset of vertices $C, C'\subseteq [n], C\cap C'=\emptyset$ one has 
\begin{equation}\label{eq:probability-internal}
\prob_{e_i}[|C\cap e_i|=1\text{~and~}|C'\cap e_i|=1|e_{[1:i-1]}]\leq 8|C||C'|/n^2,
\end{equation}

and we note that if the event above happens, the increase in the weight of the forest due to addition of $e_i$ is $(|C|+|C'|)^2-|C|^2-|C'|^2=2|C| |C'|$.  We call such edges $e_i$ {\bf internal edges}. These notions of internal edges and boundary edges will later be crucial in the proof of our main theorem through Fourier analysis (see Section~\ref{sec:analysis-overview} below for an outline and then proofs of Lemmas~\ref{lm:mass-transfer} and~\ref{lm:mass-transfer-high} for the actual application).

Putting the two observations above (increase of potential due to boundary edges and internal edges) together, we get
\begin{equation}\label{eq:weight-increase}
\begin{split}
\mathbb{E}_{e_i}\left[\|F_t^{i}\| \middle| e_{[1:i-1]}\right]&\leq \|F_t^{i-1}\|+\sum_{\substack{C\text{~nontrivial component}\\ \text{in~} F_t^{i-1}}} (2|C|+1)\cdot \prob_{e_i}[|C\cap e_i|=1|e_{[1:i-1]}]\\
&+\sum_{\substack{C, C'\text{~nontrivial components}\\ \text{in~} F_t^{i-1}}} (2|C|\cdot |C'|)\cdot \prob_{e_i}[|C\cap e_i|=1\text{~and~}|C'\cap e_i|=1|e_{[1:i-1]}]\\
\end{split}
\end{equation}

Substituting ~\eqref{eq:probability-boundary} and~\eqref{eq:probability-internal} into~\eqref{eq:weight-increase}, we obtain
\begin{equation}\label{eq:combi-one-step-evolution}
\begin{split}
\mathbb{E}_{e_i}\left[\|F_t^{i}\| \middle|e_{[1:i-1]}\right]\leq \|F_t^{i-1}\|&+\sum_{\substack{C\text{~nontrivial component}\\ \text{in~} F_t^{i-1}}} (2|C|+1)\cdot \frac{4|C|}{n}\\
&+\sum_{\substack{C, C'\text{~nontrivial components}\\ \text{in~} F_t^{i-1}}} (2|C|\cdot |C'|)\cdot \frac{8|C|\cdot |C'|}{n^2}\\
&\leq \|F_t^{i-1}\|+\frac{12\|F_t^{i-1}\|}{n}+\frac{8\|F_t^{i-1}\|^2}{n^2}\\
&=\|F_t^{i-1}\| \cdot \left(1+\frac{12}{n}+\frac{8\|F_t^{i-1}\|}{n^2}\right).
\end{split}
\end{equation}
Applying this $|M_t|=\alpha n$ times formally requires a careful application of concentration inequalities (the details are provided in Appendix~\ref{sec:app-C}), but ultimately results in 
\begin{equation}\label{eq:3409hg943g}
\begin{split}
\mathbb{E}\left[\|F_t^{|M_t|}\|\right]&\approx \|F_{t-1}\|\cdot \left(1+\frac{12}{n}+\frac{8\|F_{t-1}\|}{n^2}\right)^{|M_t|} \approx \|F_{t-1}\|\cdot (B/2)
\end{split}
\end{equation}
for some constant $B$, as long as all intermediate forests satisfy $\|F_t^{i-1}\|\ll n$ (which they do with the appropriate setting of parameters -- see Appendix~\ref{sec:app-C} for details).

Now recall that besides keeping edges of $M_t$ that intersect connected components of $F_{t-1}$, the $t$-th player also posts the bits on an arbitrarily chosen subset of edges of $M_t$ that do not share an endpoint with $F_{t-1}$ on the board (we call such edges {\bf external edges}; a similar notion plays a crucial role in our analysis of general protocols see~\eqref{eq:bound-external} and related discussion). Recall that $F_t\subseteq F_t^{|M_t|}\cup E^*$, where $E^*$ denotes the set of additional edges (not incident on any nontrivial component in $F_t$) that the $t$-th player reveals. We thus get $\|F_t\|\leq \|F_t^{|M_t|}\|+4s$, which, when put together with~\eqref{eq:3409hg943g}, gives
\begin{equation*}
\mathbb{E}\left[\|F_t\|\right]\approx \|F_{t-1}\|\cdot B+4s.
\end{equation*}
Applying the above iteratively for $t=1,\ldots, T$ results in 
$$
\mathbb{E}[\|F_t\|]\lesssim \|F_0\|\cdot B^t\lesssim s \cdot B^t.
$$ 
This (informally) establishes~\eqref{eq:193uwgweg}, and shows that the players need $s\gg n/B^t$ in order to solve \DIHP($n, \alpha, T$).

The analysis outlined above is quite simple, but does not generalize to arbitrary communication protocols. In the next section we introduce our Fourier analytic approach, and illustrate some of the main claims by instantiating them on our component growing protocol above.

\subsection{Overview of general analysis (proof of Theorem~\ref{thm:main-comm})}\label{sec:analysis-overview}
In this section we first make some remarks about the Fourier spectrum of our component growing protocol, and then use it to illustrate our analysis, which is formally presented in Section~\ref{sec:main-result}.

\subsubsection{Second level Fourier spectrum of the component growing protocol vs combinatorial analysis} 
Suppose that the players use the simple protocol as described above. In that case the set of possible values of the hidden partition $X^*$ consistent with the players' knowledge at time $t$ can be defined quite easily:
\begin{equation}\label{eq:bt-comp-growing}
\B_t=\{x\in \bool^n: \forall e=(a, b)\in F_t, x_a+x_b=w_e\}.
\end{equation}

Thus, $\B_t$ is simply a linear subspace of $\bool^n$ with constraints given by the edges in $F_t$.  We will derive expressions for the Fourier transform of the indicators $h_t:=\mathbf{1}_{\B_t}$ of $\B_t$ for this simple protocol. Similarly, set of possible values of the hidden partition $X^*$ consistent with the $t$-th player's message can be defined quite easily as well:
\begin{equation}\label{eq:at-comp-growing}
\A_t=\{x\in \bool^n: \forall e=(a, b)\in F_t\cap M_t, x_a+x_b=w_e\}.
\end{equation}

Again, $\A_t$ is simply a linear subspace of $\bool^n$ with constraints given by the edges revealed by the $t$-th player. The normalized Fourier transform of the indicator $f_t:=\mathbf{1}_{\A_t}$ of $\A_t$ for this simple protocol is quite simple: 
\begin{equation}\label{eq:wtfw}
\wt{f}(z)=\left\lbrace
\begin{array}{ll}
(-1)^{\sum w_i}&\text{~if~$z$ is matched by edges $\{e_i\}$ of $M_t$}\\
0&\text{o.w.}
\end{array}
\right.
\end{equation}

These expressions will not be directly useful for our proof, but will provide very good intuition. 

It can be verified (a calculation is included in appendix~\ref{sec:htft-app} for convenience of the reader) that $\wt{h}_t(v)$ can only be nonzero if the set $v$ has an even intersection with every cluster in $F_t$. We refer to such $v\in \bool^n$ as {\em admissible} for brevity. Furthermore, it can be verified that $\wt{h}_t(v)$ has the following simple form. For each admissible $v$ let $Q(v)$ denote a pairing of vertices of $v$ via edge-disjoint paths in $F_t$ (we associate $Q_v$ with the set of edges on these paths). This is illustrated in Fig.~\ref{fig:clusters-ft}(a), where the vertices of $v\in \bool^n$ are marked red, and the edges of $Q(v)$ are the green dashed edges. Then we have 
$$
\wt{h}_t(v)=\left\lbrace
\begin{array}{ll}
(-1)^{\sum_{e\in Q(v)} w_e}&\text{if~}v\text{~is admissible}\\
0&\text{o.w.}
\end{array}
\right.
$$
We refer to appendix~\ref{sec:htft-app} for the proof. For example, the coefficient $\wt{h}_t(\{a_1, a_2, b_1, b_2, c_1, c_2\})$ is nonzero, equals $1$ in absolute value, and its sign is determined by the parities of labels on the green paths connecting $a_1$ to $a_2$, $b_1$ to $b_2$, and $c_1$ to $c_2$ (see Fig.~\ref{fig:clusters-ft}(a)).

\begin{figure}[h]
      \begin{minipage}{0.5\textwidth}
        \subfigure[Illustration of admissibility property: coefficient $\{a_1, a_2, b_1, b_2, c_1, c_2\}$ (marked in red) is admissible, coefficient $\{a_1, c_2\}$ is not.]{
\tikzstyle{vertex}=[circle, fill=black!100, minimum size=5,inner sep=1pt]
\tikzstyle{svertex}=[circle, fill=red!100, minimum size=7,inner sep=1pt]
\tikzstyle{evertex}=[circle,draw=none, minimum size=25pt,inner sep=1pt]
\tikzstyle{edge} = [draw,-, color=red!100, very  thick]
\tikzstyle{bedge} = [draw,-, color=green!100, very  thick]
\begin{tikzpicture}[scale=0.55, auto,swap]

    \node[svertex](a0) at (-3, 0.5) [label=below:$a_1$]{};    
    \node[vertex](a1) at (-2, 1) {};
    \node[vertex](a2) at (-3.5, 2) {};    
    \node[vertex](a3) at (-2.5, 2) {};        
    \node[svertex](a4) at (-0.5, 1) [label=above:$a_2$]{};
    \node[vertex](a5) at (-1, 0.5) {};

    \node[svertex](b1) at (2, -0.5) [label=above:$b_1$]{};
    \node[vertex](b2) at (4, -1.0) {};    
    \node[svertex](b3) at (3.9, -1.7) [label=below:$b_2$]{};        
    \node[vertex](b4) at (2, -2.0) {};

    \node[vertex](c1) at (5+1, 0) {};
    \node[vertex](c2) at (7+1, 0) {};    
    \node[vertex](c3) at (6+1, -2) {};        
    \node[svertex](c4) at (4.5+1, -2) [label=below:$c_1$]{};
    \node[vertex](c5) at (5.5+1, +1) {};
    \node[vertex](c6) at (5.5+2+1.5, +1) {};    
    \node[svertex](c7) at (6+1+1, -2.5) [label=below:$c_2$]{};

    \node[vertex](c5) at (5.5+1, +1) {};
    \node[vertex](c6) at (5.5+2+1.5, +1) {};
    
    \node[vertex](d1) at (-2, -1) {};
    \node[vertex](d2) at (0, -2) {};

    

    \path[draw, line width=3pt, -, dashed, green2!100] (a0) -- (a2);    
    \path[draw, line width=1pt, -] (a1) -- (a2);
    \path[draw, line width=3pt, -, dashed, green2!100]  (a2) -- (a3);    
    \path[draw, line width=3pt, -, dashed, green2!100] (a3) -- (a4);
    \path[draw, line width=1pt, -]  (a4) -- (a5);


    \path[draw, line width=3pt, -, dashed, green2!100] (b1) -- (b2);
    \path[draw, line width=3pt, -, dashed, green2!100] (b2) -- (b3);
    \path[draw, line width=1pt, -]  (b3) -- (b4);


   \path[draw, line width=1pt, -] (c1) -- (c2);
   \path[draw, line width=1pt, -]  (c4) -- (c1);    


   \path[draw, line width=1pt, -] (c2) -- (c5);
   \path[draw, line width=1pt, -] (c2) -- (c6);

   \path[draw, line width=3pt, -, dashed, green2!100] (c3) -- (c7);
   \path[draw, line width=3pt, -, dashed, green2!100] (c4) -- (c3);

   \path[draw, line width=1pt, -] (d1) -- (d2);
    
\end{tikzpicture}
}        
\end{minipage}\hspace{0.1in}
      \begin{minipage}{0.5\textwidth}
        \subfigure[Growth of connected components of $F_{t-1}$ into $F_t$ after addition of a matching $M_t$ (its edges are shown as zig-zagged).]{
\tikzstyle{vertex}=[circle, fill=black!100, minimum size=5,inner sep=1pt]
\tikzstyle{svertex}=[circle, fill=red!100, minimum size=7,inner sep=1pt]
\tikzstyle{evertex}=[circle,draw=none, minimum size=25pt,inner sep=1pt]
\tikzstyle{edge} = [draw,-, color=red!100, very  thick]
\tikzstyle{bedge} = [draw,-, color=green!100, very  thick]
\begin{tikzpicture}[scale=0.55, auto,swap]

    \node[svertex](a0) at (-3, 0.5) [label=below:$a_1$]{};    
    \node[vertex](a1) at (-2, 1) {};
    \node[vertex](a2) at (-3.5, 2) {};    
    \node[vertex](a3) at (-2.5, 2) {};        
    \node[svertex](a4) at (-0.5, 1) [label=above:$a_2$]{};
    \node[vertex](a5) at (-1, 0.5) {};

    \node[svertex](b1) at (2, -0.5) [label=above:$b_1$]{};
    \node[vertex](b2) at (4, -1.0) {};    
    \node[svertex](b3) at (3.9, -1.7) [label=below:$b_2$]{};        
    \node[vertex](b4) at (2, -2.0) {};

    \node[vertex](c1) at (5+1, 0) {};
    \node[vertex](c2) at (7+1, 0) {};    
    \node[vertex](c3) at (6+1, -2) {};        
    \node[svertex](c4) at (4.5+1, -2) [label=below:$c_1$]{};
    \node[vertex](c5) at (5.5+1, +1) {};
    \node[vertex](c6) at (5.5+2+1.5, +1) {};    
    \node[svertex](c7) at (6+1+1, -2.5) [label=below:$c_2$]{};

    \node[vertex](c5) at (5.5+1, +1) {};
    \node[vertex](c6) at (5.5+2+1.5, +1) {};
    
    \node[vertex](d1) at (-2, -1) {};
    \node[vertex](d2) at (0, -2) {};    
    
   \draw [-,
line join=round,
decorate, decoration={
    zigzag,
    segment length=5,
    amplitude=2.0,post=lineto,
    post length=1pt
}]   (a4) -- (b1);

   \draw [-,
line join=round,
decorate, decoration={
    zigzag,
    segment length=5,
    amplitude=2.0,post=lineto,
    post length=1pt
}]   (b3) -- (c4);    
    
    

    \path[draw, line width=3pt, -, dashed, green2!100] (a0) -- (a2);    
    \path[draw, line width=1pt, -] (a1) -- (a2);
    \path[draw, line width=3pt, -, dashed, green2!100]  (a2) -- (a3);    
    \path[draw, line width=3pt, -, dashed, green2!100] (a3) -- (a4);
    \path[draw, line width=1pt, -]  (a4) -- (a5);


    \path[draw, line width=3pt, -, dashed, green2!100] (b1) -- (b2);
    \path[draw, line width=3pt, -, dashed, green2!100] (b2) -- (b3);
    \path[draw, line width=1pt, -]  (b3) -- (b4);


   \path[draw, line width=1pt, -] (c1) -- (c2);
   \path[draw, line width=1pt, -]  (c4) -- (c1);    


   \path[draw, line width=1pt, -] (c2) -- (c5);
   \path[draw, line width=1pt, -] (c2) -- (c6);

   \path[draw, line width=3pt, -, dashed, green2!100] (c3) -- (c7);
   \path[draw, line width=3pt, -, dashed, green2!100] (c4) -- (c3);

   \path[draw, line width=1pt, -] (d1) -- (d2);
    
\end{tikzpicture}
        }
      \end{minipage}
      \caption{Fourier transform of $h_t$ vs growth of connected components in $F_t$.}
    \label{fig:clusters-ft}

\end{figure}
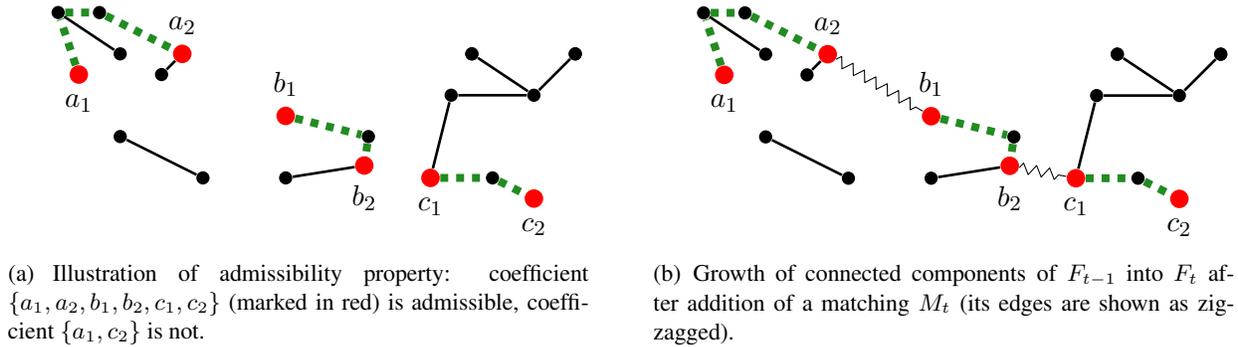

Recall that our direct analysis of the adaptive protocol in Section~\ref{sec:simple-protocol} used the {\em weight} of the forest $F_t$, defined by 
$$
\| F_t\|=\sum_{i=1}^k |C_i|^2,
$$
where $C_1, \ldots, C_k$ are the nontrivial (size strictly larger than $1$) connected components of $F_t$. As we noted in that section, this particular way of analyzing the component growing protocol does not generalize, but the following reformulation does. Consider weight two Fourier coefficients of $h_t$. As noted above, one has $\wt{h}_t(\{a, b\})\neq 0$  if and only if $a$ and $b$ belong to the same connected component in $F_t$. We therefore have 
\begin{equation*}
\begin{split}
\sum_{v\in \bool^n: |v|=2} \left|\wt{h}_t(v)\right|&=\frac1{2}\sum_{\substack{a, b \in [n],\\ a\neq b}} \left|\wt{h}_t(\{a, b\})\right|\\
&=\frac1{2}\sum_{i=1}^k \sum_{\substack{a, b \in C_i,\\ a\neq b}} \left|\wt{h}_t(\{a, b\})\right|\\
&=\sum_{i=1}^k {C_i \choose 2}\\
&=\Theta(\| F_t\|).
\end{split}
\end{equation*}

This suggests an analysis that is based on proving that 
\begin{equation}\label{eq:level-two-sum}
\sum_{v\in \bool^n: |v|=2} \left|\wt{h}_t(v)\right|\ll n
\end{equation}
for all $t=1,\ldots, T$, i.e, the sum of absolute values of second level Fourier coefficients stays small throughout the game, hoping that this is a general enough approach for handling arbitrary protocols. Our proof indeed proceeds along similar lines but there are two major difficulties that one needs to overcome in order to make this work. {\bf First}, while for the simple protocol higher order Fourier coefficients are (essentially) determined by weight two Fourier coefficients, i.e., by the collection of connected components in $F_t$, this is not the case for general protocols. We thus need to generalize~\eqref{eq:level-two-sum} to higher weights. {\bf Second}, we need to design techniques for analyzing the equivalent of `component growth' in the combinatorial version of our analysis, in terms of Fourier coefficients. 

\subsubsection{Evolution of Fourier coefficients}\label{sec:evolution}

We now present our general analysis, and illustrate it by applying to the component growing protocol from Section~\ref{sec:simple-protocol}. Recall that after $t$ players have spoken the random bipartition $X^*$ is uniform in the set $\B_t=\A_1\cap \A_2\cap \ldots \cap \A_t$.  The following lemma (see Section~\ref{sec:main-result} for the proof) is crucial for our proof: 

\newcommand{\lemmaMainComm}{
For every constant $\alpha\in (0, 10^{-10})$, integer $n\geq 1$, every $T\in [10, \ln n]$, $s\geq \sqrt{n}$, the following conditions hold if $n$ is sufficiently large. Suppose there exists a protocol $\Pi$ for \DIHP$(n,\alpha,T)$ such that $|\Pi|=s<n/(10T)^{10^9 T}$. Let $\delta=1/(1000T)$. Then there exist events  $\E_1\supseteq \E_2\supseteq \ldots\supseteq \E_T$ with $\prob[\E_1]=1$ and $\prob[\bar \E_{t+1}|\E_t]\leq 1/(100T)$ for $t=1\dots T-1$, such that $\E_t$ depends on $(S^Y_{1:t-1}, M_{1:t})$ only, and conditioned on $\E_t$ one has
\begin{enumerate}
\item  $\B_{t-1}$ is $((10^{12} T)^t, 10Ts)$-bounded (as per Definition \ref{def:bounded});
\item  $|\B_{t-1}|/2^n\geq 2^{-s(t-1)-10(t-1)\log{T}}$;
\item for any non-negative function $q$ on $\{0,1\}^{M_t}$
\begin{equation}\label{eq:pointwise}
1-\delta\leq \frac{\mathbb{E}_{x\sim \operatorname{Uniform}(\B_{t-1})}\left[q(M_t x)\right]}{\mathbbm{E}_{z\sim \operatorname{Uniform}\left(\{0,1\}^{M_t}\right)}\left[q(z)\right]}\leq 1+\delta.
\end{equation}
\end{enumerate}
}

\noindent{{\bf Lemma~\ref{lm:main-comm}} \em 
For every constant $\alpha\in (0, 10^{-10})$, integer $n\geq 1$, every $T\in [10, \ln n]$, $s\geq \sqrt{n}$, the following conditions hold if $n$ is sufficiently large. Suppose there exists a protocol $\Pi$ for \DIHP$(n,\alpha,T)$ such that $|\Pi|=s<n/(10T)^{10^9 T}$. Let $\delta=1/(1000T)$. Then there exist events  $\E_1\supseteq \E_2\supseteq \ldots\supseteq \E_T$ with $\prob[\E_1]=1$ and $\prob[\bar \E_{t+1}|\E_t]\leq 1/(100T)$ for $t=1\dots T-1$, such that $\E_t$ depends on $(S^Y_{1:t-1}, M_{1:t})$ only, and conditioned on $\E_t$ one has
\begin{enumerate}
\item  $\B_{t-1}$ is $((10^{12} T)^t, 10Ts)$-bounded (as per Definition \ref{def:bounded});
\item  $|\B_{t-1}|/2^n\geq 2^{-s(t-1)-10(t-1)\log{T}}$;
\item for any non-negative function $q$ on $\{0,1\}^{M_t}$
\begin{equation*}
1-\delta\leq \frac{\mathbb{E}_{x\sim \operatorname{Uniform}(\B_{t-1})}\left[q(M_t x)\right]}{\mathbbm{E}_{z\sim \operatorname{Uniform}\left(\{0,1\}^{M_t}\right)}\left[q(z)\right]}\leq 1+\delta.
\end{equation*}
\end{enumerate}
}

The third part of the lemma helps us conclude that the messages posted on the board do not reveal enough information to distinguish between the \YES and the \NO cases: indeed it implies that the posterior distribution of the labels that the $t$-th player observes on $M_t$ in the \YES case is pointwise close to uniform. Since we show that this is true for all players $t=1,\ldots, T$, the result follows by simple properties of the total variation distance (see Lemma~\ref{lm:main-tvd} and then proof of Theorem~\ref{thm:main-comm} in Section~\ref{sec:main-result}). The first two parts of Lemma~\ref{lm:main-comm} are the main invariants on the evolution of the Fourier spectrum of $h_t$ that drive our analysis.

Recall that by Definition~\ref{def:bounded}  a set $\B\subseteq \bool^n$ is $(C, s^*)$-bounded if for all $\ell\leq s^*$ we have
\begin{equation}\label{eq:low-weights-bound-recalled}
\sum_{\substack{v\in\{0,1\}^n \\ |v|=2\ell}}\left|\widetilde{\mathbbm{1}_\B}(v)\right|\leq \left(\frac{C\sqrt{s^*n}}{\ell}\right)^\ell;
\end{equation}
and for all $s^* < \ell < \frac{n}{C^2}$ we have
\begin{equation}\label{eq:middle-weights-bound-recalled}
\sum_{\substack{v\in\{0,1\}^n \\ |v|=2\ell}}\left|\widetilde{\mathbbm{1}_\B}(v)\right|\leq \left(\frac{C^2 n}{\ell}\right)^{\ell/2}.
\end{equation}

As per~\eqref{eq:bound-def}, defining the function $\bound(\ell)$ by
\begin{equation*}
\bound(\ell)=
\begin{cases}
1 \qquad \qquad \qquad \ell=0; \\
\left(\frac{C\sqrt{s^* n}}{\ell}\right)^\ell \qquad \ell\in [1: s^*]; \\
\left(\frac{C^2 n}{\ell}\right)^{\ell/2} \qquad \ell>s^*,
\end{cases}
\end{equation*}
we are able to simplify notation, ensuring that an indicator function $h$ is $(C, s^*)$ bounded if and only if for all $\ell<n/C^2$ we have 
\[
\sum_{\substack{v\in \{0,1\}^n \\ |v|=2\ell}} \left|\wt{h}(v)\right|\leq \bound(\ell).
\]

\paragraph{Intuition behind the choice of the bound $\bound(\ell)$.} We note that the bound above is essentially obtained as follows: one first thinks of the bounds on the $\ell_1$ norm of the Fourier transform of the indicator $f$ of the set $\A$ in the cube that is consistent with the message of a single player that follow by applying Cauchy-Schwarz to the $\ell_2$ norm bounds provided by hypercontractivity (see Lemma~\ref{lm:hypercontractivity-fourier} in Appendix~\ref{app-F}; this is the bound that was used in~\cite{GKKRW07} and follow up works). We then prove by induction that even the Fourier transform of the product of such indicator functions maintains the small $\ell_1$ norm property. The fact that $\bound(\ell)$ provides different bounds for small and large $\ell$ is a consequence of the fact that hypercontractivity only implies strong bounds on the $\ell_2^2$ mass of the Fourier spectrum of an (indicator of a) set $\A$ of density $2^{-s^*}$ when $\ell\leq s^*$, for larger $\ell$ one simply uses Parsevals' equality. See the proof of Lemma~\ref{lm:L1xorKKL} for the details of this calculation.

We also note that the bound that we prove on the $\ell_1$ norm of the Fourier transform of $h$ is indeed surprisingly strong: it is much stronger than what follows by Cauchy-Schwarz for an arbitrary function of the same $\ell_2^2$ norm.

Overall, in our analysis we distinguish between `low weight' Fourier coefficients, namely those with weights between $1$ and $s^*$, the `intermediate weight' Fourier coefficients, namely those between $s^*$ and $\frac{n}{C^2}$, and the `high weight' Fourier coefficients, namely those with weights between $\frac{n}{C^2}$ and $n$. Here, $s^*=10Ts$ is an upper bound for the total amount of information revealed by all players.  

\paragraph{The low weight bound~\eqref{eq:low-weights-bound-recalled} for the component growing protocol ($\ell\in [1, s^*]$).} Note that instantiating the first guarantee of Lemma~\ref{lm:main-comm} above for $\ell=1$ leads to a bound of 
\begin{equation*}
\sum_{\substack{v\in\{0,1\}^n \\ |v|=2}}\left|\widetilde{\mathbbm{1}_\B}(v)\right|\leq (10^{12} T)^t\cdot \sqrt{10Ts \cdot n},
\end{equation*}
which is similar to our upper bound of $s\cdot B^t$ on $\| F_t \|$ from Section~\ref{sec:combinatorial-analysis}. Also note that the bound that we get for general communication protocols with a budget of $s$ bits is similar to what one would get for the simple protocol with $\approx \sqrt{s\cdot n}$ bits (see rhs above). This is a consequence of the fact that our analysis starts with $\ell_2^2$ bounds on Fourier coefficients and converts those into $\ell_1$ bounds, with an appropriate loss from Cauchy-Schwarz.

\paragraph{The intermediate weight bound~\eqref{eq:middle-weights-bound-recalled} for the component growing protocol ($\ell\in (s^*, n/C^2]$).} Note that since the middle weight coefficients correspond to weights at least $s^*=10Ts$, they all vanish for any simple protocol of size at most $s$. Indeed, since each player reveals at most $s$ edges, the total size of non-trivial connected components in the resulting forest $F_T$ does not exceed $2Ts$. 

\paragraph{Hight weights for the component growing protocol ($\ell>n/C^2$).} Similarly, the high weight part of the spectrum is zero for the component growing protocol, since the maximum weight of a nonzero Fourier coefficient is upper bounded by twice the number of edges in the forest $F_t$, and that number never becomes close to $n$ with the appropriate setting of parameters.

We now outline the proof of Lemma~\ref{lm:main-comm}. Note that the main challenge that we had to overcome in the analysis of the component growing protocol in Section~\ref{sec:combinatorial-analysis} is bounding the rate at which connected components of different sizes are merged when edges of the next matching $M_t$ connect two nontrivial components (we refer to these edges as {\bf internal edges}, see Section~\ref{sec:combinatorial-analysis}) or connect a nontrivial component to an isolated vertex (we refer to these edges as {\bf boundary edges}, see Section~\ref{sec:combinatorial-analysis}). Our Fourier analytic approach analyzes the component merging process using the convolution theorem: we note that the arrival of internal or boundary edges results in the Fourier transform of $h_t$ being convolved with the Fourier transform of the message $f_t$ that the $t$-th player sends, and we analyze this process directly.  

The proof of Lemma~\ref{lm:main-comm} is by induction on $t$, with the inductive step being the main technical lemma of our paper. It is given by

\newcommand{\lemmaInduction}{For every $n, C, s^*, \alpha, \delta$ that satisfy conditions 
\begin{align*}
{\bf (P1)} \, \alpha < 10^{-10}\quad 
{\bf (P2)} \, C>10^6 \quad 	
{\bf (P3)} \, s^*<\frac{n}{10^9 C^3} \quad
{\bf (P4)} \, n>10^9C^4 \quad
{\bf (P5)} \, \delta \in (n^{-1}, 1/2),
\end{align*}
every $\B\subseteq \{0, 1\}^n$, if $\B$ is $(C, s^*)$-bounded (as per Definition \ref{def:bounded}) and $M$ is a uniformly random matching of size $\alpha n$,  the following conditions hold with probability at least $1-5\delta$ over the choice of $M$.

For every $\A_{reduced}\subseteq \{0, 1\}^M$ such that $|\A_{reduced}|/2^{\alpha n}\geq 2^{-s^*}$, if $\A=\{x\in \{0, 1\}^n: Mx\in \A_{reduced}\}$, then 
$\B':=\B\cap \A$ is $((10^9/\delta)C, s^*)$-bounded.

}
\noindent{{\bf Lemma ~\ref{lm:induction}} \em 
\lemmaInduction
}

We now illustrate the proof using the component growing protocol. Let $\B_t$ be as in Definition~\ref{def:at}. In particular, for the component growing protocol $\B_t$ is explicitly given by~\eqref{eq:bt-comp-growing}. For simplicity we write $M$ for $M_t$, $\B'$ for $\B_t$, $\B$ for $\B_{t-1}$, and $\A$ for $\A_t$. We further let $h:=\mathbf{1}_{\B}, f:=\mathbf{1}_\A, h':=h\cdot f=\mathbf{1}_{\B'}$. Proving that $h'$ is $((10^9/\delta)C, s^*)$-bounded involves upper bounding the $\ell_1$ norm of Fourier coefficients at various levels $2\ell, \ell\in [1, n/(2C^2)]$. Convolution theorem~\eqref{eq:convolution} together with triangle inequality give
\begin{equation}\label{eq:conv-hats-sketch}
\begin{split}
\sum_{\substack{v\in\{0,1\}^n\\ |v|=2\ell}} \left|\wt{h}'(v)\right| &=\frac{2^n}{|\B'|}\sum_{\substack{v\in\{0,1\}^n\\ |v|=2\ell}} \left|\wh{h}'(v)\right|\\
&=\frac{2^n}{|\B|}\frac{2^n}{|\A|}\sum_{\substack{v\in\{0,1\}^n\\ |v|=2\ell}} \left|\wh{h}'(v)\right|\\
&\leq \sum_{\substack{v\in\{0,1\}^n\\ |v|=2\ell}} \sum_{z\in\{0,1\}^n} \left|\frac{2^n}{|\B|}\wh{h}(v\oplus z) \frac{2^n}{|\A|}\wh{f}(z)\right|\\
&=\sum_{v\in\{0,1\}^n} \sum_{\substack{z\in\{0,1\}^n\\ |v\oplus w|=2\ell}} \left|\wt{h}(v)\wt{f}(z)\right| \\
&= \sum_{v\in\{0,1\}^n} \left|\wt{h}(v)\right| \sum_{\substack{z\in\{0,1\}^n\\ |v\oplus z|=2\ell}} \left|\wt{f}(z)\right|.
\end{split}
\end{equation}
Here in going from line~2 to line~3 we used the fact that $|\B'|/2^n=(|\B|/2^n)\cdot (|\A|/2^n)$ for our simple protocol, as long as no cycle has been revealed. This in particular proves part 2 of Lemma~\ref{lm:main-comm} for the component growing protocol without the additive loss in the exponent (we get $|\B_{t-1}|/2^n=2^{-s(t-1)}$ as opposed to just $|\B_{t-1}|/2^n\geq 2^{-s(t-1)-10(t-1)\log{T}}$; see proof of Lemma~\ref{lm:main-comm} in Section~\ref{sec:main-result} for general argument). It is useful to recall at this point (see~\eqref{eq:at-comp-growing} and~\eqref{eq:wtfw}) that $\wt{f}(w)\neq 0$ only if $w$ is perfectly matched by $M$.

Now note that for the component growing protocol the sum  $\sum_{\substack{v\in\{0,1\}^n\\ |v|=2\ell}} \left|\wt{h}'(v)\right|$ is the number of sets of $2\ell$ vertices that intersect every connected component in $F_{t-1}$ an even number of times. The sum on the rhs of the last line in~\eqref{eq:conv-hats-sketch} is over all $v$ that intersect every component in $F_{t-1}$ an even number of times, and subsets $z$ of edges of $M_t$ such that $v\oplus z$ has weight $2\ell$. {\bf For small $\ell$} most of the sum is contributed by sets of $2\ell$ vertices having two vertices in each connected component of $F_{t-1}$. In particular, every such $v$ together with $z\in\{0,1\}^n$ that satisfies $|v\oplus z|=2\ell$ and $\wt{f}(z)\neq 0$  corresponds to a collection of components in $F_{t-1}$ that are merged into a collection of $\ell$ components in $F_t$ by the edges in $z$.

The core of our proof (see Lemma \ref{lm:mass-transfer} and Lemma \ref{lm:mass-transfer-high}) shows that if $h$ is $(C, s^*)$-bounded, then
\[
\sum_{v\in\{0,1\}^n} \left|\widetilde{h}(v)\right| \mathbb{E}_M\left[\sum_{\substack{w\in\{0,1\}^n\\ |v\oplus w|=2\ell}} \left|\widetilde{f}(w)\right|\right]\leq \bound[10^9C](\ell),
\]
which in turn implies that $h'=h\cdot f$ is $(10^9 C, s^*)$-bounded on average. Applying Markov's inequality to the above then yields a proof of Lemma~\ref{lm:induction}, and applying Lemma~\ref{lm:induction} iteratively leads to a proof of Lemma~\ref{lm:main-comm}.

We now illustrate the main ideas of the proof of the implication above when $\ell$ is small, namely when $\ell\leq s^*$. For clarity of exposition we sketch the proof of the following bound:

\begin{equation}\label{eq:conv-sum-simple}
\sum_{v\in\{0,1\}^n} \left|\wt{h}(v)\right| \mathbb{E}_M\left[\sum_{\substack{z\in\{0,1\}^n\\ |v\oplus z|=2\ell}} \left|\widetilde{f}(z)\right|\right]\leq \binom{20Cs}{\ell},
\end{equation}
starting from the assumption that for all $\ell=1,\ldots, n/C^2$
\begin{equation}\label{eq:conv-sum-simple-org}
\sum_{\substack{v\in\{0,1\}^n \\ |v|=2\ell}} \left|\wt{h}(v)\right| \leq \binom{Cs}{\ell}.
\end{equation}
We note that the condition above implies $(C, s)$-boundedness.

We first note that $\wt{f}(z)\neq 0$ only if $z=Mw$ for $w\in \bool^M$ (see~\eqref{eq:at-comp-growing} and~\eqref{eq:wtfw} for the component growing protocol, and Lemma~\ref{lm:matching-fourier} for general protocols), and thus for every $v\in \bool^n$ one has $\mathbb{E}_M\left[\sum_{\substack{z\in\{0,1\}^n\\ |v\oplus z|=2\ell}} \left|\widetilde{f}(z)\right|\right]=\mathbb{E}_M\left[\sum_{\substack{w\in\bool^M\\ |v\oplus Mw|=2\ell}} \left|\widetilde{f}(Mw)\right|\right]$. We also note that since our matchings are small ($|M|\leq \alpha n$ for small constant $\alpha\in (0, 1)$, see, e.g. Lemma~\ref{lm:main-comm}), the condition $|v\oplus Mw|=2\ell$ implies that $|v|<n/100$. Thus, the only terms with a nonzero contribution to the sum that we need to bound are $v\in \bool^n$ with $|v|=2k\leq n/100$. Thus, it suffices to bound, for a parameter $k\in [0:n/200]$ and $v\in \bool^n$ with $|v|=2k$, the quantity
\begin{equation}\label{eq:sum-fixed-v-first}
\mathbb{E}_M\left[\sum_{\substack{w\in\bool^M\\ |v\oplus Mw|=2\ell}} \left|\widetilde{f}(Mw)\right|\right].
\end{equation}
Note that this quantity depends on $v$ only through $k=|v|$. We will later combine our bounds over all $k\in [0:n/200]$ to obtain the result of the lemma.

\paragraph{Bounding~\eqref{eq:sum-fixed-v-first} for fixed $v$: internal and boundary edges.} A vector $v\in\bool^n$ naturally corresponds to a subset of $[n]$ of its non-zero coordinates. Let $\Int=\{e_1^{\operatorname{int}}, e_2^{\operatorname{int}}, \dots\}$ be the set of edges $e=(a, b)\in M$ that match points of $v$, i.e. with $a, b\in v$. Let $\Bound=\{e_1^{\operatorname{bound}}, e_2^{\operatorname{bound}}, \dots \}$ be the set of boundary edges, i.e. edges $e=(a, b)\in M$ with $a \in v, b \not\in v$ or vice versa. Let $\Ext =\{e_1^{\operatorname{ext}}, e_2^{\operatorname{ext}}, \dots\}$ be the set of external edges, i.e. edges $e=(a, b)\in M$ with $a, b\in [n]\setminus v$. See Fig.~\ref{fig:int-ext-partial} for an illustration.

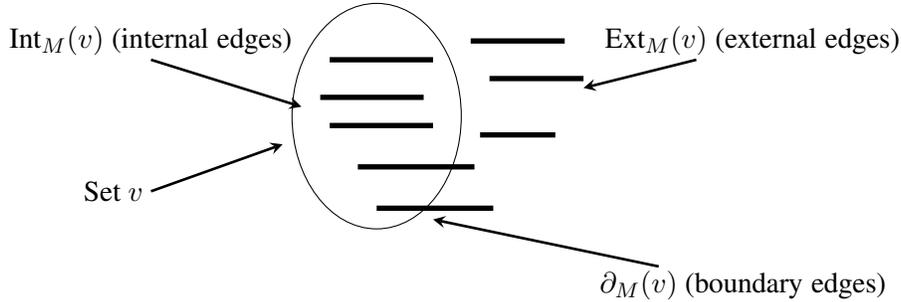
\begin{figure*}[h!]
\begin{center}
\tikzstyle{vertex}=[circle, fill=black!100, minimum size=5,inner sep=1pt]
\tikzstyle{svertex}=[circle, fill=red!100, minimum size=7,inner sep=1pt]
\tikzstyle{evertex}=[circle,draw=none, minimum size=25pt,inner sep=1pt]
\tikzstyle{edge} = [draw,-, color=red!100, very  thick]
\tikzstyle{bedge} = [draw,-, color=green!100, very  thick]
\begin{tikzpicture}[scale=0.25, auto,swap]

   \draw[rotate=90] (0,0) ellipse (6 and 4.5);

   \path[draw, line width=2pt, -] (-2.5, +3) -- (3.0, +3);
   \path[draw, line width=2pt, -] (-3, +1) -- (2.5, +1);
   \path[draw, line width=2pt, -] (-2.5, -0.5) -- (3.0, -0.5);   

   \path[draw, line width=2pt, -] (-1, -2.7) -- (5.2, -2.7);
   \path[draw, line width=2pt, -] (0, -4.9) -- (6.2, -4.9);

   \path[draw, line width=2pt, -] (-3+8, 4)     -- (2+8, 4);
   \path[draw, line width=2pt, -] (-3+9, 2)     -- (2+9, 2);   
   \path[draw, line width=2pt, -] (-2.5+8, -1) -- (1.5+8, -1);   
    
    \draw (-14, -4) node {Set $v$};    
    \path[draw, line width=1pt, ->, >=stealth] (-14+2, -4) -- (-5.0, -1.5);
    
    \draw (-14+2, 4) node {$\text{Int}_M(v)$ (internal edges)};        
    \path[draw, line width=1pt, ->, >=stealth] (-14+2, 4-1) -- (-4, +0.5);    

    \draw (-7+25+2, 4) node {$\text{Ext}_M(v)$ (external edges)};        
    \path[draw, line width=1pt, ->, >=stealth] (-7+24, 4-1) -- (-1+12, +1.5);    

    \draw (17.5+2, -9) node {$\partial_M(v)$ (boundary edges)};        
    \path[draw, line width=1pt, ->, >=stealth] (15, -8) -- (3, -5.5);


\end{tikzpicture}
\end{center}
\caption{A coefficient $v$ together edges of the matching $M$ classified into internal edges $\text{Int}_M(v)$, external edges $\text{Ext}_M(v)$ and boundary edges $\partial_M(v)$.}\label{fig:int-ext-partial}
\end{figure*}

 We decompose the sum~\eqref{eq:sum-fixed-v-first} according to the set of boundary edges in $w$:
\begin{equation*}
\mathbb{E}_M\left[\sum_{\substack{w\in\{0,1\}^M\\ |v\oplus Mw|=2\ell}} \left|\widetilde{f}(Mw)\right|\right]=\mathbb{E}_M\left[
\sum_{S \subseteq \Bound} \sum_{w\in\{0,1\}^M}
\mathbbm{1}_{\{w\cap\Bound=S \}} \mathbbm{1}_{\{|v\oplus Mw|=2\ell\}} \left|\widetilde{f}(Mw)\right|\right].
\end{equation*}
It turns out that if $w\cap\Bound=S$ then the latter indicator function can be rewritten as the indicator of $|w\oplus w_S|=\ell-(k-|\Int|)$, where $w_S\in\{0,1\}^M$ is  the set of all internal edges $\Int$ and all edges in $S$ (see the proof of Lemma \ref{lm:mass-transfer} for details). This gives the following upper bound:

\begin{equation*}
\begin{split}
\mathbb{E}_M\left[\sum_{\substack{w\in\{0,1\}^M\\ |v\oplus Mw|=2\ell}} \left|\widetilde{f}(Mw)\right|\right]&=
\mathbb{E}_M\left[\sum_{S \subseteq \Bound} \sum_{\substack{w\in\{0,1\}^M\\ |w\oplus w_S|=\ell-(k-|\Int|) }}
\mathbbm{1}_{\{w\cap\Bound=S\}} \cdot \left|\widetilde{f}(Mw)\right|\right]\\
&\leq 
\mathbb{E}_M\left[\sum_{S \subseteq \Bound} \sum_{\substack{w\in\{0,1\}^M\\ |w\oplus w_S|=\ell-(k-|\Int|) }}
 \left|\widetilde{f}(Mw)\right|\right],
\end{split}
\end{equation*}
where we dropped the indicator function in going from the first line to the second line above. Note that $|\widetilde{f}(Mw)|=1$ if all the edges of $w$ are revealed by the $t$-th player and zero otherwise. Since all the edges of $w_S$ (i.e. internal edges and some subset of the boundary edges) are revealed we know that $w\oplus w_S$ is also a subset of the revealed edges which means that we have 
\begin{equation}\label{eq:92ugjggdfd}
\sum_{\substack{w\in\{0,1\}^{M}\\ |w\oplus w_S|=\ell-(k-|\Int|) }} \left|\widetilde{f}(Mw)\right| \leq \binom{s}{\ell-k+|\Int|}. 
\end{equation}

Note that here we are crucially using the fact that $w\oplus w_S$ is a set of external edges only. This allows us to bound the sum of Fourier coefficients by a function of $s$, the communication budget per player. The equivalent statement for general protocols is provided by Lemma~\ref{lm:L1xorKKL} (see~\eqref{eq:xor-KKL-bound}), which bounds the sum of absolute values of Fourier coefficients of a dense subset of the boolean cube by a similar expression to the above. 

Since the bound in~\eqref{eq:92ugjggdfd} is independent of $S$, summing over all possible subsets $S\subseteq \Bound$, we infer
\begin{equation}\label{eq:bound-external}
\begin{split}
\mathbb{E}_M\left[\sum_{\substack{w\in\{0,1\}^M\\ |v\oplus Mw|=2\ell}} \left|\widetilde{f}(Mw)\right|\right]
&\leq
\mathbb{E}_M\left[2^{|\Bound|} \binom{s}{\ell-k+|\Int|} \right].
\end{split}
\end{equation}

We then bound the sum on the last line above:
\begin{align*}
\mathbb{E}_M\left[2^{|\Bound|}\binom{s}{\ell-k+|\Int|} \right]
&=\mathbb{E}_M\left[\sum_{i=0}^k \sum_{b=0}^{2k} 2^{b} \mathbbm{1}_{\{|\Bound|=b\text{~and}~|\Int|=i\}} \binom{s}{\ell-k+i} \right]\\
&=\sum_{i=|k-\ell|_+}^k \sum_{b=0}^{2k} 2^{b} q(k,i,b,n)  \binom{s}{\ell-k+i}.
\end{align*}   
In going from line~1 to line~2 above we used the fact that by Definition~\ref{def:qkibn} for every $v\in \bool^n$ with $|v|=2k$ one has $\expect_M[\mathbbm{1}_{\{|\Bound|=b\text{~and}~|\Int|=i\}}]=q(k, i, b, n)$, i.e. $q(k,i,b,n)$ is the probability that a uniformly random matching $M$ of size $\alpha n$ is  such that $i$ edges of $M$ match points of $v$ (i.e. $\left|\Int\right|=i$) and $b$ edges of $M$ are boundary edges (i.e. $\left|\Bound\right|=b$).  Recall that $|k-\ell|_+=\max\{0, k-\ell\}$. We note that $q(k, i, b, n)$ depends on the size  $\alpha n$ of the matching $M$, but we prefer to keep this dependence implicit to simplify notation.

We thus have that for a fixed $v\in \{0,1\}^n$ with $|v|=2k$
\begin{equation}\label{eq:9275tjffads}
\mathbb{E}_M\left[\sum_{\substack{w\in\{0,1\}^M\\ |v\oplus Mw|=2\ell}} \left|\widetilde{f}(Mw)\right|\right] 
\leq
\sum_{i=|k-\ell|_+}^{k}\sum_{b=0}^{2k} q(k,i,b,n) 2^{b}\binom{s}{\ell-k+i}.
\end{equation}
We now note that intuitively, $q(k, i, b, n)\approx n^{-i}$, since (at least for small $i$ and $k$) the probability of having an internal edge in a given set of size $2k$ is approximately $1/n$. In fact, for small $k,i$ we have $q(k,i,b,n)\lesssim 10^k n^{-i}$ (the formal bounds are somewhat different for larger $k$ and $i$, and are given in Lemma~\ref{lm:qkibninequality} and Lemma~\ref{lm:qkin}). 

We now consider two cases. 

\paragraph{Case 1: $k\geq \ell$.} This case essentially corresponds to analyzing the rate at which collection of $k$ components get merged into collections of $\ell$ components.
Using  Lemma~\ref{lm:qkibninequality} and Lemma~\ref{lm:qkin} we get that the summand in~\eqref{eq:9275tjffads} decays exponentially with $b$ and $i$, and since in this case the sum starts at $i=k-\ell$, we get
\begin{equation*}
\begin{split}
\mathbb{E}_M\left[\sum_{\substack{w\in\{0,1\}^M\\ |v\oplus Mw|=2\ell}} \left|\widetilde{f}(Mw)\right|\right] 
&\leq \sum_{i=k-\ell}^{k}\sum_{b=0}^{2k} q(k,i,b,n) 2^{b}\binom{s}{\ell-k+i}\\
&\lesssim 10^k\cdot \sum_{i=k-\ell}^{k}n^{-i}\binom{s}{\ell-k+i}\\
&\approx 10^k\cdot n^{-(k-\ell)}.
\end{split}
\end{equation*}
This is consistent with the intuition that (at least for constant $k$ and $\ell$) the probability that a given collection of $k$ constant size components becomes merged into only $\ell<k$ components is about $n^{-(k-\ell)}$: this is simply because such an event requires at least $k-\ell$ edges of the matching $M_t$ to have both endpoints inside the $k$ components.

\paragraph{Case 2: $k>\ell$.} In this case, as in case 1, the sum in~\eqref{eq:9275tjffads} above is close to the value of the maximum summand, which gives, since the sum starts with $i=|k-\ell|_+=0$, 
\begin{equation*}
\begin{split}
\mathbb{E}_M\left[\sum_{\substack{w\in\{0,1\}^M\\ |v\oplus Mw|=2\ell}} \left|\widetilde{f}(Mw)\right|\right] 
&\leq \sum_{i=0}^{k}\sum_{b=0}^{2k} q(k,i,b,n) 2^{b}\binom{s}{\ell-k+i}\\
&\lesssim 10^k\cdot\binom{s}{\ell-k}.
\end{split}
\end{equation*}
This is consistent with the intuition that a given collection of $k$ components in $F_{t-1}$ can contribute to about $\binom{s}{\ell-k}$ size $\ell$ components in $F_t$: simply consider adding any subset of $\ell-k$ edges of the matching $M_t$ that do not intersect with nontrivial components in $F_{t-1}$. Note that this is the part where we analyze the contribution of the bits that the players are actually charged for in the component growing protocol, i.e. external edges.\footnote{At the same time, one should note that while the simple intuitive overview of case 1 does not involve external edges, they do have an effect on the actual proof -- see Lemmas~\ref{lm:mass-transfer} and~\ref{lm:mass-transfer-high} in Section~\ref{sec:main-result}. The fact that their effect in this case is second order lets us present the simple intuition for case 1 above.}

Using the bounds from {\bf Case 1} and {\bf Case 2} above in equation~\eqref{eq:conv-sum-simple}, we obtain 
\begin{align*}
\sum_{v\in\{0,1\}^n} \left|\wt{h}(v)\right| \mathbb{E}_M\left[\sum_{\substack{z\in\{0,1\}^n\\ |v\oplus z|=2\ell}} \left|\widetilde{f}(z)\right|\right]
&\lesssim
\sum_{k\leq \ell}10^k\cdot\binom{Cs}{k}\cdot \binom{s}{\ell-k}+\sum_{k>\ell}10^k\cdot\binom{Cs}{k}\cdot n^{\ell-k}\\
&\lesssim \binom{20Cs}{\ell}.
\end{align*}
The sketch above is informal, but Lemma~\ref{lm:mass-transfer} (whose proof is given in Section~\ref{sec:main-result}) provides the formal version that matches the qualitative conclusion. We state this lemma here for convenience of the reader:

\noindent {\em {\bf Lemma~\ref{lm:mass-transfer}}
For every $n, s^*$, $C>1$, $\alpha\in (0, 1)$ that satisfy conditions
\begin{align*}
{\bf (P1)} \, \alpha < 10^{-10}\quad 
{\bf (P2)} \, C>10^6 \quad 	
{\bf (P3)} \, s^*<\frac{n}{10^9 C^3} \quad
{\bf (P4)} \, n>10^9C^4,
\end{align*}
 if $\B\subseteq \bool^n$ is $(C, s^*)$-bounded and $M$ is a uniformly random matching  on $[n]$ of size $\alpha n$, the following conditions hold. For every $\A_{reduced}\subseteq \bool^M$, if $\A=\{x\in \bool^n: Mx\in \A_{reduced}\}$ and $f$ is the indicator of $\A$, if $|\A|/2^n\geq 2^{-s^*}$, then for every $\ell\in [1, s^*]$
\[
\sum_{v\in\{0,1\}^n} \left|\wt{h}(v)\right| \mathbb{E}_M\left[\sum_{\substack{z\in\{0,1\}^n\\ |v\oplus z|=2\ell}} \left|\widetilde{f}(z)\right|\right]\leq \left(\frac{10^9 C \sqrt{s^* n}}{\ell}\right)^\ell.
\]
}

Similar ideas lead to the proof of Lemma \ref{lm:s1}. In the general case we also need to bound the mass transfer from intermediate and high weights, see Lemmas \ref{lm:s2} and \ref{lm:s3}. 

\section{Proof of main theorem (Theorem~\ref{thm:main-comm})}\label{sec:main-result}

In this section we give the proof of Theorem~\ref{thm:main-comm}, restated here for convenience of the reader:

\noindent{{\bf Theorem~\ref{thm:main-comm}} \em 
\thmMainComm
}

\paragraph{Section outline.} The proof is structured as follows. We first state our main technical lemma (Lemma~\ref{lm:induction}, proved in Section~\ref{sec:inductionstep}), which is the core tool behind the proof. Lemma~\ref{lm:induction} analyzes the relation between the properties of the Fourier transform of $\B_{t-1}$ and the Fourier transform of $\B_t$: we show that if $\B_{t-1}$ is $(C, s^*)$-bounded (as per Definition~\ref{def:bounded}) for some parameter $s^*\approx T\cdot s$ that is essentially the total communication budget of the players and $M_t$ is a uniformly random matching of size $\alpha n$,  then  $\B_t$ is $((10^9/\delta)C, s^*)$-bounded with high probability (note that the lemma is stated without any reference to the index $t$, but is actually applied iteratively as above in subsequent analysis). We then prove Lemma~\ref{lm:main-comm} below, which iteratively applies Lemma~\ref{lm:induction} and analyzes the evolution the Fourier coefficients of $\B_t$. The proof of Theorem~\ref{thm:main-comm} then follows by combining basic properties of the total variation distance (Lemma~\ref{lm:main-tvd} below) with Lemma~\ref{lm:main-comm}.

The following technical lemma, which is the main result of Section~\ref{sec:inductionstep}, is the main tool in our proof:

\begin{lemma}\label{lm:induction}
\lemmaInduction
\end{lemma}

\subsection{Evolution of Fourier coefficients}

The next lemma bounds the evolution of Fourier coefficients of the sets $\B_t$ for $t=1,\ldots, T$, and forms the main technical part of the proof of Theorem~\ref{thm:main-comm}. The proof of Lemma~\ref{lm:main-comm} essentially amounts to iteratively applying Lemma~\ref{lm:induction} and keeping track of the density $|\B_t|/2^n$ throught a natural inductive hypothesis, using Lemma~\ref{lm:pdfs-closeness} at every step of the induction to argue that $|\B_t|/2^n\approx (|\B_{t-1}|/2^n)\cdot (|\A_t|/2^n)$ (i.e., the densities almost multiply):

\begin{lemma}\label{lm:main-comm}
\lemmaMainComm 
\end{lemma}
\begin{proof}
We let $s^*:=10Ts$ to simplify notation.  Our proof is by induction on $t=1,\ldots, T$.

\paragraph{Base ($t=1$)} The event $\E_1$ happens with probability $1$, and the set $\B_0$ is defined as $\B_0=\bool^n$, so the inductive claim is satisfied since for any matching $M_1$ a random variable $M_1 x$ with $x\sim\operatorname{Uniform}(\bool^n)$ is uniform in $\bool^{M_1}$.

\paragraph{Inductive step: $t\to t+1$} We first define the events $\E'_{t+1}$ and $\E''_{t+1}$ and prove that both of them occur with high probability conditioned on $\E_t$. We then show that the inductive hypothesis for $t+1$ holds conditioned on $\E_{t+1}$, which we define to be $\E_t\cap \E'_{t+1}\cap \E''_{t+1}\cap \E'''_{t+1}$, where $\E'''_{t+1}$ is a similar event that ensures that the third condition is satisfied, establishing the inductive step.

\paragraph{Constructing events $\E'_{t+1}$ and $\E''_{t+1}$.}

Define $\E_{t+1}':=\{|\A_t|\geq 2^{n-s-9\log{T}}\}$. Conditioned on $\E_t$, we have by the inductive hypothesis and Lemma \ref{lm:partition}
\[
\mathbb{E}_{x\sim \operatorname{Uniform}(\B_{t-1})}[1/|\A_t|]\leq (1+\delta)\cdot \mathbb{E}_{x\sim \operatorname{Uniform}\left(\{0,1\}^{M_t}\right)}[1/|\A_t|] \leq (1+\delta)\cdot 2^{s-n}.
\]
 This means that by Markov's inequality for the event  we have $\prob[\E_{t+1}' | \E_t]>1-(1+\delta)\cdot 2^{-9\log{T}}>1-\delta$.

\paragraph{Lower bounding $|\B_t|/2^n$.} We start by showing that conditioned on $\E_t$ the set $\B_t$ satisfies 
\begin{equation}\label{eq:bt-size}
\frac{|\B_t|}{2^n}\geq (1-\delta)\frac{|\B_{t-1}|}{2^n}\frac{|\A_t|}{2^n}.
\end{equation}
Indeed, it suffices to take for $q$ the indicator function of $\A_{reduced, t}$ and apply~\eqref{eq:pointwise}. Since
\[
\mathbb{E}_{x\sim \operatorname{Uniform}(\B_{t-1})}[q(Mx)]=\frac{|\B_{t-1}\cap \A_t|}{|\B_{t-1}|}=\frac{|\B_t|}{|\B_{t-1}|}, \qquad
\mathbb{E}_{z\sim \operatorname{Uniform}(\{0,1\}^{M_t})}[q(z)]=\frac{|\A_{reduced,t}|}{2^{\alpha n}}=\frac{|\A_t|}{2^n},
\]  
we have 
\begin{equation*}
\begin{split}
\frac{|\B_t|}{2^n}&=\frac{|\B_{t-1}|}{2^n}\cdot \mathbb{E}_{x\sim \operatorname{Uniform}(\B_{t-1})}[q(Mx)]\\
&\geq (1-\delta)\frac{|\B_{t-1}|}{2^n}\cdot \mathbb{E}_{z\sim \operatorname{Uniform}(\{0,1\}^{M_t})}[q(z)]\text{~~~~~~~~(by~\eqref{eq:pointwise})}\\
&= (1-\delta)\frac{|\B_{t-1}|}{2^n}\cdot \frac{|\A_t|}{2^n}
\end{split}
\end{equation*}
establishing~\eqref{eq:bt-size}. Now by conditioning on $\E'_{t+1}$ we have $|\A_t|\geq 2^{n-s-9\log{T}}$. Putting this together with~\eqref{eq:bt-size} yields
\begin{equation}
\begin{split}
\frac{|\B_t|}{2^n}&\geq (1-\delta)\frac{|\B_{t-1}|}{2^n}\frac{|\A_t|}{2^n}\\
&\geq (1-\delta)\frac{|\B_{t-1}|}{2^n}\cdot 2^{-s-9\log{T}}\text{~~~~~~~~~~~~~~~~~~~~~~~~~~~(by conditioning on $\E'_{t+1}$)}\\
&\geq (1-\delta)2^{-s(t-1)-10(t-1)\log{T}}\cdot 2^{-s-9\log{T}}\text{~~~~~(by the inductive hypothesis)}\\
&=(1-\delta)\cdot 2^{\log{T}}\cdot 2^{-st-10t\log{T}}\\
&\geq 2^{-st-10t\log{T}}\text{~~~~~~~~~~~~~~~~~~~~~~~~~~~~~~~~~~~~~~~~~~~~~~~~~~(since $\delta < 1/2$ and $T \geq 10$)}
\end{split}
\end{equation}
establishing the lower bound on $|\B_t|$ required for the inductive step.

\paragraph{Proving that $\B_t$ is $((10^{12} T)^{t+1}, 10Ts)$-bounded.} Let $\E''_{t+1}$ be the success event from Lemma~\ref{lm:induction} applied to the sets $\B:=\B_{t-1}$ and $\A_{reduced}=\A_{reduced, t}$ with $C=(10^{12}T)^t$, $\delta=1/(10^3 T)$, $s^*=10T s$. By Lemma~\ref{lm:induction} (we verify that preconditions are satisfied below) we have $\prob[\E''_{t+1}|\E_t]\geq 1-5\delta$ and conditioned on $\E''_{t+1}$  the set $\B_t=\B_{t-1}\cap \A_t$ is $((10^9/\delta)C, s^*)$-bounded, as required.

We now verify that the preconditions of Lemma \ref{lm:induction} are satisfied. We have to check five conditions:
\begin{align*}
{\bf (P1)} \, \alpha < 10^{-10}\quad 
{\bf (P2)} \, C>10^6 \quad 	
{\bf (P3)} \, s^*<\frac{n}{10^9 C^3} \quad
{\bf (P4)} \, n>10^9C^4 \quad
{\bf (P5)} \, \delta \in (n^{-1}, 1/2)
\end{align*}
First two conditions are clear since $\alpha < 10^{-10}$ and $C=(10^{12} T)^t>10^6$. Condition {\bf (P3)} is satisfied since 
\[
s^*=10Ts< \frac{n}{(10T)^{10^9T-1}}\leq \frac{n}{(10^{12}T)^{100T}}\leq \frac{n}{10^9\cdot (10^{12}T)^{3T}}\leq \frac{n}{10^9 C^3}.
\] 
Condition {\bf (P4)} is satisfied because $s \geq 1$ implies $n\geq (10T)^{10^9 T}$ and so 
\[
n\geq (10T)^{10^9 T} > 10^9\cdot (10^{12}T)^{4T}\geq 10^9 C^4.
\]
Condition {\bf (P5)} is satisfied since $1/2>\frac{1}{10^3 T} = \delta >n^{-1}$, as $T\leq \ln n$ by assumption.

\paragraph{Proving point-wise closeness of pdfs.} 
Conditioned on $\E_t\cap \E_{t+1}'\cap \E_{t+1}''$ the set $\B_t$ is $((10^{12} T )^{t+1}, s^*)$-bounded and satisfies $|\B_t|/2^n\geq 2^{-st-10t\log T}\geq 2^{-s^*}$, since $T\leq \ln n$ and $s\geq \sqrt{n}$ by assumption of the theorem. Let $\E_{t+1}'''$ be the event (over the choice of the matching $M_{t+1}$) that for any non-negative function $q$ on $\{0,1\}^{M_{t+1}}$
\begin{equation}\label{eq:pointwisenew}
1-\delta\leq \frac{\mathbb{E}_{x\sim \operatorname{Uniform}(\B_{t-1})}\left[q(M_tx)\right]}{\mathbbm{E}_{z\sim \operatorname{Uniform}\left(\{0,1\}^{M_t}\right)}\left[q(z)\right]}\leq 1+\delta,
\end{equation}
that is, the event that the third assumption is satisfied for $t+1$. By Lemma~\ref{lm:pdfs-closeness} invoked with $C=(10^{12} T )^{t+1}$ and $s^*=10Ts$ and $\delta=1/(1000T)$ we have $\prob[\E_{t+1}''' | \E_t\cap \E_{t+1}'\cap \E_{t+1}'']\geq 1-\delta$. Indeed, preconditions of the lemma are satisfied since we have $C=(10^{13} T )^{t+1}>100, \delta=1/(10^3 T)\in (n^{-1}, 1/2)$, $\alpha < 1/100$ and 
$$
s^*=10Ts\leq (10 T)n/(10T)^{10^9 T}= n/(10T)^{10^9 T-1}\leq \left(\frac1{1000T}\right)^4 n/(10^{13} T )^{2t}=\delta^4 n/C^2.
$$

It remains to note that for $\E_{t+1}:=\E_t\cap \E_{t+1}'\cap \E_{t+1}'' \cap \E_{t+1}'''$ we have 
\begin{equation*}
\begin{split}
\prob[\bar \E_{t+1}|\E_t]\leq \prob[\E_{t+1}'|\E_t]+\prob[\E_{t+1}''| \E_t]+\prob[\E_{t+1}''' | \E_t\cap \E_{t+1}'\cap \E_{t+1}'']\leq \delta+5\delta+\delta<1/(100T).
\end{split}
\end{equation*}

\end{proof}

\subsection{Putting it together}

We now combine Lemma~\ref{lm:main-comm} from the previous section with basic properties of the total variation distance to obtain a proof of Theorem~\ref{thm:main-comm}. Specifically, the proof relies on Lemma~\ref{lm:main-tvd} below, which is essentially the hybrid argument. Informally, the lemma says the following. Consider the \YES case distribution $\D^Y$, and for each $t=1,\ldots, T-1$ compare the distribution of $S^Y_{t+1}$ to the distribution of messages of the $(t+1)$-st player obtained by supplying this player with random labels on their edges as opposed to labels consistent with the hidden bipartition $X^*$ (the latter leads to $r_t(M_{1:t}, S^Y_{1:t-1}, U_t)$). If the resulting distributions are close in total variation distance, then the joint distribution of {\em all messages and matchings $M_t$} posted on the board after $T$ rounds in the \YES case is close to the same distribution in the \NO case.

\begin{lemma}\label{lm:main-tvd}
Let $X^*\sim UNIF(\bool^n)$ be a uniformly random binary vector of length $n$. For each $t=1,\ldots, T$ let $U_t\sim UNIF(\bool^{\alpha n})$ be an independent uniformly random vector of length $\alpha n$. Let $M_1,\ldots, M_T$ be independently chosen random matchings on $[n]$. 

Let $S^Y_0=S^N_0:=0$ and for each $t=1,\ldots,  T$ let 
\begin{equation*}
\begin{split}
S^Y_{t}&:=r_t(M_{1:t}, S^Y_{1:t-1}, M_t X^*)\\
\text{and}&\\
S^N_{t}&:=r_t(M_{1:t}, S^N_{1:t-1}, U_t)\\
\end{split}
\end{equation*}
for some functions $r_t, t=1,\ldots, T$. Suppose that there exists a sequence of events $\E_1\supset \E_2 \supset \ldots \supset \E_T$ such that $\E_t$ depends only on $M_{1:t}$ and $S^Y_{1:t-1}$, and $\E_1$ occurs with probability $1$, such that for any fixed $M_{1:t}, S^Y_{1:t-1}$ satisfying $\E_t$ one has for some $\gamma>0$
\begin{equation}\label{eq:ind-assumption}
||S^Y_t-r_t(M_{1:t}, S^Y_{1:t-1}, U_t)||_{tvd}\leq \gamma/T.\tag{*}
\end{equation}
Suppose further that for each $t=1,\ldots, T$ we have $\prob[\bar \E_t | \E_{t-1}]\leq\gamma/T$.
Then 
$$
||(M_{1:T}, S^Y_{1:T})-(M_{1:T}, S^N_{1:T})||_{tvd}\leq 2\gamma.
$$
\end{lemma}
The (simple) proof of the lemma is given in Appendix~\ref{app:B}.

We now give

\begin{proofof}{Theorem~\ref{thm:main-comm}}
Follows by putting together Lemma~\ref{lm:main-tvd} and Lemma~\ref{lm:main-comm}, as we show below.

Assume by contradiction that for any $C_0$ there exists a protocol $\Pi$ for $\DIHP(n,\alpha,T)$ succeeding with probability at least $2/3$ and satisfying $|\Pi|< n/(C_0/\eps)^{C_0/\eps^2}$. Take $C_0=10^{30}$, since $T\leq 10^{14}/\eps^2$, we have
\[
|\Pi|<n/(C_0/\eps)^{C_0/\eps^2}\leq n/(10^{30}/\eps)^{10^{16}T}\leq  n/(\sqrt{10T})^{10^{16}T} < n/(10 T)^{10^9 T}.
\]
Without loss of generality we may assume that $s=|\Pi|$ satisfies $s\geq \sqrt{n}$, and hence we can apply Lemma \ref{lm:main-comm} since $\alpha=10^{-11}\in (0,10^{-10})$ and $T\in [10, \ln{n}]$ for $n$ large enough. We choose $\E_t$ to be the events whose existence is guaranteed by Lemma~\ref{lm:main-comm}. By Lemma~\ref{lm:main-comm} one has, conditioned on $\E_t$, for every fixed $M_t$, that
\begin{equation}\label{eq:tvd-one-step}
||M_tX^*-UNIF(\bool^{M_t})||_{tvd}\leq \delta=1/(1000 T)
\end{equation}
for every $t=1,\ldots, T$. We claim that these events satisfy the preconditions of Lemma~\ref{lm:main-tvd}. Indeed, recalling that $S^Y_t=r_t(M_{1:t}, S_{1:t-1}^Y, M_t X^*)$ (see Definition~\ref{def:at}), we get for any fixed $S^Y_{t-1}$ and $M_{1:t}$ satisfying $\E_{t-1}$
\begin{equation}
\begin{split}
||S^Y_t-r_t(M_{1:t}, S^Y_{1:t-1}, U_t)||_{tvd}&=||r_t(M_{1:t}, S^Y_{1:t-1}, M_t X^*)-r_t(M_{1:t}, S^Y_{1:t-1}, U_t)||_{tvd}\\
&\leq ||M_t X^*-U_t||_{tvd},
\end{split}
\end{equation}
where we applied Claim~\ref{cl:1} in the last transition with $f=r_t$, $W=(M_{1:t}, S^Y_{1:t-1})$ and $X=M_t X^*$ and $Y=U_t$. We stress the fact that here $W$ is deterministic and so the tvds are over the randomness of $X^*\sim \operatorname{Uniform}(\B_{t-1}), U_t\sim\operatorname{Uniform}(\bool^{M_t})$. 

Since by Lemma \ref{lm:main-comm} we have $\prob[\bar \E_t|\E_{t-1}]\leq 10\delta$ and $||M_t X^*-U_t||_{tvd} \leq \delta$, we can then apply Lemma \ref{lm:main-tvd} with $\gamma=10\delta T = 1/100$ to deduce that
\[
||(M_{1:T}, S^Y_{1:T})-(M_{1:T}, S^N_{1:T})||_{tvd}\leq \gamma,
\]
which means that it is not possible to distinguish between \YES and \NO cases with probability more than $(1+\gamma)/2$ leading to a contradiction.
\end{proofof}


\section{Proof of main technical lemma (Lemma~\ref{lm:induction})}\label{sec:inductionstep}
The main result of this section is a proof of Lemma~\ref{lm:induction}, restated here for convenience of the reader:\\

\noindent{{\bf Lemma~\ref{lm:induction}} \em (Restated)
For every $n, C, s^*, \alpha, \delta$ that satisfy conditions 
\begin{align*}
{\bf (P1)} \, \alpha < 10^{-10}\quad 
{\bf (P2)} \, C>10^6 \quad 	
{\bf (P3)} \, s^*<\frac{n}{10^9 C^3} \quad
{\bf (P4)} \, n>10^9C^4 \quad
{\bf (P5)} \, \delta \in (n^{-1}, 1/2),
\end{align*}
every $\B\subseteq \{0, 1\}^n$, if $\B$ is $(C, s^*)$-bounded and $M$ is a uniformly random matching of size $\alpha n$,  the following conditions hold with probability at least $1-5\delta$ over the choice of $M$.

For every $\A_{reduced}\subseteq \{0, 1\}^M$ such that $|\A_{reduced}|/2^{\alpha n}\geq 2^{-s^*}$, if $\A=\{x\in \{0, 1\}^n: Mx\in \A_{reduced}\}$, then 
$\B':=\B\cap \A$ is $((10^9/\delta)C, s^*)$-bounded.
}\\

We also restate the definition of $(C, s^*)$-boundedness and the definition of $\bound(\ell)$ for convenience.

\noindent{\em {\bf Definition~\ref{def:bounded}}($(C, s^*)$-bounded set; restated)
Let $\B\subset \{0,1\}$ with indicator function $h$. We say that $\B$ (or $h$) is {\em $(C, s^*)$-bounded} if 
\begin{itemize}
\item For all $\ell\leq s^*$ we have
\[
\sum_{\substack{v\in\{0,1\}^n \\ |v|=2\ell}}\left|\wt{\mathbbm{1}_\B}(v)\right|\leq \left(\frac{C\sqrt{s^*n}}{\ell}\right)^\ell;
\]
\item For all $s^* < \ell < \frac{n}{C^2}$ we have
\[
\sum_{\substack{v\in\{0,1\}^n \\ |v|=2\ell}}\left|\wt{\mathbbm{1}_\B}(v)\right|\leq \left(\frac{C^2 n}{\ell}\right)^{\ell/2}.
\]
\end{itemize}
}

As per~\eqref{eq:bound-def}, defining the function $\bound(\ell)$ by
\begin{equation*}
\bound(\ell)=
\begin{cases}
1 \qquad \qquad \qquad \ell=0; \\
\left(\frac{C\sqrt{s^* n}}{\ell}\right)^\ell \qquad \ell\in [1: s^*]; \\
\left(\frac{C^2 n}{\ell}\right)^{\ell/2} \qquad \ell>s^*,
\end{cases}
\end{equation*}
we are able to simplify notation, ensuring that an indicator function $h$ is $(C, s^*)$ bounded if and only if for all $\ell<n/C^2$ we have 
\[
\sum_{\substack{v\in \{0,1\}^n \\ |v|=2\ell}} \left|\wt{h}(v)\right|\leq \bound(\ell).
\]

\paragraph{Proof outline.} The lemma starts with the assumption that a subset $\B$ of $\bool^n$ is $(C, s^*)$-bounded, and proves that the intersection $\B'=\B\cap \A$ with a subset $\A$ of $\bool^n$ (which should be thought of as the typical message corresponding to a player who receives labels $MX$ for $X$ selected uniformly at random from $\B$) is $(10^9 C, s^*)$-bounded.  It is instructive to contrast the bounds implied by the set $\B$ being $(C, s^*)$-bounded (see Definition~\ref{def:bounded} and eq.~\eqref{eq:bound-def}, also restated above)  with corresponding bounds for our component growing algorithm presented in Section~\ref{sec:outline-app} (see Section~\ref{sec:analysis-overview}). Intuitively, the rhs in the definition of $(C, s^*)$-boundedness above shows that the amount of Fourier mass that a general protocol can have at some level $\ell$ is upper bounded by the amount of mass that the component growing protocol with slightly increased budget, namely $\sqrt{s^* n}=s^*\cdot \sqrt{n/s^*}\gg s^*$, can have at levels $\ell\in [1, s^*]$. This increase from $s^*$ to $\sqrt{s^* n}$ is due to our conversion on the  $\ell_2^2$ norm of the Fourier transform to bounds on the $\ell_1$ norm.

The proof of Lemma~\ref{lm:induction} is based on the convolution theorem, and follows quite closely the outline presented in Section~\ref{sec:outline-app}. 
The main part of the proof that goes beyond the outline presented in Section~\ref{sec:outline-app} is the analysis of contribution to and from weights higher than $s^*$ (specifically, weight levels in $[s^*, n/(2C^2)]$ and $[n/(2C^2), n]$), as such weight levels carry zero Fourier mass for the simple component growing protocol from Section~\ref{sec:outline-app}.  The corresponding analysis is carried out in Lemma~\ref{lm:mass-transfer} (mass transfer {\bf to} low weights, possibly from low, intermediate or high weight coefficients) and Lemma~\ref{lm:mass-transfer-high} (mass transfer {\bf to} intermediate weight coefficients). At the same time, we note that most of the final contribution to the Fourier transform of the final function comes from `low' weight levels, namely from $\ell\in [1, s^*]$ (see Lemma~\ref{lm:s1}), similarly to the component growing protocol in Section~\ref{sec:outline-app}, and the dominant bound is thus given by Lemma~\ref{lm:mass-transfer}.

\begin{proofof}{Lemma~\ref{lm:induction}}
By Lemma \ref{lm:pdfs-closeness} with probability at least $1-\delta$ over the choice of $M$ we have 
\[
\frac{2^n}{|\B'|}\leq \frac{2^n}{|\B|}\cdot \frac{2^n}{|\A|}\cdot \frac{1}{1-\delta}.
\]
Indeed we just need to take $q$ the indicator function of $A_{reduced}$ (see the proof of \ref{eq:bt-size} in Lemma \ref{lm:main-comm}). So we show that conditioned on this the set $\B'$ is $((10^9/\delta)C, s^*)$-bounded with probability at least $1-4\delta$.

We let $h:=\mathbf{1}_{\B}, f:=\mathbf{1}_\A, h':=h\cdot f=\mathbf{1}_{\B'}$. Proving that $h'$ is $((10^9/\delta)C, s^*)$-bounded involved upper bounding the $\ell_1$ norm of Fourier coefficients at various levels $2\ell, \ell\in [1, n/(2C^2)]$. Convolution theorem~\eqref{eq:convolution} together with triangle inequality give
\begin{equation}\label{eq:conv-hats}
\begin{split}
\sum_{\substack{v\in\{0,1\}^n\\ |v|=2\ell}} \left|\widehat{h}'(v)\right| &\leq \sum_{\substack{v\in\{0,1\}^n\\ |v|=2\ell}} \sum_{w\in\{0,1\}^n} \left|\wh{h}(v\oplus w)\wh{f}(w)\right|\\&=\sum_{v\in\{0,1\}^n} \sum_{\substack{w\in\{0,1\}^n\\ |v\oplus w|=2\ell}} \left|\wh{h}(v)\wh{f}(w)\right| \\&= \sum_{v\in\{0,1\}^n} \left|\wh{h}(v)\right| \sum_{\substack{w\in\{0,1\}^n\\ |v\oplus w|=2\ell}} \left|\wh{f}(w)\right|.
\end{split}
\end{equation}

We need to bound $\sum_{\substack{v\in\{0,1\}^n\\ |v|=2\ell}} \left|\wt{h}'(v)\right|=\frac{2^n}{|\B'|}\sum_{\substack{v\in\{0,1\}^n\\ |v|=2\ell}} \left|\widehat{h}'(v)\right|$, but it turns out to be more useful to upper bound $\frac{2^n}{|\B|}\cdot \frac{2^n}{|\A|}\sum_{\substack{v\in\{0,1\}^n\\ |v|=2\ell}} \left|\widehat{h}'(v)\right|$ which is within a factor of $(1-\delta)$ of what we need. Multiplying both sides of~\eqref{eq:conv-hats} by $\frac{2^n}{|\B|}\cdot \frac{2^n}{|\A|}$, we get
\begin{align*}
\frac{2^n}{|\A|}\cdot\frac{2^n}{|\B|}\cdot\sum_{\substack{v\in\{0,1\}^n\\ |v|=2\ell}} \left|\widehat{h}(v)\right| &\leq \sum_{\substack{v\in\{0,1\}^n\\ |v|=2\ell}} \sum_{w\in\{0,1\}^n} \left|\widetilde{h}(v\oplus w)\widetilde{f}(w)\right|\\&=\sum_{v\in\{0,1\}^n} \sum_{\substack{w\in\{0,1\}^n\\ |v\oplus w|=2\ell}} \left|\widetilde{h}(v)\widetilde{f}(w)\right| \\&= \sum_{v\in\{0,1\}^n} \left|\widetilde{h}(v)\right| \sum_{\substack{w\in\{0,1\}^n\\ |v\oplus w|=2\ell}} \left|\widetilde{f}(w)\right|.
\end{align*}
Taking expectation over the choice of $M$ and recalling that $2^n/|\B'|\leq (2^n/|\B|)\cdot (2^n/|\A|) / (1-\delta)$ we get 
\begin{align}\label{eq:rhs}
\frac{2^n}{|\B'|}\mathbb{E}_M\left[\sum_{\substack{v\in\{0,1\}^n\\ |v|=2\ell}} \left|\widetilde{h}'(v)\right|\right] \leq \frac{1}{1-\delta}\cdot \sum_{v\in\{0,1\}^n} \left|\widetilde{h}(v)\right| \mathbb{E}_M\left[\sum_{\substack{w\in\{0,1\}^n\\ |v\oplus w|=2\ell}} \left|\widetilde{f}(w)\right|\right].
\end{align}
By Lemma \ref{lm:mass-transfer} we have for all $\ell\in [1, s^*]$, using the fact that $|\A|/2^n=|\A_{reduced}|/2^{\alpha n}\geq 2^{-s^*}$ by assumption of the lemma, we get
\[
\sum_{v\in\{0,1\}^n} \left|\widetilde{h}(v)\right| \mathbb{E}_M\left[\sum_{\substack{w\in\{0,1\}^n\\ |v\oplus w|=2\ell}} \left|\widetilde{f}(w)\right|\right]\leq \left(\frac{10^9C\sqrt{s^* n}}{\ell}\right)^\ell,
\]
and by Lemma~\ref{lm:mass-transfer-high} we have for all $\ell\in [s^*, n/(2C^2)]$, using the fact that $|\A|/2^n=|\A_{reduced}|/2^{\alpha n}\geq 2^{-s^*}$ by assumption of the lemma,
\[
\sum_{v\in\{0,1\}^n} \left|\widetilde{h}(v)\right| \mathbb{E}_M\left[\sum_{\substack{w\in\{0,1\}^n\\ |v\oplus w|=2\ell}} \left|\widetilde{f}(w)\right|\right]\leq \left(\frac{(10^9C)^2 n}{\ell}\right)^\ell.
\]

Putting these bounds together and recalling the definition of $\bound[10^9 C](\ell)$ (see~\eqref{eq:bound-def}, as well as the definition restated above) we have for each $\ell\in [1:n/(2C^2)]$ and all sets $\A_{reduced}$
\[
\sum_{v\in\{0,1\}^n} \left|\widetilde{h}(v)\right| \mathbb{E}_M\left[\sum_{\substack{w\in\{0,1\}^n\\ |v\oplus w|=2\ell}} \left|\widetilde{f}(w)\right|\right]\leq \bound[10^9C](\ell).
\]
Thus, by Markov's inequality the probability over the choice of $M$ that  
\[
\sum_{v\in\{0,1\}^n} \left|\widetilde{h}(v)\right| \sum_{\substack{w\in\{0,1\}^n\\ |v\oplus w|=2\ell}} \left|\widetilde{f}(w)\right|\geq \bound[10^9C/\delta](\ell)
\]
for some $\ell>0$ and some $\A_{reduced}$ is at most
\[
\sum_{\ell=1}^{n/(2C^2)} \frac{\bound[10^9C](\ell)}{(1-\delta)\cdot\bound[10^9C/\delta](\ell)}=\frac{1}{1-\delta}\cdot \sum_{\ell=1}^{n/(2C^2)}\delta^\ell \leq \frac{\delta}{(1-\delta)^2}\leq 4\delta.
\] 
Recalling that $2^n/|\B'|\leq (2^n/|\B|)\cdot (2^n/|\A|) / (1-\delta)$ happens with probability at least $1-\delta$ we conclude that $\B'$ is $((10^9/\delta)C, s^*)$-bounded with probability at least $1-5\delta$.
\end{proofof}

\subsection{Bounding mass transfer to low weight Fourier coefficients}
The goal of this subsection is to prove Lemma~\ref{lm:mass-transfer} below. 

\paragraph{Proof outline.} The lemma starts with the assumption that a subset $\B$ of $\bool^n$ is $(C, s^*)$-bounded and upper bounds the $\ell_1$ norm of the convolution of the normalized Fourier transform $\wt{h}$ of $\B$ with the normalized Fourier transform $\wt{f}$ of a subset $\A$ of $\bool^n$ (which should be thought of as the typical message corresponding to a player who receives labels $MX$ for $X$ selected uniformly at random from $\B$) on coefficients with Hamming weight $2\ell$, for $\ell\in [1, s^*]$.

         The proof follows the outline presented in Section~\ref{sec:outline-app} (see Section~\ref{sec:analysis-overview}). Indeed, we first partition the pairs $v$ (a subset of the vertices, or a coefficient of $\wt{h}$), and $w$ (a subset of the edges of the matching that $\wt{f}$ is supported on) classes depending on the number of {\bf boundary}, {\bf internal} and {\bf external} edges of the matching $M$ involved in the mass transfer. See the proof below for the definition of these types of edges  $M$, Fig.~\ref{fig:int-ext-partial} for an illustration and Section~\ref{sec:analysis-overview} for an illustration of these notions on the simple example of the component growing protocol from Section~\ref{sec:outline-app}. We then reduce the problem of bounding the Fourier mass to the problem of verifying certain sums that reflect the tradeoffs between the amount of mass at various Fourier levels and basic combinatorics on matchings. A crucial parameter that governs these calculations is the probability, over the choice of a uniformly random matching, that this matching has a given number of internal and boundary edges with respect to a fixed subset of the vertices (i.e. a fixed Fourier coefficient). Convenient upper bounds on such quantities are provided by Lemma~\ref{lm:qkin} and Lemma~\ref{lm:qkibn}. Using these lemmas, we reduce the problem to verifying several combinatorial bounds (which at this point are disjoint from any Fourier analytic considerations), which is done in Lemma~\ref{lm:s0}, Lemma~\ref{lm:s1}, Lemma~\ref{lm:s2} and Lemma~\ref{lm:s3} from Section~\ref{sec:technical-low}

       Overall, the proof of Lemma~\ref{lm:mass-transfer} follows quite closely the outline presented in Section~\ref{sec:outline-app}, with the main that goes beyond this outline being the analysis of contribution to and from weights higher than $s^*$ (specifically, weight levels in $[s^*, n/(2C^2)]$ and $[n/(2C^2), n]$), as such weight levels carry zero Fourier mass for the simple component growing protocol from Section~\ref{sec:outline-app}. This analysis is provided by  Lemmas~\ref{lm:s0}, Lemma~\ref{lm:s1}, Lemma~\ref{lm:s2} and Lemma~\ref{lm:s3} from Section~\ref{sec:technical-low}. At the same time, we note that most of the final contribution to the Fourier transform of the final function comes from `low' weight levels, namely from $\ell\in [1, s^*]$ (see Lemma~\ref{lm:s1}), similarly to the component growing protocol in Section~\ref{sec:outline-app}.

\begin{lemma}[Mass transfer to low weights]\label{lm:mass-transfer}
For every $n, s^*$, $C>1$, $\alpha\in (0, 1)$ that satisfy conditions
\begin{align*}
{\bf (P1)} \, \alpha < 10^{-10}\quad 
{\bf (P2)} \, C>10^6 \quad 	
{\bf (P3)} \, s^*<\frac{n}{10^9 C^3} \quad
{\bf (P4)} \, n>10^9C^4,
\end{align*}
 if $\B\subseteq \bool^n$ is $(C, s^*)$-bounded and $M$ is a uniformly random matching  on $[n]$ of size $\alpha n$, the following conditions hold. For every $\A_{reduced}\subseteq \bool^M$, if $\A=\{x\in \bool^n: Mx\in \A_{reduced}\}$ and $f$ is the indicator of $\A$, if $|\A|/2^n\geq 2^{-s^*}$, then for every $\ell\in [1, s^*]$
\[
\sum_{v\in\{0,1\}^n} \left|\wt{h}(v)\right| \mathbb{E}_M\left[\sum_{\substack{z\in\{0,1\}^n\\ |v\oplus z|=2\ell}} \left|\widetilde{f}(z)\right|\right]\leq \left(\frac{10^9 C \sqrt{s^* n}}{\ell}\right)^\ell.
\]
\end{lemma}
\begin{proof}
We first note that $\wt{f}(z)\neq 0$ only if $z=Mw$ for $w\in \bool^M$, and thus for every $v\in \bool^n$ one has $\mathbb{E}_M\left[\sum_{\substack{z\in\{0,1\}^n\\ |v\oplus z|=2\ell}} \left|\widetilde{f}(z)\right|\right]=\mathbb{E}_M\left[\sum_{\substack{w\in\bool^M\\ |v\oplus Mw|=2\ell}} \left|\widetilde{f}(Mw)\right|\right]$. We also note that for every $v\in \bool^n$ such that $|v\oplus Mw|=2\ell$ one has  $|v|\leq |v\oplus Mw|+|Mw|\leq \ell+|Mw|\leq s^*+2\alpha n\leq n/200$ since $s^*\leq n/10^9$ and $\alpha<1/400$ by assumptions {\bf (P1)}, {\bf (P2)} and {\bf (P3)}. Thus, the only terms with a nonzero contribution to the sum that we need to bound are $v\in \bool^n$ with $|v|=2k\leq n/200$. Thus, it suffices to bound, for a parameter $k\in [0:n/100]$ and $v\in \bool^n$ with $|v|=2k$, the quantity
\begin{equation}\label{eq:sum-fixed-v}
\mathbb{E}_M\left[\sum_{\substack{w\in\bool^M\\ |v\oplus Mw|=2\ell}} \left|\widetilde{f}(Mw)\right|\right].
\end{equation}
We will later (see~\eqref{eq:9243hg9g} below) combine our bounds over all $k\in [0:n/100]$ to obtain the result of the lemma.

 Let $\Int=\{e_1^{\operatorname{int}}, e_2^{\operatorname{int}}, \dots\}$ be the set of edges $e=(a, b)\in M$ that match points of $v$, i.e. with $a, b\in v$. Let $\Bound=\{e_1^{\operatorname{bound}}, e_2^{\operatorname{bound}}, \dots \}$ be the set of boundary edges, i.e. edges $e=(a, b)\in M$ with $a, b\in v$. Let $\Ext =\{e_1^{\operatorname{ext}}, e_2^{\operatorname{ext}}, \dots\}$ be the set of external edges, i.e. edges $e=(a, b)\in M$ with $a, b\in [n]\setminus v$. 
 
 We decompose the sum~\eqref{eq:sum-fixed-v} according to the number of boundary edges in $Mw$:
\begin{equation}\label{eq:weihewihg}
\mathbb{E}_M\left[\sum_{\substack{w\in\{0,1\}^M\\ |v\oplus Mw|=2\ell}} \left|\widetilde{f}(Mw)\right|\right]=\mathbb{E}_M\left[
\sum_{S \subseteq \Bound} \sum_{w\in\{0,1\}^M}
\mathbbm{1}_{\{w\cap\Bound=S \}} \mathbbm{1}_{\{|v\oplus Mw|=2\ell\}} \left|\widetilde{f}(Mw)\right|\right].
\end{equation}
We now rewrite the latter indicator function. For a subset $S\subseteq \Bound$ define $w_S\in\{0,1\}^M$ as the set of all internal edges $\Int$ and all edges in $S$. We then have $|v\oplus Mw_S|=2k-2|\Int|$, since adding a boundary edge to $v$ does not change the Hamming weight, and adding an internal edge reduces it by $2$. Also note that $|w\oplus w_S|=\ell-(k-|\Int|)$. Indeed, $|v\oplus Mw_S|=2k-2\Int$, $|v\oplus w|=2\ell$,  $w$ can be obtained from $w_S$ be removing internal edges and adding external edges and both of these changes increase $|Mw\oplus Mw_s|$ by $2$.  These observations together with~\eqref{eq:weihewihg} yield the following upper bound on~\eqref{eq:sum-fixed-v}:
\begin{equation}\label{eq:093h4g94hg}
\begin{split}
\mathbb{E}_M\left[\sum_{\substack{w\in\{0,1\}^M\\ |v\oplus Mw|=2\ell}} \left|\widetilde{f}(Mw)\right|\right]&\leq 
\mathbb{E}_M\left[\sum_{S \subseteq \Bound} \sum_{\substack{w\in\{0,1\}^M\\ |w\oplus w_S|=\ell-(k-|\Int|) }}
\mathbbm{1}_{\{w_S=S\}} \cdot \left|\widetilde{f}(Mw)\right|\right]\\
&\leq 
\mathbb{E}_M\left[\sum_{S \subseteq \Bound} \sum_{\substack{w\in\{0,1\}^M\\ |w\oplus w_S|=\ell-(k-|\Int|) }}
 \left|\widetilde{f}(Mw)\right|\right],
\end{split}
\end{equation}
where we dropped the indicator function in going from the first line to the second line above. We now apply Lemma~\ref{lm:L1xorKKL} to the inner summation on the rhs of~\eqref{eq:093h4g94hg}. We now let $f'$ denote the indicator of $\A_{reduced}$ and note that $\wt{f'}(w)=\wt{f}(Mw)$ for all $w\in \bool^M$. Since $\A_{reduced}/2^{\alpha n}= |\A|/2^n \geq 2^{-s^*}$, this, in turn, can be bounded by Lemma \ref{lm:L1xorKKL} applied to the set $\A_{reduced}\subset \{0,1\}^M$ (i.e. $m=\alpha n$) with $q=\ell-k+\Int$ and $y=w_s$. We obtain
\begin{equation}\label{eq:xor-KKL-bound}
\begin{split}
\sum_{\substack{w\in\{0,1\}^{M}\\ |w\oplus w_S|=\ell-(k-|\Int|) }} \left|\widetilde{f}(Mw)\right|&=\sum_{\substack{w\in\{0,1\}^{M}\\ |w\oplus w_S|=\ell-(k-|\Int|) }} \left|\widetilde{f'}(w)\right|
\\&\leq 
\sqrt{\binom{\alpha n}{\ell-k+|\Int|} \left(\frac{64s^*}{\ell-k+|\Int|}\right)^{\ell-k+|\Int|}}
\\&\leq
\left(\frac{15\sqrt{s^*\alpha n}}{\ell-k+|\Int|}\right)^{\ell-k+|\Int|}.
\end{split}
\end{equation}

Summing over all possible subsets $S\subseteq \Bound$, we then infer
\begin{align*}
\mathbb{E}_M\left[\sum_{\substack{w\in\{0,1\}^M\\ |v\oplus Mw|=2\ell}} \left|\widetilde{f}(Mw)\right|\right]
&\leq \mathbb{E}_M\left[\sum_{S \subseteq \Bound} \left(\frac{15\sqrt{s^*\alpha n}}{\ell-k+|\Int|}\right)^{\ell-k+|\Int|} \right]\\
&\leq
\mathbb{E}_M\left[2^{|\Bound|} \left(\frac{15\sqrt{s^*\alpha n}}{\ell-k+|\Int|}\right)^{\ell-k+|\Int|} \right],
\end{align*}
where in going from line~1 to line~2 in the equation above we used the fact that the bound is independent of the set $S$ and upper bounded the summation by multiplying by the number of such sets $S$, i.e. $2^{|\Bound|}$.

   We now bound the sum on the last line above. We have 
\begin{align*}
&\mathbb{E}_M\left[2^{|\Bound|}\left(\frac{15\sqrt{s^*\alpha n}}{\ell-k+|\Int|}\right)^{\ell-k+|\Int|} \right]\\
&=\mathbb{E}_M\left[\sum_{i=0}^k \sum_{b=0}^{2k} 2^{b} \mathbbm{1}_{\{|\Bound|=b\text{~and}~|\Int|=i\}} \left(\frac{15\sqrt{s^*\alpha n}}{\ell-k+i}\right)^{\ell-k+i} \right]\\
&=\sum_{i=|k-\ell|_+}^k \sum_{b=0}^{2k} 2^{b} q(k,i,b,n)  \left(\frac{15\sqrt{s^*\alpha n}}{\ell-k+i}\right)^{\ell-k+i}
\end{align*}   
In going from line~2 to line~3 above we used the fact that by Definition~\ref{def:qkibn} for every $v\in \bool^n$ with $|v|=2k$ one has $\expect_M[\mathbbm{1}_{\{|\Bound|=b\text{~and}~|\Int|=i\}}]=q(k, i, b, n)$, i.e. $q(k,i,b,n)$ is the probability that a uniformly random matching $M$ of size $\alpha n$ is  such that $i$ edges of $M$ match points of $v$ (i.e. $\left|\Int\right|=i$) and $b$ edges of $M$ are boundary edges (i.e. $\left|\Bound\right|=b$).  Note that we sum over $b$ between $0$ and $2k$, as the number of boundary edges of $v$ with respect to $M$ is bounded by its Hamming weight $2k$. Similarly, the number of internal edges $i$ cannot be larger than $|v|/2=k$, and must be at least $k-\ell$ in order for the binomial coefficient ${\alpha n \choose \ell-k+i}$ on the first line above to be nonzero (combinatorially, this means that in order for $v+Mw$ to have weight $2\ell$ for some $w\in \bool^M$ one must have $|\Int|\geq k-\ell$).  Recall that $|k-\ell|_+=\max\{0, k-\ell\}$. 

We thus have that for a fixed $v\in \{0,1\}^n$ with $|v|=2k$
\[
\mathbb{E}_M\left[\sum_{\substack{w\in\{0,1\}^M\\ |v\oplus Mw|=2\ell}} \left|\widetilde{f}(Mw)\right|\right] 
\leq
\sum_{i=|k-\ell|_+}^{k}\sum_{b=0}^{2k} q(k,i,b,n) 2^{b}\left(\frac{15\sqrt{s^*\alpha n}}{\ell-k+i}\right)^{\ell-k+i} 
\]
Recall that by Lemma \ref{lm:qkibn} we have $q(k,i,b,n)\leq q(k,i,n)20^{-b}4^{k-i}$ (note that $k\leq n/100$ and $\alpha<1/100$ by assumption {\bf (P2)}, so the preconditions of the lemma are satisfied). We thus get, substituting into the rhs of the equation above,
\begin{equation}\label{eq:i24gu4bg}
\begin{split}
\mathbb{E}_M\left[\sum_{\substack{w\in\{0,1\}^M\\ |v\oplus Mw|=2\ell}} \left|\widetilde{f}(Mw)\right|\right] 
&\leq 
\sum_{i=|k-\ell|_+}^{k}\sum_{b=0}^{2k} 4^{k-i}q(k,i,n) 2^{-b}\left(\frac{15\sqrt{s^*\alpha n}}{\ell-k+i}\right)^{\ell-k+i}
\\&\leq 
2\sum_{i=|k-\ell|_+}^{k} 4^{k-i}q(k,i,n)\left(\frac{15\sqrt{s^*\alpha n}}{\ell-k+i}\right)^{\ell-k+i} 
\\&\leq 
2\sum_{i=|k-\ell|_+}^{k} 4^\ell q(k,i,n) \left(\frac{15\sqrt{s^* \alpha n}}{\ell-k+i}\right)^{\ell-k+i} 
\\&\leq 
8^\ell\sum_{i=|k-\ell|_+}^{k} q(k,i,n) \left(\frac{15\sqrt{s^* \alpha n}}{\ell-k+i}\right)^{\ell-k+i}.
\end{split}
\end{equation}

Equipped with the upper bound~\eqref{eq:i24gu4bg}, we now sum over possible $v$ of Hamming weight $k\in [0:n/100]$:
\begin{equation}\label{eq:9243hg9g}
\begin{split}
\sum_{v\in\{0,1\}^n} \left|\wt{h}(v)\right| &\mathbb{E}_M\left[\sum_{\substack{w\in\{0,1\}^M\\ |v\oplus Mw|=2\ell}} \left|\widetilde{f}(Mw)\right|\right] 
\\&\leq 
8^\ell\sum_{k=0}^{n/100}\left(\sum_{\substack{v\in\{0,1\}^n \\ |v|=2k}} \left|\wt{h}(v)\right|\right)\cdot \left(\sum_{i=|k-\ell|_+}^{k} q(k,i,n) \left(\frac{15\sqrt{s^* \alpha n} }{\ell-k+i}\right)^{\ell-k+i}\right)
\end{split}
\end{equation}
Recall that we want to bound the expression above for $\ell\leq s^*$. We now split the sum depending on how large $k$ is. We consider three intervals. For $k\in [0:100s^*]$ and $k\in [100s^*+1:n/C^2]$ we use the bound 
$$
\sum_{\substack{v\in\{0,1\}^n \\ |v|=2k}} \left|\wt{h}(v)\right|\leq \bound(k),
$$
but for $k\in [n/C^2+1:n/100]$ we use the bound that follows by Cauchy-Schwarz together with Parseval's inequality: sum of squares of all normalized Fourier coefficients is $2^n/|\B| \leq 2^{s^*}$ due to $(C, s^*)$-boundedness of $\B$, and so 
\[
\sum_{\substack{v\in \{0,1\}^n \\ |v|=2k}}\left|\wt{h}(v)\right|\leq \sqrt{\left(\sum_{\substack{v\in \{0,1\}^n \\ |v|=2k}}\left|\wt{h}(v)\right|^2\right)\cdot\binom{n}{2k}}\leq \sqrt{2^{s^*}\binom{n}{2k}}.
\]

Thus, it is sufficient to derive strong upper bounds on $S_0, S_1, S_2, S_3$ that we define below:
\begin{equation}
\label{def:s0}
S_0:=\sum_{k=0}\bound(k)\cdot \left(\sum_{i=|k-\ell|_+}^{k} q(k,i,n) \left(\frac{15\sqrt{s^* \alpha n} }{\ell-k+i}\right)^{\ell-k+i}\right)
\end{equation}
\begin{equation}
\label{def:s1}
S_1:=\sum_{k=1}^{100s^*}\bound(k)\cdot \left(\sum_{i=|k-\ell|_+}^{k} q(k,i,n) \left(\frac{15\sqrt{s^* \alpha n} }{\ell-k+i}\right)^{\ell-k+i}\right)
\end{equation}
\begin{equation}
\label{def:s2}
S_2:=\sum_{k=100 s^*+1}^{n/C^2}\bound(k)\cdot \left(\sum_{i=k-\ell}^{k} q(k,i,n) \left(\frac{15\sqrt{s^* \alpha n} }{\ell-k+i}\right)^{\ell-k+i}\right)
\end{equation}
\begin{equation}
\label{def:s3}
S_3:=\sum_{k=n/C^2+1}^{n/100}\sqrt{2^{s^*}\binom{n}{2k}}\cdot \left(\sum_{i=k-\ell}^{k} q(k,i,n) \left(\frac{15\sqrt{s^* \alpha n} }{\ell-k+i}\right)^{\ell-k+i}\right)
\end{equation}
Lemma~\ref{lm:s0}, Lemma~\ref{lm:s1}, Lemma~\ref{lm:s2} and Lemma~\ref{lm:s3} show that $S_0\leq \bound[15](\ell)$, $S_1\leq \bound[10^8C](\ell)$, $S_2\leq \bound[10^7C](\ell)$, and $S_3\leq 1$, which concludes the proof since
\begin{align*}
8^\ell\left(\bound[15](\ell)+\bound[10^8 C](\ell)+\bound[10^7C](\ell)+1\right)
&=
8^\ell\left((15/C)^\ell+10^{8\ell}+10^{7\ell}+1\right)\left(\frac{C\sqrt{s^* n}}{\ell}\right)^\ell 
\\&\leq
\left(\frac{10^9 C\sqrt{s^* n}}{\ell}\right)^\ell
\end{align*}
We note that the dominant contribution comes from $S_1$, i.e. from mass transfer from low weights to low weights, and the amount of mass transfer is consistent with what we would expect for the component growing protocol with the slightly increased communication budget of $\sqrt{s^* n}$ (as opposed to $s^*$).
\end{proof}

\subsubsection{Bounding $S_0, S_1, S_2, S_3$ (technical lemmas)}\label{sec:technical-low}
The first lemma bounds transfer of mass from weight zero to low weights:
\begin{lemma}[Mass transfer from weight zero to low weights]\label{lm:s0}
For every $n, s^*$, every $\ell\in [1:s^*]$, if parameters $C,\alpha,s^*, n$ satisfy 
\begin{align*}
{\bf (P1)} \, \alpha < 10^{-10}\quad 
{\bf (P2)} \, C>10^6 \quad 	
{\bf (P3)} \, s^*<\frac{n}{10^9 C^3} \quad
{\bf (P4)} \, n>10^9C^4,
\end{align*}
then
$$
\sum_{k=0}\bound(k)\cdot \left(\sum_{i=|k-\ell|_+}^{k} q(k,i,n) \left(\frac{15\sqrt{s^* \alpha n} }{\ell-k+i}\right)^{\ell-k+i}\right)\leq \bound[15](\ell).
$$
\end{lemma}
\begin{proof}
The sum only contains the summand corresponding to $i=k=0$; since $q(0,0,n)=1$, we get at most $\left(\frac{15\sqrt{s^* \alpha n} }{\ell}\right)^{\ell}\leq \bound[15](\ell)$.
\end{proof}
The lemma below bounds transfer of mass from low weights to low weights:
\begin{lemma}[Mass transfer from low weights to low weights]\label{lm:s1}
For every $n, s^*$, every $\ell\in [1:s^*]$, if parameters $C,\alpha,s^*, n$ satisfy 
\begin{align*}
{\bf (P1)} \, \alpha < 10^{-10}\quad 
{\bf (P2)} \, C>10^6 \quad 	
{\bf (P3)} \, s^*<\frac{n}{10^9 C^3} \quad
{\bf (P4)} \, n>10^9C^4,
\end{align*}
then
$$
\sum_{k=1}^{100s^*}\bound(k)\cdot \left(\sum_{i=|k-\ell|_+}^{k} q(k,i,n) \left(\frac{15\sqrt{s^* \alpha n} }{\ell-k+i}\right)^{\ell-k+i}\right)\leq \bound[10^8 C](\ell).
$$
\end{lemma}
\begin{remark}
We note that the contribution of `low' to `low' weights, bounded by Lemma~\ref{lm:s1} is the dominant one in the application in Lemma~\ref{lm:mass-transfer}. The calculations in this lemma are quite similar to those provided in the overview of the analysis for the component growing protocol in Section~\ref{sec:outline-app} (see Section~\ref{sec:evolution}).
\end{remark}

\begin{proofof}{Lemma~\ref{lm:s1}}
Note that for $k\leq 100 s^*$ we have by definition of $\bound(k)$ (see~\eqref{eq:bound-def})
\begin{align*}
\bound(k)
&\leq
\max\left\{\left(\frac{C\sqrt{s^* n}}{k}\right)^{k}, \left(\frac{C^2 n}{k}\right)^{k/2}\right\} 
\\ &= 
\max\left\{\left(\frac{C\sqrt{s^* n}}{k}\right)^{k}, \left(\frac{C \sqrt{s^* n}\cdot\sqrt{k/s^*}}{k}\right)^{k}\right\} 
\\ &\leq
\left(\frac{10C\sqrt{s^* n}}{k}\right)^{k}.
\end{align*}
We now define, as in~\eqref{def:s1},
$$
S_1:=\sum_{k=1}^{100s^*}\bound(k)\cdot \left(\sum_{i=|k-\ell|_+}^{k} q(k,i,n) \left(\frac{15\sqrt{s^* \alpha n} }{\ell-k+i}\right)^{\ell-k+i}\right)
$$
We now get, using the fact that $q(k, i, n)=\binom{\alpha n}{i}\binom{n-2\alpha n}{2(k-i)}\binom{n}{2k}^{-1}$ by Lemma~\ref{lm:qkin},
\begin{align*}
S_1 &\leq 
\sum_{k=1}^{100s^*}\left(\frac{10C\sqrt{s^* n}}{k}\right)^{k}\sum_{i=|k-\ell|_+}^{k} q(k,i,n)\left(\frac{15\sqrt{\alpha s^* n}}{\ell+i-k}\right)^{\ell+i-k} 
\\&=
\sum_{k=1}^{100s^*}\sum_{i=|k-\ell|_+}^{k} \left(\frac{10C\sqrt{s^* n}}{k}\right)^{k}\binom{\alpha n}{i}\binom{n-2\alpha n}{2(k-i)}\binom{n}{2k}^{-1}\left(\frac{15\sqrt{\alpha s^* n}}{\ell+i-k}\right)^{\ell+i-k} 
\\&\leq
\left(\frac{10C\sqrt{s^* n}}{\ell}\right)^{\ell} \sum_{k=1}^{100s^*}\sum_{i=|k-\ell|_+}^{k} (10C)^{k-\ell}\left(\sqrt{s^* n}\right)^{i} n^{i+2(k-i)-2k} \alpha^{i+(\ell+i-k)/2} \cdot \Gamma,
\end{align*}
where 
$$
\Gamma:=\frac{\ell^\ell (2k)!}{i!(2k-2i)!k^{k}(\ell+i-k)^{\ell+i-k}}2^{2k}15^{\ell+i-k}\leq 10^{3(\ell+k+i)}
$$
by Lemma~\ref{quotient of factorials}. Indeed, we first apply Lemma~\ref{quotient of factorials} with $m=2$, $a_1=\ell, a_2=2k$ to get $\ell^\ell \cdot 2k!\leq (2k+\ell)^{2k+\ell}$, and then apply Lemma~\ref{quotient of factorials} with $m=4$, $a_1=i, a_2=2k-2i, a_3=k, a_4=\ell+i-k$, obtaining $i!(2k-2i)!k^{k}(\ell+i-k)^{\ell+i-k}\geq ((2k+\ell)/12)^{2k+\ell}$. This yields an upper bound of $\Gamma\leq 12^{2k+\ell} \cdot 2^{2k}\cdot 15^{\ell+i-k}\leq 10^{3(\ell+k+i)}$. 

The exponent of $n$ is $-i$ which together with the $(\sqrt{s^* n})^{i}$ factor gives us $(s^*/n)^{i/2}$. So we have  
\begin{align*}
S_1\cdot \left(\frac{10C\sqrt{s^* n}}{\ell}\right)^{-\ell} 
&\leq
\sum_{k=1}^{100s^*}\sum_{i=|k-\ell|_+}^{k} 10^{i}C^{k-\ell} \left(s^*/n\right)^{i/2} 10^{3(\ell+k+i)}\alpha^{i+(\ell+i-k)/2}
\\&\leq
\sum_{k=1}^{100s^*}\sum_{i=|k-\ell|_+}^{k} 10^{i}C^{k-\ell-i} \left(C\sqrt{s^*/n}\right)^{i} 10^{6\ell+6i}\alpha^{i+(\ell+i-k)/2} 
\\&\leq
10^{6\ell}\sum_{i=0}^\infty\sum_{k=i}^{i+\ell} (\sqrt{\alpha}/C)^{\ell+i-k}\left(10^7 C\alpha \sqrt{s^*/n}\right)^{i}
\\&\leq
10^{6\ell}\sum_{i=0}^\infty\left(10^7 C\alpha \sqrt{s^*/n}\right)^{i} \sum_{k=i}^{i+\ell} (\sqrt{\alpha}/C)^{\ell+i-k}
\\  &\leq 
10^{6\ell} \sum_{i=0}^\infty 2\left(10^7 C\alpha \sqrt{s^*/n}\right)^{i}
\\&\leq 
10^{7\ell}
\end{align*}
In going from line~4 to line~5 we used the fact that $\sqrt{\alpha}/C<1/2$ by assumptions {\bf (P1)} and {\bf (P2)}. Similarly, in going from line~5 to line~6 we used the fact that $10^7 C\alpha \sqrt{s^*/n}<1/2$ by assumptions {\bf (P1), (P2), (P3)}. 
\end{proofof}

The next lemma bounds transfer of mass from intermediate weights to low weights:
\begin{lemma}[Mass transfer from intermediate to low weights]\label{lm:s2}
For every $n, s^*$, every $\ell\in [1:s^*]$, if $C,\alpha,s^*, n$ satisfy 
\begin{align*}
{\bf (P1)} \, \alpha < 10^{-10}\quad 
{\bf (P2)} \, C>10^6 \quad 	
{\bf (P3)} \, s^*<\frac{n}{10^9 C^3} \quad
{\bf (P4)} \, n>10^9C^4,
\end{align*}
then
$$
\sum_{k=100 s^*+1}^{n/C^2}\bound(k)\cdot \left(\sum_{i=k-\ell}^{k} q(k,i,n) \left(\frac{15\sqrt{s^* \alpha n} }{\ell-k+i}\right)^{\ell-k+i}\right)\leq \bound[10^7 C](\ell).
$$

\end{lemma}

\begin{remark}
A careful inspection of the proof of Lemma~\ref{lm:s2} reveals that the upper bound is in fact somewhat stronger than what Lemma~\ref{lm:mass-transfer} needs. This is mostly due to the fact that Lemma~\ref{lm:s2} analyzes contribution of weights $k$ starting at $100s^*$ to weights $\ell\leq s^*$. Since the strengthening is not consequential for the analysis, we prefer to keep the bound in the present form for simplicity. However, it is interesting to note that the `low weight' regime (i.e. the regime of Lemma~\ref{lm:s1}), for which the component growing protocol from Section~\ref{sec:outline-app} is a reasonable illustration, provides the dominant contribution to the Fourier spectrum. 
\end{remark}

\begin{proofof}{Lemma~\ref{lm:s1}}
Since in this sum we have $k>100s^*$ we must have $i\geq k-\ell>99s^*$ and also $i\geq k-\ell>k/2$. For $k\in [100s^*, n/C^2]$ we have by~\eqref{eq:bound-def}
\begin{equation}\label{eq:ib38h34n3gg}
\bound(k)\leq \left(\frac{C^2 n}{k}\right)^{k/2}=\left(\sqrt{\frac{k}{s^*}}\right)^{k}\left(\frac{C\sqrt{s^* n}}{k}\right)^{k}.
\end{equation}
We now define, as in~\eqref{def:s2},
$$
S_2:=\sum_{k=100 s^*+1}^{n/C^2}\bound(k)\cdot \left(\sum_{i=k-\ell}^{k} q(k,i,n) \left(\frac{15\sqrt{s^* \alpha n} }{\ell-k+i}\right)^{\ell-k+i}\right).
$$
Putting this together with~\eqref{eq:ib38h34n3gg}, we get
\begin{align*}
S_2 &\leq 
\sum_{k=100s^*+1}^{n/C^2}\left(\sqrt{\frac{k}{s^*}}\right)^{k}\left(\frac{C\sqrt{s^* n}}{k}\right)^{k}\sum_{i=k-\ell}^{k} q(k,i,n)\left(\frac{15\sqrt{\alpha s^* n}}{\ell+i-k}\right)^{\ell+i-k} 
\\&=
\sum_{k=100s^*+1}^{n/C^2}\sum_{i=k-\ell}^{k} \left(\sqrt{\frac{k}{s^*}}\right)^{k}\left(\frac{C\sqrt{s^* n}}{k}\right)^{k}\binom{\alpha n}{i}\binom{n-2\alpha n}{2(k-i)}\binom{n}{2k}^{-1}\left(\frac{15\sqrt{\alpha s^* n}}{\ell+i-k}\right)^{\ell+i-k} 
\\&\leq
\left(\frac{C\sqrt{s^* n}}{\ell}\right)^{\ell} \sum_{k=100s^*+1}^{n/C^2}\sum_{i=k-\ell}^{k} \left(\sqrt{\frac{k}{s^*}}\right)^{k} C^{k-\ell}\left(\sqrt{s^* n}\right)^{i} n^{i+2(k-i)-2k} \alpha^{i+(\ell+i-k)/2} \cdot \Gamma 
\end{align*}
where
$$
\Gamma:= \frac{\ell^\ell (2k)!}{i!(2k-2i)!k^{k}(\ell+i-k)^{\ell+i-k}}2^{2k}15^{\ell+i-k}\leq 10^{3(\ell+k+i)}
$$
by Lemma~\ref{quotient of factorials}. Indeed, we first apply Lemma~\ref{quotient of factorials} with $m=2$, $a_1=\ell, a_2=2k$ to get $\ell^\ell \cdot 2k!\leq (2k+\ell)^{2k+\ell}$, and then apply Lemma~\ref{quotient of factorials} with $m=4$, $a_1=i, a_2=2k-2i, a_3=k, a_4=\ell+i-k$, obtaining $i!(2k-2i)!k^{k}(\ell+i-k)^{\ell+i-k}\geq ((2k+\ell)/12)^{2k+\ell}$. This yields an upper bound of $\Gamma\leq 12^{2k+\ell} \cdot 2^{2k}\cdot 15^{\ell+i-k}\leq 10^{3(\ell+k+i)}$. 

The exponent of $n$ is $-i$ which together with the $(\sqrt{s^* n})^{i}$ factor gives us $(s^*/n)^{i/2}$. So we have
\begin{equation}\label{eq:in43hy}
\begin{split}
S_2\cdot \left(\frac{C\sqrt{s^* n}}{\ell}\right)^{-\ell} 
&\leq
\sum_{k=100s^*+1}^{n/C^2}\sum_{i=k-\ell}^{k}\left(\sqrt{\frac{k}{s^*}}\right)^{k} C^{k-\ell} \left(\sqrt{s^*/n}\right)^{i} 10^{3(\ell+k+i)}\alpha^{i+(\ell+i-k)/2}
\\&\leq
\sum_{k=100s^*+1}^{n/C^2}\sum_{i=k-\ell}^{k} \left(\sqrt{\frac{k}{s^*}}\right)^{k} C^{k-\ell-i} \left(C\sqrt{s^*/n}\right)^{i} 10^{6\ell+6i}\alpha^{i+(\ell+i-k)/2} 
\\&\leq
10^{6\ell}\sum_{k=100s^*+1}^{n/C^2}\sum_{i=k-\ell}^{k} \left(\sqrt{\frac{k}{s^*}}\right)^{k} (C/\sqrt{\alpha})^{k-\ell-i}\left(10^4 C\alpha \sqrt{s^*/n}\right)^{i}
\\&\leq 
10^{6\ell}\sum_{k=100s^*+1}^{n/C^2}\sum_{i=k-\ell}^{k} \left(10^4 C\alpha\sqrt{k/n}\right)^{i}\left(\sqrt{\frac{k}{s^*}}\right)^{k-i}
\end{split}
\end{equation}
We used the fact that $C/\sqrt{\alpha}>1$ and $k-\ell-i\leq 0$ for $i$ in the range of summation in going from line~3 to line~4 above.
We now note that $10^4 C\alpha\sqrt{k/n}\leq 10^4 \alpha$ since $k\leq n/C^2$ in the range of summation, and $\left(\sqrt{\frac{k}{s^*}}\right)^{k-i}\leq \left(\sqrt{\frac{k}{s^*}}\right)^\ell\leq \left(\sqrt{\frac{k}{s^*}}\right)^{s^*}$ for $i$ in the range of summation, since $\ell\leq s^*$ by assumption of the lemma. We thus get that, substituting these bounds into~\eqref{eq:in43hy}, 

\begin{align*}
S_2\cdot \left(\frac{C\sqrt{s^* n}}{\ell}\right)^{-\ell} 
&\leq 
10^{6\ell}\sum_{k=100s^*+1}^{n/C^2}\sum_{i=k-\ell}^{k} \left(10^4\alpha\right)^{i}\left(\sqrt{\frac{k}{s^*}}\right)^{s^*}
\\&\leq 
10^{6\ell}\sum_{k=100s^*+1}^{n/C^2}\sum_{i=k-\ell}^{k} \left(10^4\alpha\right)^{i}e^{k/2e}
\\&\leq 
10^{6\ell}\sum_{i=99s^*+1}^{n/C^2}\sum_{k=i}^{i+\ell} \left(10^4\alpha\right)^{i}e^{k/2e}
\\{\bf (P1)}&\leq 
10^{6\ell} \sum_{i=99s^*+1}^{n/C^2}\sum_{k=i}^{i+\ell} 10^{-i}
\\&\leq  
10^{6\ell}\sum_{i=99s^*+1}^{n/C^2}2\ell \cdot 10^{-i}
\\&\leq
10^{7\ell}
\end{align*}
We used the fact that $(\sqrt{k/s^*})^s=(k/s^*)^{s^*/2}\leq e^{k/(2e)}$ in going from line~1 to line~2 above (this follows since $(k/s^*)^{s^*}$ is maximized when $s^*=k/e$). The transition from line~3 to line~4 follows since $i\geq k-\ell\geq k/2$ (as $\ell\leq s^*$ and $k\geq 100s^*$), and hence 
$$
\left(10^4\alpha\right)^{i}e^{k/2e}\leq \left(10^4\alpha\right)^{i}e^{i/e}=\left(10^4e^{1/e}\alpha\right)^{i}\leq 10^{-i},
$$
since $\alpha<10^{-6}$ by assumption {\bf (P1)}. The transition from line~5 to line~6 follows since $\ell\leq s^*$ by assumption.
\end{proofof}
\begin{lemma}[Mass transfer from high weights to low weights]\label{lm:s3}
For every $n, s^*$, every $\ell\in [1:s^*]$, if parameters $C,\alpha,s^*,n$ satisfy 
\begin{align*}
{\bf (P1)} \, \alpha < 10^{-10}\quad 
{\bf (P2)} \, C>10^6 \quad 	
{\bf (P3)} \, s^*<\frac{n}{10^9 C^3} \quad
{\bf (P4)} \, n>10^9C^4,
\end{align*}
\[
\sum_{k=n/C^2}^{n/100}\sqrt{2^{s^*}\binom{n}{2k}}\cdot \left(\sum_{i=k-\ell}^{k} q(k,i,n) \left(\frac{15\sqrt{s^* \alpha n} }{\ell-k+i}\right)^{\ell-k+i}\right)\leq 1.
\]
\end{lemma}
\begin{remark}
We note that the bound of Lemma~\ref{lm:s3} is much stronger than what Lemma~\ref{lm:mass-transfer} needs. Indeed, a much weaker bound of $\approx (\sqrt{s^* n}/\ell)^\ell$ would have been sufficient for our purposes. The contribution of high weights  to low weights is significantly lower than the dominant terms ($\ell \in [1:s^*]$ and $\ell\in [s^*, n/(2C^2)]$, see Lemma~\ref{lm:s1} and Lemma~\ref{lm:s2}) provide, and hence we choose (rather arbitrarily) to prove the upper bound of $1$.
\end{remark}

\begin{proofof}{Lemma~\ref{lm:s3}}
We first define, as in~\eqref{def:s3},
$$
S_3:=\sum_{k=n/C^2+1}^{n/100}\bound(k)\cdot \left(\sum_{i=k-\ell}^{k} q(k,i,n) \left(\frac{15\sqrt{s^* \alpha n} }{\ell-k+i}\right)^{\ell-k+i}\right).
$$
This allows us to write 
\begin{align*}
S_3 &\leq 
\sum_{k=n/C^2+1}^{n/10}\sum_{i=k-\ell}^{k} \sqrt{2^{s^*}\binom{n}{2k}}\cdot q(k,i,n)\left(\frac{15\sqrt{\alpha s^* n}}{\ell+i-k}\right)^{\ell+i-k}
\\ &\leq
\sum_{k=n/C^2+1}^{n/10}\sum_{i=k-\ell}^{k} 4^{s^*}\binom{n}{2k}^{1/2}\binom{\alpha n}{i}\binom{n-2\alpha n}{2(k-i)}\binom{n}{2k}^{-1}\binom{n}{\ell+i-k}
\\ &\leq
\sum_{k=n/C^2+1}^{n/10}\sum_{i=k-\ell}^{k} 4^{s^*}\binom{\alpha n}{i}\binom{n-2\alpha n}{2(k-i)}\binom{n}{2k}^{-1/2}\binom{n}{\ell+i-k}
\end{align*}

The transition from line~1 to line~2 follows since $15\sqrt{\alpha s^* n}\leq n$ and $k\leq n/10\leq n/2$ in the range of summation. Since $\ell+i-k\leq \ell\leq s^*$ and $s^*\leq n/4$ by assumptions {\bf (P2)} and {\bf (P3)}, we have $\binom{n}{\ell+i-k}\leq \binom{n}{s^*}$. At the same time, since $k-i\leq \ell\leq s^*$, we also have $\binom{n-2\alpha n}{2(k-i)}\leq \binom{n}{2s^*}$.
Using the bounds above as well as the fact that $(a/b)^b\leq {a \choose b}\leq (ea/b)^b$ for all integer $a, b\geq 0$, we get
\begin{align*}
S_3&\leq 
\sum_{k=n/C^2+1}^{n/10}\sum_{i=k-\ell}^{k} 4^{s^*}\left(\frac{e\alpha n}{i}\right)^{i}\left(\frac{n}{2k}\right)^{-k}\binom{n}{s^*}\binom{n}{2s^*}
\\&\leq 
\sum_{k=n/C^2+1}^{n/10}\sum_{i=k-\ell}^{k} 4^{i}\left(\frac{e\alpha n}{i}\right)^{i}\left(\frac{n}{2i}\right)^{-i}\binom{n}{s^*}^3\text{~~~~~~~~~~~~~~(since $i\geq s^*$)}
\\&\leq
\sum_{k=n/C^2+1}^{n/10}\sum_{i=k-\ell}^{k} (22\alpha)^{i}e^{3 H(s^*/n)\cdot n}\text{~~~~~~(since $\binom{n}{s^*}\leq e^{H(s^*/n)n}$ by Lemma~\ref{lm:binom-entropy})}
\\&\leq
\sum_{k=n/C^2+1}^{n/10} 2(22\alpha)^{k-\ell}e^{3 H(s^*/n)\cdot n}\text{~~~~~~~~(summing geometric series)}
\\&\leq
4(22\alpha)^{n/C^2-s^*}e^{3 H(s^*/n)\cdot n}\text{~~~~~~~~(summing geometric series and using $\ell\leq s^*$)}
\\{\bf (P3)}&\leq
4(22\alpha)^{n/(2C^2)}e^{3 H(1/C^3)\cdot n}\text{~~~~~~~~~~~(since $\ell\leq s^*\leq n/C^3$)}
\\{\bf (P1, P2)}&\leq
n^2 e^{-n/(4C^2)}
\\{\bf (P4)}&\leq
1
\end{align*}
\end{proofof}

\subsection{Bounding mass transfer to intermediate weight Fourier coefficients}
The goal of this subsection is to prove Lemma~\ref{lm:mass-transfer-high} below.

\paragraph{Proof outline.} The proof of the lemma is similar to that of Lemma~\ref{lm:mass-transfer}, with different supporting combinatorial lemmas (Lemma~\ref{lm:t1} and Lemma~\ref{lm:t2} from Section~\ref{sec:technical-high}).  Unlike the proof of Lemma~\ref{lm:mass-transfer}, which mostly deals with low weight coefficients, Lemma~\ref{lm:mass-transfer-high} analyzes high weight coefficients, which are zero for the simple component growing protocol from Section~\ref{sec:outline-app}.

\begin{lemma}\label{lm:mass-transfer-high}
For every $n, s^*, C, \alpha\in (0, 1)$ that satisfy conditions
\begin{align*}
{\bf (P1)} \, \alpha < 10^{-10}\quad 
{\bf (P2)} \, C>10^6 \quad 	
{\bf (P3)} \, s^*<\frac{n}{10^9 C^3} \quad
{\bf (P4)} \, n>10^9C^4,
\end{align*}
 if $\B\subseteq \bool^n$ is $(C, s^*)$-bounded and $M$ is a uniformly random matching  on $[n]$ of size $\alpha n$, the following conditions hold. For every $\A_{reduced}\subseteq \bool^M$, if $\A=\{x\in \bool^n: Mx\in \A_{reduced}\}$ and $f$ is the indicator of $\A$, if $|\A|/2^n\geq 2^{-s^*}$, then for all $\ell\in [s^*, n/(2C^2)]$ we have 
\[
\sum_{v\in\{0,1\}^n} \left|\wt{h}(v)\right| \mathbb{E}_M\left[\sum_{\substack{w\in\{0,1\}^n\\ |v\oplus w|=2\ell}} \left|\widetilde{f}(w)\right|\right]\leq \left(\frac{(10^9 C)^2 n}{\ell}\right)^{\ell/2}.
\]
\end{lemma}
\begin{proof}
We start similarly to the proof of Lemma \ref{lm:mass-transfer}. 

We first note that $\wt{f}(z)\neq 0$ only if $z=Mw$ for $w\in \bool^M$, and thus for every $v\in \bool^n$ one has $\mathbb{E}_M\left[\sum_{\substack{w\in\{0,1\}^n\\ |v\oplus w|=2\ell}} \left|\widetilde{f}(w)\right|\right]=\mathbb{E}_M\left[\sum_{\substack{z\in\bool^M\\ |v\oplus Mw|=2\ell}} \left|\widetilde{f}(Mw)\right|\right]$. We also note that for any $w\in \bool^M$ one has $|v|\leq |v\oplus Mw|+|Mw|\leq \ell+\alpha n\leq n/(2C^2)+\alpha n\leq n/200$ since $n/(2C^2)\leq n/400$ and $\alpha<1/400$ by assumptions {\bf (P1)}, {\bf (P2)} and {\bf (P3)}. Thus, the only terms with a nonzero contribution to the sum that we need to bound are $v\in \bool^n$ with $|v|=2k\leq n/200$. Thus, it suffices to bound, for a parameter $k\in [0:n/100]$ and $v\in \bool^n$ with $|v|=2k$, the quantity
\begin{equation}\label{eq:sum-fixed-v-v2}
\mathbb{E}_M\left[\sum_{\substack{w\in\bool^n\\ |v\oplus Mw|=2\ell}} \left|\widetilde{f}(Mw)\right|\right].
\end{equation}
We will later (see~\eqref{eq:sum} below) combine our bounds over all $k\in [0:n/100]$ to obtain the result of the lemma.

 Let $\Int=\{e_1^{\operatorname{int}}, e_2^{\operatorname{int}}, \dots\}$ be the set of edges $e=(a, b)\in M$ that match points of $v$, i.e. with $a, b\in v$. Let $\Bound=\{e_1^{\operatorname{bound}}, e_2^{\operatorname{bound}}, \dots \}$ be the set of boundary edges, i.e. edges $e=(a, b)\in M$ with $a, b\in v$. Let $\Ext =\{e_1^{\operatorname{ext}}, e_2^{\operatorname{ext}}, \dots\}$ be the set of external edges, i.e. edges $e=(a, b)\in M$ with $a, b\in [n]\setminus v$. 
 
 We decompose the sum~\eqref{eq:sum-fixed-v-v2} according to the number of boundary edges in $Mw$:
\begin{equation}\label{eq:weihewihg-v2}
\mathbb{E}_M\left[\sum_{\substack{w\in\{0,1\}^M\\ |v\oplus Mw|=2\ell}} \left|\widetilde{f}(Mw)\right|\right]=\mathbb{E}_M\left[
\sum_{S \subseteq \Bound} \sum_{w\in\{0,1\}^M}
\mathbbm{1}_{\{w\cap\Bound=S \}} \mathbbm{1}_{\{|v\oplus Mw|=2\ell\}} \left|\widetilde{f}(Mw)\right|\right].
\end{equation}
We now rewrite the latter indicator function. For a subset $S\subseteq \Bound$ define $w_S\in\{0,1\}^M$ as the set of all internal edges $\Int$ and all edges in $S$. We then have $|v\oplus Mw_S|=2k-2|\Int|$, since adding a boundary edge to $v$ does not change the Hamming weight, and adding an internal edge reduces it by $2$. Also note that $|w\oplus w_S|=\ell-(k-|\Int|)$. Indeed, $|v\oplus Mw_S|=2k-2\Int$, $|v\oplus w|=2\ell$,  $w$ can be obtained from $w_S$ be removing internal edges and adding external edges and both of these changes increase $|Mw\oplus Mw_s|$ by $2$.  These observations together with~\eqref{eq:weihewihg-v2} yield the following upper bound on~\eqref{eq:sum-fixed-v-v2}:
\begin{equation}\label{eq:093h4g94hg-v2}
\begin{split}
\mathbb{E}_M\left[\sum_{\substack{w\in\{0,1\}^M\\ |v\oplus Mw|=2\ell}} \left|\widetilde{f}(Mw)\right|\right]&\leq 
\mathbb{E}_M\left[\sum_{S \subseteq \Bound} \sum_{\substack{w\in\{0,1\}^M\\ |w\oplus w_S|=\ell-(k-|\Int|) }}
\mathbbm{1}_{\{w_S=S\}} \cdot \left|\widetilde{f}(Mw)\right|\right]\\
&\leq 
\mathbb{E}_M\left[\sum_{S \subseteq \Bound} \sum_{\substack{w\in\{0,1\}^M\\ |w\oplus w_S|=\ell-(k-|\Int|) }}
 \left|\widetilde{f}(Mw)\right|\right]
\end{split}
\end{equation}
where we dropped the indicator function in going from the first line to the second line above. This in turn can be bounded by Cauchy-Schwarz (Lemma~\ref{ineq:CS}) given that the sum of squares of all normalized Fourier ciefficient is $2^{\alpha n}/\A_{reduced}= 2^n/|\A| \leq 2^{s^*}$ by assumption:
\[
\sum_{\substack{w\in\{0,1\}^M\\ |w\oplus w_S|=\ell-(k-|\Int|) }} \left|\widetilde{f}(Mw)\right|\leq \sqrt{2^{s^*}\binom{\alpha n}{\ell-k+|\Int|}}.
\]
Summing over all possible subsets $S\subset \Bound$ we then infer
\[
\mathbb{E}_M\left[\sum_{\substack{w\in\{0,1\}^M\\ |v\oplus Mw|=2\ell}} \left|\widetilde{f}(Mw)\right|\right]\leq
\mathbb{E}_M\left[\sum_{S \subseteq \Bound} 
\sqrt{2^{s^*}\binom{\alpha n}{\ell-k+|\Int|}} \right].
\]
Note that the bound is independent of the set $S$ so we can replace the summation with the multiplication by $2^{|\Bound|}$ to obtain
\[
\mathbb{E}_M\left[\sum_{\substack{w\in\{0,1\}^M\\ |v\oplus Mw|=2\ell}} \left|\widetilde{f}(Mw)\right|\right]\leq
\mathbb{E}_M\left[2^{|\Bound|}
\sqrt{2^{s^*}\binom{\alpha n}{\ell-k+|\Int|}} \right].
\]

   We now bound the sum on the last line above. We have 
\begin{equation}\label{eq:iwehgihg}
\begin{split}
&\mathbb{E}_M\left[2^{|\Bound|}\sqrt{\binom{\alpha n}{\ell-k+|\Int|} \left(\frac{64s^*}{\ell-k+|\Int|}\right)^{\ell-k+|\Int|}} \right]\\
&=\mathbb{E}_M\left[\sum_{i=0}^k \sum_{b=0}^{2k} 2^{b} \mathbbm{1}_{\{|\Bound|=b\text{~and}~|\Int|=i\}} \sqrt{2^{s^*}\binom{\alpha n}{\ell-k+i}}  \right]\\
&=\sum_{i=|k-\ell|_+}^k \sum_{b=0}^{2k} 2^{b} q(k,i,b,n)  \sqrt{2^{s^*}\binom{\alpha n}{\ell-k+i}} 
\end{split}
\end{equation}
In going from line~2 to line~3 above we used the fact that by Definition~\ref{def:qkibn} for every $v\in \bool^n$ with $|v|=2k$ one has $\expect_M[\mathbbm{1}_{\{|\Bound|=b\text{~and}~|\Int|=i\}}]=q(k, i, b, n)$, i.e. $q(k,i,b,n)$ is the probability that a uniformly random matching $M$ of size $\alpha n$ is  such that $i$ edges of $M$ match points of $v$ (i.e. $\left|\Int\right|=i$) and $b$ edges of $M$ are boundary edges (i.e. $\left|\Bound\right|=b$).  Note that we sum over $b$ between $0$ and $2k$, as the number of boundary edges of $v$ with respect to $M$ is bounded by its Hamming weight $2k$. Similarly, the number of internal edges $i$ cannot be larger than $|v|/2=k$, and must be at least $k-\ell$ in order for the binomial coefficient ${\alpha n \choose \ell-k+i}$ on the first line above to be nonzero (combinatorially, this means that in order for $v+Mw$ to have weight $2\ell$ for some $w\in \bool^M$ one must have $|\Int|\geq k-\ell$).  Recall that $|k-\ell|_+=\max\{0, k-\ell\}$. 

We thus get by~\eqref{eq:iwehgihg}, using the bound $q(k,i,b,n)\leq  q(k,i,n) 4^{k-i} 20^{-b}$ (Lemma~\ref{lm:qkin})

\begin{equation}\label{eq:2h8g8h4g}
\begin{split}
\mathbb{E}_M\left[\sum_{\substack{w\in\{0,1\}^M\\ |v\oplus Mw|=2\ell}} \left|\widetilde{f}(Mw)\right|\right] 
&\leq
\sum_{i=|k-\ell|_+}^{k}\sum_{b=0}^{2k} q(k,i,b,n)2^{b} \binom{\alpha n}{\ell-k+i}^{1/2}2^{s^*/2} 
\\&\leq 
2^{s^*/2}\sum_{i=|k-\ell|_+}^{k}\sum_{b=0}^{2k} 4^{k-i}q(k,i,n) 2^{-b}\binom{\alpha n}{\ell-k+i}^{1/2} 
\\&\leq 
2^{s^*}\sum_{i=|k-\ell|_+}^{k} 4^{k-i}q(k,i,n) \left(\frac{3\alpha n }{\ell-k+i}\right)^{(\ell-k+i)/2} 
\\&\leq 
16^{\ell}\sum_{i=|k-\ell|_+}^{k} q(k,i,n) \left(\frac{3\alpha n }{\ell-k+i}\right)^{(\ell-k+i)/2}.
\end{split}
\end{equation}
We used the bound ${a \choose b}\leq (e a/b)^b\leq (3a/b)^b$ in going from line~2 to line~3, and the fact that $s^*\leq \ell$ by assumption and $k-i\leq k-|k-\ell|_+\leq \ell$ for $i$ in the range of the summation in going from line~3 to line~4.

Summing~\eqref{eq:2h8g8h4g} over all possible $v$, we get 
\begin{equation}\label{eq:sum}
\begin{split}
\sum_{v\in\{0,1\}^n} \left|\wt{h}(v)\right| &\mathbb{E}_M\left[\sum_{\substack{w\in\{0,1\}^M\\ |v\oplus Mw|=2\ell}} \left|\widetilde{f}(Mw)\right|\right] \\&\leq 16^\ell\sum_{k=0}^{n/100}\left(\sum_{\substack{v\in\{0,1\}^n \\ |v|=2k}} \left|\wt{h}(v)\right|\right)\cdot \left(\sum_{i=|k-\ell|_+}^{k} q(k,i,n) \left(\frac{3\alpha n }{\ell-k+i}\right)^{(\ell-k+i)/2}\right)\\
\\&=16^\ell\sum_{k=0}^{n/100}\left(\sum_{\substack{v\in\{0,1\}^n \\ |v|=2k}} \left|\wt{h}(v)\right|\right)\cdot \left(\sum_{i=|k-\ell|_+}^{k} q(k,i,n) \left(\frac{3\alpha n }{\ell-k+i}\right)^{(\ell-k+i)/2}\right),
\end{split}
\end{equation}
where we used the fact that $k\leq n/100$ by assumptions of the lemma that we established earlier.

We now split the sum into two parts and derive upper bounds on $\sum_{\substack{v\in\{0,1\}^n \\ |v|=2k}} \left|\wt{h}(v)\right|$ in the two regimes. 

\noindent{\bf Upper bounding $\sum_{\substack{v\in\{0,1\}^n \\ |v|=2k}} \left|\wt{h}(v)\right|$ for $k\in [1:n/C^2]$.}  We show that for every $k\in [1:n/C^2]$
\begin{equation}\label{eq:mid-ub}
\sum_{\substack{v\in \{0,1\} \\ |v|=2k}} \left|\wt{h}(v)\right|\leq 2^{s^*}\left(\frac{C^2 n}{k}\right)^{k/2}.
\end{equation}
Indeed, for $k\in [s^*:n/C^2]$ one this follows directly from the assumption that $h$ is $(C, s^*)$-bounded, whereas for $k\in [0:s^*]$ this follows by
\begin{align*}
\sum_{\substack{v\in \{0,1\} \\ |v|=2k}} \left|\wt{h}(v)\right|&\leq \left(\frac{C \sqrt{s^* n}}{k}\right)^{k}\text{~~~~~~~~(since $h$ is $(C, s^*)$-bounded)}\\
&=\left(\frac{C^2 n}{k}\right)^{k/2}\left(\frac{s^*}{k}\right)^{k/2}\\
&\leq \left(\frac{C^2 n}{k}\right)^{k/2} e^{\frac{s^*}{2e}},
\end{align*}
and the claim follows as $e^{1/2e}<2$.

\noindent{\bf Upper bounding $\sum_{\substack{v\in\{0,1\}^n \\ |v|=2k}} \left|\wt{h}(v)\right|$  for $k\in [n/C^2+1:n/100]$.}  We have, using the assumption that $|\B|/2^n\geq 2^{-s^*}$, Parseval's equality and Cauchy-Schwarz, that 
\begin{equation}\label{eq:high-ub}
\sum_{\substack{v\in \{0,1\}^n \\ |v|=2k}}\left|\wt{h}(v)\right|\leq \sqrt{\left(\sum_{\substack{v\in \{0,1\}^n \\ |v|=2k}}\left|\wt{h}(v)\right|^2\right)\cdot\binom{n}{2k}}\leq \sqrt{2^{s^*}\binom{n}{2k}}.
\end{equation}

\paragraph{Putting it together.} We bound the contribution of $k\in [1:n/C^2]$ and $k\in [n/C^2+1:n/100]$ separately.

First, substituting~\eqref{eq:mid-ub} into~\eqref{eq:sum} and restricting the summation to $k\in [1:n/C^2]$, we get
\begin{equation}\label{eq:t1}
\begin{split}
\sum_{k=1}^{n/C^2}\sum_{\substack{v\in \{0,1\}^n \\ |v|=2k}}\left|\wt{h}(v)\right|&\leq \sum_{k=1}^{n/C^2} 2^{s^*}\left(\frac{C^2 n}{k}\right)^{k/2}\cdot \left(\sum_{i=|k-\ell|_+}^{k} q(k,i,n) \left(\frac{3\alpha n }{\ell-k+i}\right)^{(\ell-k+i)/2}\right)\\
&:=T_1
\end{split}
\end{equation}

Similarly, we get, substituting~\eqref{eq:high-ub} into~\eqref{eq:sum} and restricting the summation to $k\in [n/C^2+1: n/100]$, 
\begin{equation}\label{eq:t2}
\begin{split}
\sum_{k=n/C^2+1}^{n/100}\sum_{\substack{v\in \{0,1\}^n \\ |v|=2k}}\left|\wt{h}(v)\right|&\leq \sum_{k=n/C^2+1}^{n/10}\sqrt{2^{s^*}\binom{n}{2k}}\cdot \left(\sum_{i=k-\ell}^{k} q(k,i,n) \left(\frac{3\alpha n }{\ell-k+i}\right)^{(\ell-k+i)/2}\right)\\
&=:T_2.
\end{split}
\end{equation}

Lemma~\ref{lm:t1} and Lemma~\ref{lm:t2} below show that $T_1\leq \bound[10^7C](\ell)$, $T_2\leq \bound[10^7C](\ell)$ which concludes the proof since in the double sum \eqref{eq:sum} the only summand corresponding to $k=0$ is $q(0,0,n)\left(\frac{3\alpha n}{\ell}\right)^{\ell/2}$ which is at most $\bound[\sqrt{3}](\ell)$ since $q(0,0,n)=1$; and we clearly have 
\begin{align*}
16^\ell\left(\bound[\sqrt{3}](\ell)+\bound[10^7C](\ell)+\bound[10^7C](\ell)\right)
&=
16^\ell\left((\sqrt{3}/C)^\ell+10^{8\ell}+10^{7\ell}\right)\left(\frac{C^2 n}{\ell}\right)^{\ell/2} 
\\&\leq
\left(\frac{\left(10^9 C\right)^2 n}{\ell}\right)^{\ell/2}.
\end{align*}
\end{proof}

\subsubsection{Bounding $T_1$ and $T_2$ (technical lemmas)}\label{sec:technical-high}
The following two lemmas deal with  sum $T_1$ and $T_2$ defined in the proof of Lemma~\ref{lm:mass-transfer-high}. 
\begin{lemma}\label{lm:t1}
Suppose parameters $C,\alpha,s^*,n$ satisfy 
\begin{align*}
{\bf (P1)} \, \alpha < 10^{-10}\quad 
{\bf (P2)} \, C>10^6 \quad 	
{\bf (P3)} \, s^*<\frac{n}{10^9 C^3} \quad
{\bf (P4)} \, n>10^9C^4,
\end{align*}
then for every $\ell\in [s^*, n/(2C^2)]$
$$
\sum_{k=1}^{n/C^2} 2^{s^*}\left(\frac{C^2 n}{k}\right)^{k/2}\cdot \left(\sum_{i=|k-\ell|_+}^{k} q(k,i,n) \left(\frac{3\alpha n }{\ell-k+i}\right)^{(\ell-k+i)/2}\right)\leq \bound[10^7C](\ell)
$$
\end{lemma}
\begin{remark}
We note that Lemma~\ref{lm:t1} provides the dominant contribution in the application in Lemma~\ref{lm:mass-transfer-high} (which is natural, since it bounds the transfer from a range of weights that overlaps with the target range of coefficient weights). 
\end{remark}

\begin{proofof}{Lemma~\ref{lm:t1}}
Define, similarly to~\eqref{eq:t1},
\begin{align*}
T_1:=&\sum_{k=1}^{n/C^2}2^{s^*}\left(\frac{C^2 n}{k}\right)^{k/2}\cdot \left(\sum_{i=|k-\ell|_+}^{k} q(k,i,n)\left(\frac{3\alpha n}{\ell-k+i}\right)^{(\ell-k+i)/2}\right)\\
&\leq 
2^\ell\sum_{k=1}^{n/C^2}\sum_{i=|k-\ell|_+}^{k} \left(\frac{C^2 n}{k}\right)^{k/2}\binom{\alpha n}{i}\binom{n-2\alpha n}{2(k-i)}\binom{n}{2k}^{-1}\left(\frac{3\alpha n}{\ell+i-k}\right)^{(\ell+i-k)/2},
\end{align*}
where in going from line~1 to line~2 we used the fact that $\ell\geq s^*$ by assumption of the lemma, as well as the fact that by Lemma~\ref{lm:qkin} for every $0\leq i\leq k\leq n/2$ one has
\[
q(k,i,n)=\binom{\alpha n}{i}\binom{n-2\alpha n}{2(k-i)}\binom{n}{2k}^{-1}.
\]

We now get, using the fact that $\binom{\alpha n}{i}\leq (\alpha n)^i/i!$, $\binom{n-2\alpha n}{2(k-i)}\leq n^{2(k-i)}/(2(k-i))!$ and $\binom{n}{2k}^{-1}\geq 2^{2k} n^{2k}/(2k)!$ for $k\leq n/4$ (which is satisfied since $k\leq n/100$ in our setting), 
\begin{equation}\label{eq:9023hgg}
\begin{split}
T_1&\leq 
\left(\frac{4C^2 n}{\ell/2}\right)^{\ell/2} \sum_{k=1}^{n/C^2}\sum_{i=|k-\ell|_+}^{k} C^{k-\ell} n^{k/2+(\ell+i-k)/2-\ell/2}n^{i+2(k-i)-2k} \alpha^{i+(\ell+i-k)/2} \cdot \Gamma,
\end{split}
\end{equation}
where 
$$
\Gamma= \frac{(\ell/2)^{\ell/2} (2k)!}{i!(2k-2i)!(k/2)^{k/2}(\ell+i-k)^{(\ell+i-k)/2}}2^{2k}2^{\ell+i-k}.
$$
We now bound $\Gamma$ using Lemma~\ref{quotient of factorials}. First, applying the lemma to the numerator with $m=2$, $a_1=\ell/2$ and $a_2=2k$, getting that the numerator is upper bounded by $(2k+\ell/2)^{2k+\ell/2}$. Next, applying the lemma to the denominator with $m=4$, $a_1=i$, $a_2=2k-2i$, $a_3=k/2$ (we lower bound $k^{k/2}$ by $(k/2)^{k/2}$) and $a_4=(\ell+i-k)/2$ (we lower bound $(\ell+i-k)^{(\ell+i-k)/2}$ by $(\ell+i-k/2)^{(\ell+i-k)/2}$). By Lemma~\ref{quotient of factorials} we get that the denominator is lower bounded by 
$$
((2k+\ell/2+i/2)/12)^{2k+\ell/2+i/2}\geq 12^{-(2k+\ell/2+i/2)}((2k+\ell/2)^{2k+\ell/2} (i/2)^{i/2},
$$
where we used the fact that $(a+b)^{a+b}\geq a^ab^b$ for all integer $a, b\geq 0$. Putting the above bounds together, we get that $\Gamma\leq i^{i/2} 10^{3(\ell+k+i)}$. 

Gathering the powers of $n$ and powers of $\alpha$ in the inner summation in~\eqref{eq:9023hgg}, we get that the exponent of $n$ is $-i$ and that the exponent of $\alpha$ is at least $i$. This, together with our upper bound on $\Gamma$ gives 
\begin{align*}
T_1\cdot \left(\frac{4C^2 n}{\ell}\right)^{-\ell/2} 
&\leq 
\sum_{k=1}^{n/C^2}\sum_{i=|k-\ell|_+}^{k} C^{k-\ell} \left(i/n\right)^{i/2} 10^{3(\ell+k+i)}\alpha^{i} 
\\&\leq 
10^{3\ell} \sum_{k=1}^{n/C^2}\sum_{i=|k-\ell|_+}^{k} C^{k-\ell-i} \left(\frac{10^8 C^2\alpha^2 i}{n}\right)^{i}
\end{align*}
Since $k\leq n/(2C^2)$ for all $k$ in the range of the summation, and $i\leq k$, one has $\left(\frac{10^8 C^2\alpha^2 i}{n}\right)^{i}\leq (10^8\alpha)^i\leq 2^{-i}$ by {\bf (P1)}. Substituting into the equation above, we get
\begin{align*}
T_1\cdot \left(\frac{4C^2 n}{\ell}\right)^{-\ell/2} &\leq 10^{3\ell} \sum_{k=1}^{\ell} \sum_{i=0}^{k} C^{k-\ell-i} 2^{-i}+\sum_{k=\ell+1}^{n/(2C^2)} \sum_{i=k-\ell}^{k} C^{k-\ell-i} 2^{-i}\\
&\leq 10^{3\ell} \sum_{k=1}^{\ell} C^{k-\ell}\sum_{i=0}^{k} C^{-i} 2^{-i}+\sum_{k=\ell+1}^{n/(2C^2)}2^{-(k-\ell)}  \sum_{i=0}^{k-\ell} C^{-i} 2^{-i}\\
&\leq 4\cdot 10^{3\ell}\leq 4\cdot 10^{6\ell},
\end{align*}
where in going from line~2 to line~3 we used assumption {\bf (P2)}. This implies 
\[
T_1\leq \left(\frac{4\cdot 16\cdot 10^{12} C^2 n}{\ell}\right)^{\ell/2}\leq \left(\frac{(10^7 C)^2 n}{\ell}\right)^{\ell/2}.
\]

\end{proofof}

\begin{lemma}\label{lm:t2}
Suppose parameters $C,\alpha,s^*,n$ satisfy 
\begin{align*}
{\bf (P1)} \, \alpha < 10^{-10}\quad 
{\bf (P2)} \, C>10^6 \quad 	
{\bf (P3)} \, s^*<\frac{n}{10^9 C^3} \quad
{\bf (P4)} \, n>10^9C^4,
\end{align*}
then for every $\ell\in [s^*, n/2C^2]$ 
$$
\sum_{k=n/C^2+1}^{n/100}\sqrt{2^{s^*}\binom{n}{2k}}\cdot \left(\sum_{i=k-\ell}^{k} q(k,i,n) \left(\frac{3\alpha n }{\ell-k+i}\right)^{(\ell-k+i)/2}\right)\leq \left(\frac{(10^8C)^2 n}{\ell}\right)^{\ell/2}.
$$
\end{lemma}
\begin{remark}
A close inspection of the proof of Lemma~\ref{lm:t2} reveals that a stronger bound than is needed by the application in Lemma~\ref{lm:mass-transfer-high} holds. However, we prefer to keep the bound in present form to simplify presentation. 
\end{remark}
\begin{proofof}{Lemma~\ref{lm:t2}}
Note that we have $\ell\leq n/2C^2$ and $k>n/C^2$, so we must have $i>n/2C^2$. This allows us to write 
\begin{align*}
T_2 &:=\sum_{k=n/C^2+1}^{n/100}\sqrt{2^{s^*}\binom{n}{2k}}\cdot \left(\sum_{i=k-\ell}^{k} q(k,i,n) \left(\frac{3\alpha n }{\ell-k+i}\right)^{(\ell-k+i)/2}\right)\\
&\leq 
\sum_{k=n/C^2+1}^{n/100}\sum_{i=k-\ell}^{k} \sqrt{2^{s^*}\binom{n}{2k}} q(k,i,n)\left(\frac{3\alpha n}{\ell-k+i}\right)^{(\ell-k+i)/2}
\end{align*}
Substituting the expression for $q(k, i, n)$ from Lemma~\ref{lm:qkin} and using the assumption that $\ell\geq s^*$ to upper bound $2^{s^*}$ by $2^\ell$, we get
\begin{equation}\label{eq:g4g0gag}
\begin{split}
T_2&\leq 
2^{\ell}\sum_{k=n/C^2+1}^{n/100}\sum_{i=k-\ell}^{k} \binom{n}{2k}^{1/2}\binom{\alpha n}{i}\binom{n-2\alpha n}{2(k-i)}\binom{n}{2k}^{-1}\left(\frac{3\alpha n}{\ell+i-k}\right)^{(\ell+i-k)/2} 
\\&=
2^{\ell}\sum_{k=n/C^2+1}^{n/100}\sum_{i=k-\ell}^{k} \binom{\alpha n}{i}\binom{n-2\alpha n}{2(k-i)}\binom{n}{2k}^{-1/2}\left(\frac{3\alpha n}{\ell+i-k}\right)^{(\ell+i-k)/2} 
\\&\leq
\left(\frac{16C^2 n}{\ell}\right)^{\ell/2} \sum_{k=n/C^2+1}^{n/100}\sum_{i=k-\ell}^{k} C^{-\ell} n^{(\ell+i-k)/2-\ell/2}n^{i+2(k-i)-2k/2} \alpha^{i+(\ell+i-k)/2} \cdot \Gamma, 
\end{split}
\end{equation}
where 
$$
\Gamma=\frac{(\ell/2)^{\ell/2} k^k}{i!(2k-2i)!((\ell+i-k)/2)^{(\ell+i-k)/2}}2^{2k}2^{\ell+i-k}.
$$
In~\eqref{eq:g4g0gag} we used the bounds $\binom{\alpha n}{i}\leq (\alpha n)^i/i!$, $\binom{n-2\alpha n}{2(k-i)}\leq n^{2(k-i)}/(2(k-i))!$ and $\binom{n}{2k}\geq 2^{-2k} n^k/(2k)!$ when going from line~2 to line~3. We also upper bounded $\sqrt{(2k!)}$ by $2^k k^k$ to simplify the expression for $\Gamma$.

We now show using Lemma~\ref{quotient of factorials} that $\Gamma\leq \left(\frac{k-i}{2}\right)^{-\frac{k-i}{2}} 10^{3(\ell+k+i)}$.  We first apply  the lemma to the numerator of $\Gamma$ with $m=2$, $a_1=\ell/2$, $a_2=k$, obtaining an upper bound of $(k+\ell/2)^{k+\ell/2}$. Applying the lemma to the denominator of $\Gamma$ with $m=3$, $a_1=i, a_2=2k-2i, a_3=(\ell+i-k)/2$, we obtain a lower bound of $((k+\ell/2)/9)^{k+\ell/2}$. Putting the bounds above together, we get that $\Gamma\leq \left(\frac{k-i}{2}\right)^{-\frac{k-i}{2}} 10^{3(\ell+k+i)}$, as required.

Gathering the powers of $n$ in~\eqref{eq:g4g0gag}, we obtain $n^{(k-i)/2}$ so this together with factor of $\left(\frac{k-i}{2}\right)^{-\frac{k-i}{2}}$ from the upper bound on $\Gamma$ gives us $\left(\frac{2n}{k-i}\right)^{\frac{k-i}{2}}$. Finally the power of $\alpha$ is at least $i$. Putting these bounds together, we get

\begin{align*}
T_2\cdot \left(\frac{16C^2 n}{\ell}\right)^{-\ell/2} 
&\leq
\sum_{k=n/C^2+1}^{n/100}\sum_{i=k-\ell}^{k} C^{-\ell} \left(\frac{2n}{k-i}\right)^{\frac{k-i}{2}}\alpha^{i} 10^{3(\ell+k+i)}
\\&\leq
10^{6\ell} C^{-\ell} \left(\frac{2n}{\ell}\right)^{\ell/2}\sum_{k=n/C^2+1}^{n/100}\sum_{i=k-\ell}^{k} \alpha^i 10^{6i} 10^{k-\ell-i}
\\&\leq
10^{6\ell} \left(\frac{2n/C^2}{\ell}\right)^{\ell/2}\sum_{k=n/C^2+1}^{n/100}\sum_{i=k-\ell}^{k} \left(10^6\alpha\right)^i
\\&\leq
10^{6\ell} e^{n/C^2} \sum_{k=n/C^2+1}^{n/100}\sum_{i=k-\ell}^{k} \left(10^6\alpha\right)^i
\\&\leq
10^{6\ell} e^{n/C^2} \sum_{k=n/C^2+1}^{n/100} 2\cdot \left(10^6\alpha\right)^{k-\ell}
\\&\leq
4\cdot 10^{6\ell}e^{n/C^2} \left(10^6\alpha\right)^{n/(2C^2)}
\\&\leq
10^{6\ell}.
\end{align*}

We used the fact that $k-i\leq \ell\leq n/2$ by {\bf (P3)} and {\bf (P2)} in going from line~1 to line~2, and the fact that $k-\ell-i\leq 0$ for all $i$ in the range of summation in going from line~2 to line~3. In going from line~3 to line~4 we used the fact that $\left(\frac{2n/C^2}{\ell}\right)^{\ell/2}\leq 4^{n/(2C^2)}\leq e^{n/C^2}$ for $\ell\in [s^*, n/(2C^2)]$. In going from line~4 to line~5 we used the fact that $\sum_{i=k-\ell}^{k} \left(10^6\alpha\right)^i\leq 2\left(10^6\alpha\right)^{k-\ell}$ by {\bf (P1)}. In going from line~5 to line~6 we used the fact that $\sum_{k=n/C^2+1}^{n/100}\left(10^6\alpha\right)^{k-\ell}\leq 2\cdot\left(10^6\alpha\right)^{n/C^2-\ell}\leq 2\cdot \left(10^6\alpha\right)^{n/(2C^2)}$ by {\bf (P1)} as well as the assumption that $\ell\leq n/(2C^2)$ and $k\geq n/C^2$. In going from line~6 to line~7 we used assumption {\bf (P1)}.

Putting the bounds above together, we get
\[
T_2\leq 10^{6\ell}\left(\frac{16C^2 n}{\ell}\right)^{\ell/2}\leq \left(\frac{\left(10^8C\right)^2 n}{\ell}\right)^{\ell/2},
\]
as required.
\end{proofof}

\section{Acknowledgements}
The authors would like to thank Madhu Sudan for many useful comments on drafts of this manuscript. Michael Kapralov would like to thank Sanjeev Khanna for many useful discussions on precusors to this work.

\pdfbookmark[1]{\refname}{My\refname} 
\newcommand{\etalchar}[1]{$^{#1}$}

\appendix
\section{Proof of Lemma~\ref{lm:pdfs-closeness}}\label{app:A}
\begin{proofof}{Lemma~\ref{lm:pdfs-closeness}}
The claim of the lemma is equivalent to saying that with probability at least $1-\delta$ over the choice of the matching $M$ we have, for any $z_0\in\{0,1\}^M$, the following inequality
\[
1-\delta\leq 2^{|M|}\mathbb{P}_{x\sim \operatorname{Uniform}(\B)}\left[Mx=z_0\right] \leq 1+\delta,
\]
which would in turn follow by Markov inequality from the following fact:
\[
\mathbb{E}_M\left[\max_{z_0\in\{0,1\}^M}\left|2^{|M|}\mathbb{P}_{x\sim \operatorname{Uniform}(\B)}\left[Mx=z_0\right]-1\right|\right]\leq \delta^2. 
\]
In order to prove this we express the LHS in terms of the Fourier transform of $h=\mathbbm{1}_\B$. Define $g_{z_0}(x):=\mathbbm{1}_{\{x:\, Mx=z_0\}}$. We have 
\begin{align*}
2^{|M|}\mathbb{P}_{x\sim \operatorname{Uniform}(\B)}[Mx=z_0]-1 &= \frac{2^{|M|}\left|\B\cap \{x:\, Mx=z_0\}\right|-|\B|}{|\B|} \\ &=
\frac{2^{|M|} 2^n \widehat{hg_{z_0}}(0)-2^n\widehat{h}(0)}{|\B|} \\&=
\frac{2^{n+|M|}}{|\B|}\sum_{\substack{v\in \{0,1\}^n\\ v\not = 0}}\widehat{h}(v)\widehat{g_{z_0}}(v).
\end{align*}
Note that in order for $\widehat{g}_{z_0}(v)$ to be non-zero, $|v|$ must be even. We then use a triangle inequality to obtain 
\[
\mathbb{E}_M\left[\max_{z_0\in\{0,1\}^M}\left|2^{|M|}\mathbb{P}_{x\sim \operatorname{Uniform}(\B)}\left[Mx=z_0\right]-1\right|\right]\leq \sum_{\ell=1}^{n/2}\mathbb{E}_M\left[\max_{z_0\in\{0,1\}^M}\left|\frac{2^{n+|M|}}{|\B|}\sum_{\substack{v\in \{0,1\}^n\\ |v|= 2\ell}}\widehat{h}(v)\widehat{g_{z_0}}(v)\right|\right].
\]
Recall that $p_\alpha(\ell,n)$ stands for the probability that fixed $2\ell$ points are matched by a uniformly random matching of size $\alpha n$. For $v\in\{0,1\}^n$ which is perfectly matched by $M$ (i.e. $M$ restricted to $v$ is a perfect matching) let $e(v)\in \{0,1\}^M$ denote the set of edges induced by $v$. Using explicit structure of $\widehat{g_{z_0}}$ we find that 
\begin{equation}\label{eq:239h9f23fl}
\begin{split}
\frac{2^{n+|M|}}{|\B|}\sum_{\substack{v\in \{0,1\}^n\\ |v|= 2\ell}}\widehat{h}(v)\widehat{g_{z_0}}(v) &=\frac{2^{n+|M|}}{|\B|}\sum_{\substack{v\in \{0,1\}^n\\ |v|= 2\ell \\ v \text{ is matched by } M}}\widehat{h}(v)\cdot 2^{-|M|}(-1)^{z_0\cdot e(v)} \\&\leq \frac{2^n}{|\B|} \sum_{\substack{v\in \{0,1\}^n\\ |v|= 2\ell \\ v \text{ is matched by } M}}\left|\widehat{h}(v)\right|,
\end{split}
\end{equation}
where we used the fact that $\wh{g_{z_0}}(v)$ is zero for $v$'s that are not perfectly matched by $M$. The bound is independent of $z_0$ which allows us to write, after taking expectation over $M$,
\begin{equation}\label{eq:expectation-bound}
\begin{split}
&\mathbb{E}_M\left[\max_{z_0\in\{0,1\}^M}\left|2^{|M|}\mathbb{P}_{x\sim \operatorname{Uniform}(\B)}\left[Mx=z_0\right]-1\right|\right]\\
 &\leq \sum_{\ell=1}^{n/2} p(\ell,n)\frac{2^n}{|\B|}\sum_{\substack{v\in \{0,1\}^n \\ |v|=2\ell}}\left|\widehat{h}(v)\right| \\ 
  &\leq \sum_{\ell=1}^{\alpha n} p(\ell,n)\frac{2^n}{|\B|}\sum_{\substack{v\in \{0,1\}^n \\ |v|=2\ell}}\left|\widehat{h}(\ell)\right| \\ 
 &\leq \sum_{\ell=1}^{s^*} p(\ell,n)\bound(\ell)+\sum_{\ell=s^*+1}^{n/2}p(\ell,n)\sqrt{2^s \binom{n}{2\ell}}.
\end{split}
\end{equation}
In going from line~2 to line~3 we used the fact that $\widehat{h}(v)=0$ for any $v$ that is not perfectly matched by $M$, as well as the assumption that $M$ is a matching of size $\alpha n$.
In going from line~3 to line~4 we used the fact that $\frac{2^n}{|\B|}\sum_{\substack{v\in \{0,1\}^n \\ |v|=2\ell}}\left|\widehat{h}(\ell)\right|\leq \bound(\ell)$ for $\ell$ between $1$ and $s^*$ (since $\B$ is $(C, s^*)$-bounded by assumption, as well as the fact that for every $\ell$ between $0$ and $n/2$ 
$$
\frac{2^n}{|\B|}\sum_{\substack{v\in \{0,1\}^n \\ |v|=2\ell}}\left|\widehat{h}(\ell)\right|\leq \sqrt{{n \choose 2\ell}} \frac{2^n}{|\B|}\sqrt{\sum_{\substack{v\in \{0,1\}^n \\ |v|=2\ell}}\widehat{h}(\ell)^2}\leq \sqrt{2^{s^*} \binom{n}{2\ell}}.
$$

The latter bound holds for any subset $\B$ of the cube with $|\B|/2^n\geq 2^{-s^*}$ by Parseval's equality (\eqref{eq:parseval} and Remark~\ref{rm:indicator-l2}) together with Cauchy-Schwarz. 

We now upper bound individual terms in the summation over $\ell$ in~\eqref{eq:expectation-bound} above. We use the expression 
\begin{equation}\label{eq:923hgg}
p(\ell,n)=\binom{\alpha n}{\ell}\binom{n}{2\ell}^{-1}
\end{equation} 
provided by Lemma~\ref{lm:pln}.

\paragraph{Upper bounding the contribution from low weights ($1\leq \ell\leq s^*$).} For low weights ($\ell\leq s^*$) we have by the assumption that $h$ is $(C, s^*)$-bounded
\begin{equation}\label{eq:ub-low-weights}
\begin{split}
p(\ell,n)\bound(\ell)&=\binom{\alpha n}{\ell}\binom{n}{2\ell}^{-1} \left(\frac{C\sqrt{s^* n}}{\ell}\right)^\ell \\
&\leq \left(\frac{e\alpha n}{\ell}\right)^\ell\cdot \frac{(2\ell)!}{(n/2)^{2\ell}}\cdot \left(\frac{C\sqrt{s^* n}}{\ell}\right)^\ell \\ 
& \leq \left(44\alpha C\sqrt{s^*/n}\right)^\ell\\
&\leq \left(44\alpha\right)^\ell \delta^{2\ell}.
\end{split}
\end{equation}
In going from line~1 to line~2 above we used the lower bound ${n \choose 2\ell}\geq (n/2)^{2\ell}/(2\ell)!$, since $\ell\leq s^*\leq n/C^2\leq n/4$ (as $C\geq 100$ by assumption), as well as the bound ${\alpha n \choose \ell}\leq (e \alpha n/\ell)^\ell$. We used the assumption $s^*\leq \delta^4 n/C^2$ to from line~3 to line~4.

\paragraph{Upper bounding the contribution from high weights ($s^*<\ell \leq \alpha n$).} We have, using~\eqref{eq:923hgg},
\begin{equation}\label{eq:ub-high-weights}
\begin{split}
 p(\ell, n)\sqrt{2^{s^*}\binom{n}{2\ell}}&=\binom{\alpha n}{\ell} \binom{n}{2\ell}^{-1/2}2^{s^*/2} \\
 &\leq
\frac{(\alpha n)^\ell}{\ell!}\left(\frac{(n/2)^{2\ell}}{(2\ell)!}\right)^{-1/2}2^{\ell/2}
\\&=
(2\sqrt{2}\alpha)^{\ell}\sqrt{\binom{2\ell}{\ell}}
\\&\leq 
(4\sqrt{2}\alpha)^{\ell}
\\&\leq e^{-2s^*}.
\end{split}
\end{equation}
In going from line~1 to line~2 we used the fact that $\binom{n}{2\ell}\geq (n/2)^{2\ell}/(2\ell)!$ for any $\ell\leq \alpha n\leq n/4$. In going from line~4 to line~5 we used the assumption that $\alpha\leq 1/100$.

Putting~\eqref{eq:ub-low-weights} together with~\eqref{eq:ub-high-weights} and summing over all $\ell\in [1:\alpha n]$, we get using~\eqref{eq:expectation-bound} 
\begin{equation*}
\begin{split}
&\mathbb{E}_M\left[\max_{z_0\in\{0,1\}^M}\left|2^{|M|}\mathbb{P}_{x\sim \operatorname{Uniform}(\B)}\left[Mx=z_0\right]-1\right|\right]\\
&\leq \sum_{\ell=1}^{s^*} \left(44\alpha\right)^\ell \delta^{2\ell}+\sum_{\ell=s^*+1}^{n/2} e^{-2s^*}\\
&\leq  88\alpha\cdot \delta^2+n^{-10}\\
&\leq  \delta^2.\\
\end{split}
\end{equation*}
In going from line~2 to line~3 we used the assumption that $s^*\geq 10\ln (n+1)$, and in going from line~3 to line~4 we used the assumption that $\alpha<1/100$ and $\delta\in (1/n, 1/2)$. An application of Markov's inequality now gives the result. 
\end{proofof}

\section{Proof of Lemma~\ref{lm:main-tvd}}\label{app:B}
 We will use some basic properties of total variation distance, which we now state.
\begin{lemma}\label{lm:tvd-expectation}
Let $\mu, \nu$ be two probability distributions on the same finite sample space $\Omega$, and consider independent random variables $X,\widetilde{X}\sim \mu, Y\sim\nu$ taking values in $\Omega$. Then one has 
\[
2\cdot\|\mu-\nu\|_{tvd}=\mathbb{E}_{X}[|1-\prob_{Y}[Y=X]/\prob_{\widetilde{X}}[\widetilde{X}=X]|].
\] 
\end{lemma} 
\begin{proof}
Identify $\Omega$ with $\{1,2,\dots, n\}$ and let $p_k=\prob[X=k], q_k=\prob[Y=k]$. Then we have 
\[
\mathbb{E}_{X}[|1-\prob_{Y}[Y=X]/\prob_{\widetilde{X}}[\widetilde{X}=X]|]=\sum_{k=1}^np_k\cdot\left|1-q_k/p_k\right|=\sum_{k=1}^n |p_k-q_k|=2\cdot\|\mu-\nu\|_{tvd}.
\]
\end{proof}
\begin{lemma}[Substitution lemma]\label{lm:tvd-subst}
Let $X^1, X^2$ be random variables taking values on finite sample space $\Omega_1$. Let $Z^1, Z^2$ be random variables taking values on samples space $\Omega_2$, and suppose that $Z^2$ is independent of $X^1, X^2$. Let $f:\Omega_1\times \Omega_2 \to \Omega_3$ be a function. Then 
$$
||(X^1, f(X^1, Z^1))-(X^2, f(X^2, Z^2))||_{tvd}\leq ||(X^1, f(X^1, Z^1))-(X^1, f(X^1, Z^2))||_{tvd}+||X^1-X^2||_{tvd}.
$$
\end{lemma}
\begin{proof}
By triangle inequality the left hand side is at most 
\[
||(X^1, f(X^1, Z^1))-(X^1, f(X^1, Z^2))||_{tvd}+||(X^1, f(X^1, Z^2))-(X^2, f(X^2, Z^2))||_{tvd}.
\]
It remains to note that the second summand is at most $||X^1-X^2||_{tvd}$ by Claim \ref{cl:1}.
\end{proof}
The next lemma also follows easily from the definition of total variation distance (see, e.g. Claim~6.5 in~\cite{KKS15} for a proof):
\begin{lemma}\label{lm:tv-function}
For any random variables $X, Y$ taking values on finite sample space $\Omega_1$, independent random variable $Z$ taking values on finite sample space $\Omega_2$ and any function $f:\Omega_1\times \Omega_2 \to \Omega_3$ one has $||f(X, Z)-f(Y, Z)||_{tvd}\leq ||X-Y||_{tvd}$.
\end{lemma}

\begin{proofof}{Lemma~\ref{lm:main-tvd}}
The  proof is by induction on $t=1,\ldots, T$. We prove that for all $t$ one has 
$$
||(M_{1:t}, S^Y_{1:t})-(M_{1:t}, S^N_{1:t})||_{tvd}\leq \gamma t/T+\sum_{j=1}^t \prob[\bar \E_j| \E_{j-1}].
$$

\begin{description}
\item[Base:$t=1$] We have, conditional on the event $\E_1$, 
\begin{equation*}
\begin{split}
||(M_1, S^Y_{1})-(M_1, S^N_1)||_{tvd}&=||(M_1, S^Y_1)-(M_1, r_1(M_1, U_1, S^N_0))||_{tvd, \E_1}\\
&=||(M_1, S^Y_1)-(M_1, r_1(M_1, U_1, S^Y_0))||_{tvd, \E_1}\\ 
\end{split}
\end{equation*}
where we used the fact that $S^Y_0=S^N_0$. By assumption~Eq.~\eqref{eq:ind-assumption} of the lemma and the fact that total variation distance is bounded by $1$ we have
\begin{equation*}
\begin{split}
||(M_1, S^Y_1)-(M_1, r_1(M_1, U_1, S^N_0))||_{tvd, \E_1}\leq \gamma/T+\prob[\bar \E_1]=\gamma/T
\end{split}
\end{equation*}
as required. This proves the base of the induction.

\item[Inductive step: $t-1\to t$]

We condition on $\E_t$ in what follows, and write $||\cdot ||_{tvd, \E_t}$ to denote the total variation distance between conditional distributions. We have
\begin{equation*}
\begin{split}
&||(M_{1:t}, S^Y_{1:t})-(M_{1:t}, S^N_{1:t})||_{tvd, \E_t}\\
=&||(M_{1:t-1}, S^Y_{1:t-1}, M_t, r_t(M_{1:t-1}, S^Y_{1:t-1}, M_t, M_t X^*))-(M_{1:t-1}, S^N_{1:t-1}, r_t(M_{1:t-1}, S^N_{1:t-1}, M_t, U_t))||_{tvd, \E_t}\\
\end{split}
\end{equation*}
We would like to apply Lemma~\ref{lm:tvd-subst} to the expression above. To that effect we define 
$$
Q^Y_{t-1}=(M_{1:t-1}, S^Y_{1:t-1})\text{~~and~~}Q^N_{t-1}=(M_{1:t-1}, S^N_{1:t-1}).
$$
With this notation in place we have 
\begin{equation}\label{eq:FHFIKEsdf}
\begin{split}
&||(M_{1:t-1}, S^Y_{1:t-1}, M_t, r_t(M_{1:t-1}, S^Y_{1:t-1}, M_t, M_t X^*))-(M_{1:t-1}, S^N_{1:t-1}, r_t(M_{1:t-1}, S^N_{1:t-1}, M_t, U_t))||_{tvd, \E_t}\\
=&||(Q^Y_{t-1}, M_t, r_t(Q^Y_{t-1}, M_t, M_t X^*))-(Q^N_{t-1}, M_t, r_t(Q^N_{t-1}, M_t, U_t))||_{tvd, \E_t}.
\end{split}
\end{equation}

We now apply the substitution lemma (Lemma~\ref{lm:tvd-subst}) to Eq.~\eqref{eq:FHFIKEsdf}. The parameters are as follows. We let $X^1=Q^Y_{t-1}$ and $X^2=Q^N_{t-1}$. The variables $Z^1$ and $Z^2$ are set as $Z^1=(M_t, M_t X^*)$ (recall that $X^*$ is the hidden bipartition) and $Z^2=(M_t, U_t)$.  Note that this setting of $Z^2$ satisfies the preconditions of Lemma~\ref{lm:tvd-subst}: $Z^2=(M_t, U_t)$ is independent of $X^1, X^2$, as required.

In order to apply Lemma~\ref{lm:tvd-subst}, it remains to define the function $f$ that maps tuples $(X, Z)$ to some universe so that 
$$
(X^1, f(X^1, Z^1))=(Q^Y_{t-1}, f(Q^Y_{t-1}, (M_t, M_t X^*)))
$$ 
equals $(Q^Y_{t-1}, M_t, r_t(Q^Y_{t-1}, M_t, M_t X^*))$
and 
$$
(X^2, f(X^2, Z^2))=(Q^N_{t-1}, f(Q^N_{t-1}, (M_t, U_t)))
$$
equals $(Q^N_{t-1}, M_t, r_t(Q^N_{t-1}, M_t, U_t))$. For that, it is sufficient to let $f$ be the function that maps input tuple $(X, (B, C))$ to $(B, r_t(X, B, C))$. 

Applying Lemma~\ref{lm:tvd-subst} with these settings, we get
\begin{equation*}
\begin{split}
&||(Q^Y_{t-1}, M_t, r_t(Q^Y_{t-1}, M_t, M_t X^*))-(Q^N_{t-1}, r_t(Q^N_{t-1}, M_t, U_t))||_{tvd, \E_t}\\
\leq &||Q^Y_{t-1}-Q^N_{t-1}||_{tvd}+||(M_t, r_t(Q^Y_{t-1}, M_t, M_tx))-(M_t, Q^Y_{t-1}, r_t(Q^Y_{t-1}, M_t, U_t))||_{tvd, \E_t}\\
\end{split}
\end{equation*}

The first term is bounded by $\gamma(t-1)/T+\sum_{j=1}^{t-1} \prob[\bar \E_j|\E_{j-1}]$ by the inductive hypothesis. Using the assumption~Eq.~\eqref{eq:ind-assumption} of the lemma, we get 
\begin{equation*}
\begin{split}
&||(M_t, r_t(Q^Y_{t-1}, M_t, M_t X^*))-(M_t, Q^Y_{t-1}, r_t(Q^Y_{t-1}, M_t, U_t))||_{tvd, \E_t}\\
&=\expect_{(M_t, Q^Y_{t-1})\in \E_t}\left[||(M_t, r_t(Q^Y_{t-1}, M_t, M_t X^*))-(M_t, Q^Y_{t-1}, r_t(Q^Y_{t-1}, M_t, U_t))||_{tvd}\right]\\
&\leq \gamma/T,
\end{split}
\end{equation*}
where the total variation distance in the second line is over $X^*\sim UNIF(S^Y_{t-1})$. Putting the two bounds together yields 
$$
||(M_{1:t}, S^Y_{1:t})-(M_{1:t}, S^N_{1:t})||_{tvd}\leq \gamma t/T+\sum_{j=1}^{t} \prob[\bar \E_j| \E_{j-1}]
$$
as required. Substituting $t=T$ we get 
\[
||(M_{1:T}, S^Y_{1:T})-(M_{1:T}, S^N_{1:T})||_{tvd}\leq \gamma +\sum_{j=1}^{T} \prob[\bar \E_j| \E_{j-1}]\leq \gamma+T\cdot \gamma/T=2\gamma,
\]
thus proving the lemma.
\end{description}
\end{proofof}

\section{Details omitted from Section~\ref{sec:outline-app}}\label{sec:app-C}

Here we formally prove that if $s\ll n/B^T$ for large enough $B$ then with high probability we will have $\|F_T\|\ll n$. We start by the following Lemma.
\begin{lemma}\label{lm:martingale}
Let $m,T$ be positive integers and  let $\{X_k\}_{k\in[mT]}$ be a sequence of positive random variables satisfying $X_0<m/100^T$, and for $k\in [mT]$ we have 

\[
\mathbb{E}[X_k| X_{k-1}]\leq X_{k-1}\cdot \left(1+\frac{1}{m}+\frac{X_{k-1}}{m^2}\right).
\]
Then $\prob[\max_{i\in[mT]} X_{i}>n/2^T]<2^{-T}$.
\end{lemma}
\begin{proof}
Consider stopping time $\tau:=\inf\{t: X_t>m/2^{T}\}$, if $X_{mT}\leq m/2^T$ we define $\tau:=mT$. Note that for $k\leq \tau$ we have 
\[
1+\frac{1}{m}+\frac{X_{k-1}}{m^2}<1+\frac{2}{m},
\] 
which means that the process $Y_k:=X_k/(1+2/m)^k$ stopped at time $\tau$ is a supermartingale (i.e. $\mathbb{E}[Y_k|Y_{k-1}]\leq Y_{k-1}$). By Markov's inequality applied to $Y_\tau$ we have
\[
\prob[\max_{i\in[mT]} X_{i}>n/2^T]= \prob[X_\tau>n/2^T]\leq \prob\left[Y_\tau>\frac{m}{2^T\cdot (1+2/m)^{mT}}\right]\leq\frac{\mathbb{E}[Y_\tau]}{\left(\frac{m}{2^T\cdot (1+2/m)^{mT}}\right)}.
\]
Optional stopping theorem implies that $\mathbb{E}[Y_\tau]\leq \mathbb{E}[Y_0]=\mathbb{E}[X_0]<m/100^T$. So we have 
\[
\prob[\max_{i\in[mT]} X_{i}>n/2^T]\leq \frac{m/100^T}{\left(\frac{m}{2^T\cdot (1+2/m)^{mT}}\right)}\leq \frac{m/100^T}{m/(2e^2)^T}<6^{-T}.
\]
\end{proof}
We then deduce the following lemma
\begin{lemma}
Suppose the communication budget is $s<n/800^T$, then with probability at least $1-2^{-T}$ the forest $F_T$ formed by all edges revealed by players satisfies $\|F_T\|\leq n/2^T$. \xxx[DK]{How to make the statement more formal?..}
\end{lemma}
\begin{proof}
Recall that we consider forests $F_{k}, k\in[0\ldots T]$ formed by edges revealed by first $k$ players. We also define $F_t^{i}$ for $i\in[0\ldots \alpha n]$ inductively as 
$$
F_t^j:=\left\lbrace
\begin{array}{ll}
F_t^{j-1}\cup \{e_j\}&\text{if~}e_j\text{~intersects with a nontrivial component in~}F_t^{j-1}\\
F_t^{j-1}&\text{o.w.}
\end{array}
\right.
$$
We then deduced that (see \eqref{eq:combi-one-step-evolution})
\begin{equation}\label{eq:one-step}
\mathbb{E}_{e_i}\left[\|F_t^{i}\| \middle|e_{[1:i-1]}\right]\leq \|F_t^{i-1}\| \cdot \left(1+\frac{12}{n}+\frac{8\|F_t^{i-1}\|}{n^2}\right).
\end{equation}
We refer to Section \ref{sec:combinatorial-analysis} for more details. Let $m=\alpha n + 1$ and define a sequence of random variables $X_k, k\in[0\ldots mT]$ by
\[
X_k:=\|F_{t+1}^i\|\cdot 2^{T-t},
\]
where $i=0\ldots m-1$ and $k=m\cdot t+i$. We now check that $X_k$ satisfies the assumption of Lemma \ref{lm:martingale}. For $k$ not divisible by $m$ this trivially follows from \eqref{eq:one-step} since $m<n/12$ and $X_k\geq \|F_{t+1}^i\|$. For $k=m\cdot t$ we have $X_k=\|F_t^0\|\cdot 2^{T-t}=\|F_{t-1}\|\cdot 2^{T-t}=(\|F_{t-1}^{\alpha n}\|+4s)\cdot 2^{T-t}$ and $X_{k-1}=F_{t-1}^{\alpha n}\cdot 2^{T-t+1}$.	Since we definitely have $\|F_{t-1}^{\alpha n}\|\geq \|F_0\|=4s$, this implies that $X_k\leq X_{k-1}$ deterministically, and so the condition of Lemma \ref{lm:martingale} is satisfied in this case as well.
Note also that $X_0=4s\cdot 2^{T}<n/100^T$. Lemma \ref{lm:martingale} then implies that $\prob[X_{mT}>n/2^T]<2^{-T}$. So we infer
\[
\prob[\|F_T\|>n/2^T]=\prob[X_{mT}>n/2^T]<2^{-T}.
\]

\end{proof}

\section{Useful facts} \label{tools}

\subsection{Concentration inequalities}
We prove Lemma~\ref{lm:chernoff}, restated here for convenience of the reader:

\noindent{{\bf Lemma ~\ref{lm:chernoff}} \em 
\lemmaChernoff
}
\begin{proof}
Note that for any $u>0$ we have 
\[
\mathbb{E}\left[e^{uX}\right]=\mathbb{E}\left[\prod_{k=1}^n e^{uX_k}\right]\leq (1-p+p\cdot e^u)\cdot \mathbb{E}\left[\prod_{k=1}^{n-1} e^{uX_k}\right]\leq\ldots\leq (1-p+p\cdot e^u)^n. 
\]
By Markov inequality we then have 
\[
\prob[X\geq \mu+\Delta]\leq\min_{u>0}\frac{(1-p+p\cdot e^u)^n}{e^{u(\mu+\Delta)}}=\min_{v>0}\frac{(1+pv)^{\mu/p}}{(1+v)^{\mu+\Delta}},
\]
where in the last equality we made a substitution $v=e^u-1$. Using the fact that $e^{x-x^2/2}\leq 1+x\leq e^x$ for all positive $x$ we then infer
\[
\prob[X\geq \mu+\Delta]\leq\min_{v>0}\frac{e^{pv\cdot \mu/p}}{e^{(v-v^2/2)\cdot(\mu+\Delta)}}=\min_{v>0}\exp{\left(-v\cdot \Delta + v^2/2\cdot (\mu+\Delta)\right)}.
\]
The minimum is achieved at $v=\frac{\Delta}{\mu+\Delta}$ and equals $-\frac{\Delta^2}{2(\mu+\Delta)}$, as desired.
\end{proof}

We also prove Lemma~\ref{lm:cut}, restated here for convenience of the reader:

\noindent{\em {\bf Lemma~\ref{lm:cut}} 
Let $G$ be a miltigraph with $n$ vertices and $m$ edges (counted with multiplicities) in which each edge has multiplicity at most $k$. Let $S\subset [n]$ be a uniformly random subset of vertices and $X$ be the number of edges crossing $(S, \bar{S})$. Then for any $\delta>0$ we have 
\[
\prob[X<m/2\cdot (1-\delta)]\leq \frac{k}{\delta^2 m}.
\] 
}
\begin{proof}
Let $m_1,m_2,\ldots, m_s$ be multiplicities of edges of $G$ and let $\{p_i\}_{i\in[s]}$ be $0/1$ random variables indicating if the corresponding edge crosses the cut or not. Note that $\mathbb{E}[p_i]=1/2$ and $\mathbb{E}[p_ip_j]=1/4$ for $i\not= j$. We then infer that 
\[
\mathbb{E}[X]=\sum_{i=1}^s m_i\cdot \mathbb{E}[p_i]=m/2,\qquad \mathbb{E}[X^2]=\sum_{i=1}^s \sum_{j=1}^s m_im_j\cdot \mathbb{E}[p_ip_j]=m^2/4+\frac{1}{4}\sum_{i=1}^s m_i^2.
\]   
By Chebyshev inequality we then have 
\[
\prob[X<m/2\cdot (1-\delta)]\leq \frac{\left(\sum m_i^2\right)/4}{(\delta m/2)^2}\leq \frac{k\cdot m}{\delta^2 m^2}=\frac{k}{\delta^2 m}.
\]
\end{proof}

We now prove Lemma \ref{lm:gap} restated here for convenience of the reader:

\noindent{{\bf Lemma ~\ref{lm:gap}} \em 
\lemmagap
}

\begin{proof}
We let $m_0=\frac{\alpha n T}{2}\cdot (1-\delta)$ with $\delta=\e/100$. In the \YES case the graph is bipartite so the value of MAX-CUT is equal to the number of edges in the graph. Since in the \YES case we only keep those edges of the matchings which cross a fixed random bipartition, and in the union of matchings every edge has multiplicity at most $T$, Lemma~\ref{lm:cut} ensures that the probability that the number of edges if smaller than $m_0$ is at most 
\[
\frac{T}{\delta^2 \alpha n T}=\frac{1}{\delta^2\alpha n}\leq 1/\sqrt{n},
\]
since $\e, \alpha>n^{-1/10}$ by assumption of the lemma.

We now consider the \NO case. Since every edge of the matchings is kept with probability $1/2$ independently of the others, by Lemma~\ref{lm:chernoff} with probability at least $1-\exp{\left(-\frac{\delta^2\alpha n T }{4(1+\delta)}\right)}$ (we invoke Lemma~\ref{lm:chernoff} with $\mu=\alpha n T/2$ and $\Delta=\delta \mu$) the number $m$ of edges in the graph will be at most $\frac{\alpha n T}{2}\cdot(1+\delta)$. In the following we assume that $m\leq\frac{\alpha n T}{2}\cdot(1+\delta)$. 
 
Fix a cut $(S, \bar{S})$ where $S \subseteq V$. Let $k:=|S|$. Consider the edges in $E_t$, that is the matching generated at step $t$. Consider an arbitrary order of the edges in $E_t$. Conditioned on all previous edges, an edge $e$ in $E_t$ is a uniformly random edge with endpoints in a set $V$ of vertices not covered by the previous edges. Clearly, $|V|\geq n-2\alpha n>n/2$. The probability that the edge $e$ crosses the cut $(S, \bar{S})$ is
\[
\frac{|S\cap V|\cdot |\bar{S}\cap V|}{\binom{|V|}{2} }\leq \frac{|V|^2/4}{|V|\cdot (|V|-1)/2 } =\frac{1}{2}\cdot \frac{|V|}{|V|-1}\leq \frac{1+3/n}{2}.
\]
We then use Lemma~\ref{lm:chernoff} with $p=\frac{1+3/n}{2}$ and $\mu=p\cdot m$ applied to the random variables indexed by edges of the graph and equal $1$ if the edge crosses the cut $(S,\bar{S})$ and $0$ otherwise. Then for the random variable $Y_S$ which is the number of edges crossing the cut $(S,\bar{S})$, Lemma~\ref{lm:chernoff} gives
\begin{equation*}
\begin{split}
\prob\left[ Y_S > \frac{m_0}{2-\e}\right] \leq \exp{\left(-\frac{(\frac{m_0}{2-\e}-p\cdot m)^2}{2m_0/(2-\e)}\right)}.
\end{split}
\end{equation*}
We have 
\[
\frac{m_0}{2-\e}-p\cdot m\geq \frac{\alpha n T\cdot (1-\delta)}{2\cdot (2-\e)}-\frac{1+3/n}{2}\cdot \frac{\alpha n T}{2}\cdot (1+\delta)\geq \frac{\alpha n T}{4}\cdot \left(\frac{1-\delta}{1-\e/2}-(1+\delta)^2\right)\geq \alpha n T \e/16,
\]
and 
\[
2m_0/(2-\e)=\frac{\alpha n T\cdot (1-\delta)}{2-\e} \leq \alpha n T. 
\]
Putting this together we get 
\[
\prob\left[ Y_S > \frac{m_0}{2-\e}\right] \leq \exp{\left(-\frac{(\alpha n T \e/16)^2}{\alpha n T }\right)}=\exp\left(-\alpha n T \e^2 / 256\right)\leq \exp{(-2n)}.
\]
where we used the assumption that $T=\ctwo/(\alpha \e^2)$. Taking a union bound over all $2^n$ possible cuts completes the proof.
\end{proof}

\subsection{Combinatorics}\label{app:combi}
\begin{lemma}\label{lm:partition}
Let $\mathbf{P}_1\cup \mathbf{P}_2\dots\cup \mathbf{P}_m$ be a partition of the cube $\{0,1\}^n$. For $x\in\{0,1\}^n$ let $\mathbf{P}_x$ be the unique part $P_i$ which contains $x$. Then 
\[
\mathbb{E}_{x\sim\operatorname{Uniform}\left(\{0,1\}^n\right)}[1/|\mathbf{P}_x|]=m/2^n.
\]
\end{lemma}
\begin{proof}
Indeed, by the definition of the expectation we have 
\[
\mathbb{E}_{x\sim\operatorname{Uniform}\left(\{0,1\}^n\right)}[1/|P_x|]=\sum_{i=1}^m \mathbb{P}_{x\sim\operatorname{Uniform}\left(\{0,1\}^n\right)}[x\in \mathbf{P}_i]/ |\mathbf{P}_i|=\sum_{i=1}^m |\mathbf{P}_i|\cdot 2^{-n}/ |\mathbf{P}_i|=m/2^n.
\]
\end{proof}

We now give proof of Lemmas from Section~\ref{sec:prelims-random-matchings} which we restate for the convinience of the reader.

\noindent{{\bf Lemma~\ref{lm:pln}} \em 
\lmpln
}
\begin{proof}
Let $M$ be a random matching of size $\alpha n$. The expected number of sets of size $2\ell$ which are matched by $M$ is $\binom{\alpha n}{\ell}$ (in fact, it is always equal to $\binom{\alpha n}{\ell}$). The statement of the lemma easily follows. 
\end{proof}

\noindent{{\bf Lemma~\ref{lm:qkin}} \em 
\lmqkin
}
\begin{proof}
For any fixed matching $M$ of size $\alpha n$ the number of sets $A$ of cardinality $2k$ for which exactly $2i$ points are matched and there are no boundary edges is $\binom{\alpha n}{i}\binom{n-2\alpha n}{2(k-i)}$. Indeed, to construct such a set $A$ we need to choose $i$ edges of $M$ which match points of $A$ and also choose $2(k-i)$ points of $A$ which are not matched by $M$. The statement of the lemma easily follows. 
\end{proof}

\noindent{{\bf Lemma~\ref{lm:qkibn}} \em 
\lmqkibn
}
\begin{proof}
For any fixed matching $M$ of size $\alpha n$ the number of sets $A$ of size with $b$ boundary edges and $i$ inner edges is 
\[
\binom{\alpha n}{i}\binom{\alpha n - i}{b}2^{b}\binom{n-2\alpha n}{2(k-i)-b}.
\]
Indeed, we need to choose $i$ inner edges, then chose $b$ boundary edges, then for each boundary edge decide which of its end-points belongs to $A$, and choose remaining $2k-2i-b$ points of $A$ which are not matched by $M$. The statement of the lemma easily follows.
\end{proof}

\noindent{{\bf Lemma~\ref{lm:qkibninequality}} \em 
\lmqkibninequality
}
\begin{proof}
By Lemma \ref{lm:qkibn} and Lemma \ref{lm:qkin} we have 
\[
q(k,i,b, n)=\binom{\alpha n}{i}\binom{\alpha n - i}{b}2^{b}\binom{n-2\alpha n}{2(k-i)-b}\binom{n}{2k}^{-1}, \qquad 
q(k,i,n)=\binom{\alpha n}{i}\binom{n-2\alpha n}{2(k-i)}\binom{n}{2k}^{-1}.
\]
Consequently, we have  
\begin{equation}\label{eq:9h23g4g9j}
\frac{q(k,i,b,n)}{q(k,i,n)}=2^b \binom{\alpha n-i}{b} \frac{\binom{n-2\alpha n}{2(k-i)-b}}{\binom{n-2\alpha n}{2(k-i)}}.
\end{equation}
Since we have 
\[
\frac{\binom{n-2\alpha n}{2(k-i)-b}}{\binom{n-2\alpha n}{2(k-i)}}=\frac{(2k-2i)!(n-2\alpha n - 2(k-i))!}{(2(k-i)-b)!(n-2\alpha n - 2(k-i)+b)!}=\frac{\binom{2k-2i}{b}}{\binom{n-2\alpha n - 2(k-i)+b}{b}},
\]
we can rewrite~\eqref{eq:9h23g4g9j} as 
\[
\frac{q(k,i,b,n)}{q(k,i,n)}=2^b \frac{\binom{\alpha n-i}{b} \binom{2k-2i}{b}}{\binom{n-2\alpha n - 2(k-i)+b}{b}}.
\]
We then use the following bounds to handle factors in the numerator:
\[
\binom{\alpha n-i}{b}\leq \frac{(\alpha n - i)^b}{b!}\leq \frac{(\alpha n)^b}{b!}, \qquad \binom{2k-2i}{b}\leq 2^{2k-2i}.
\]
To lower bound the denominator we note that $n-2\alpha n - 2(k-i) \geq n - n/50 - n/5 > n/2$ and so 
\[
\binom{n-2\alpha n - 2(k-i)+b}{b}\geq \frac{(n-2\alpha n - 2(k-i))^b}{b!}\geq \frac{(n/2)^b}{b!}.
\]
Putting this together we obtain
\[
\frac{q(k,i,b,n)}{q(k,i,n)}\leq \frac{(\alpha n)^{b}}{b!}2^{b}\frac{2^{2(k-i)}}{(n/2)^{b}/b!} = (4\alpha)^{b}4^{k-i}\leq 20^{-b}4^{k-i},
\]
which completes the proof.
\end{proof}
\subsection{Inequalities}
Since we deal with sums when we know much more about sum of squares, we repeatedly use the following fact.
\begin{lemma}\label{ineq:CS}
Let $a_1,a_2,\dots,a_m$ be real numbers. Then 
\[
\sum_{i=1}^m |a_i|\leq \sqrt{m\sum_{i=1}^m a_i^2}.
\]
\end{lemma}
\begin{proof}
Indeed, we have 
\[
\left(\sum_{i=1}^m |a_i|\right)^2=m\cdot \sum_{i=1}^m a_i^2 - \sum_{1\leq i< j \leq m} (a_i-a_j)^2\leq m\cdot\sum_{i=1}^m a_i^2.
\]
\end{proof}
We also need the following inequality for binomial coefficients in terms of the entropy function. 
\begin{lemma}\label{lm:binom-entropy}
Let $H(x)=-x\log{x}-(1-x)\log{(1-x)}$ be the entropy function. Then for all $0\leq k\leq n$ we have 
\[
\frac{e^{nH(k/n)}}{n+1}\leq \binom{n}{k}\leq e^{nH(k/n)}.
\]
\end{lemma}
\begin{proof}
For the lower bound note that for every $t>1$ we have 
\[
I(t):=\int_0^1 (xt+(1-x))^n \, dx = \int_0^1 (1+x(t-1))^n \, dx = \frac{t^{n+1}-1}{(n+1)(t-1)}=\frac{1+t+\dots+t^n}{n+1}.
\]
On the other hand, 
\[
I(t)=\sum_{i=0}^n t^i \binom{n}{i}\int_0^1 x^i(1-x)^{n-i} \, dx.
\]
Comparing coefficients at $t^k$ we get
\[
\int_0^1 x^k(1-x)^{n-k}\, dx=\frac{1}{(n+1)\binom{n}{k}}.
\]
Apply AM-GM inequality to $x_1=x_2=\dots=x_k=x/k, x_{k+1}=\dots=x_n=(1-x)/(n-k)$ to get
\[
\frac{x^k(1-x)^{n-k}}{k^k\cdot (n-k)^{n-k}}=\prod_{i=1}^n x_i \leq \left(\frac{x_1+x_2+\dots+x_n}{n}\right)^n=1/n^n.
\]
This implies 
\[
\frac{1}{(n+1)\binom{n}{k}}=\int_0^1 x^k(1-x)^{n-k}\, dx \leq \max_{x\in [0,1]} x^k(1-x)^{n-k} \leq \frac{k^k\cdot (n-k)^{n-k}}{n^n}= e^{-nH(k/n)},
\]
thus showing the lower bound. 

For the upper bound let $X$ be a uniformly random subset of $\{1,2,\dots,n\}$ of size $k$. We view $X$ as a bit string $(x_1,x_2,\dots,x_n)$ where $x_i=\mathbbm{1}_{i\in X}$. Note that $\prob[x_i=1]=k/n$. We then have 
\[
\log{\binom{n}{k}}=H(X)\leq H(x_1)+\dots+H(x_n)=n\cdot H(x_1)=n\cdot H(k/n),
\]
which implies the upper bound.
\end{proof}
We frequently deal with the expressions of the form $(m/x)^x$ or $(m/x)^{x/2}$. The following lemma gives an upper bound.
\begin{lemma}
For any $m>0$ and any $x>0$ the following inequalities hold true
\[
\left(\frac{m}{x}\right)^x\leq e^{m/e},\qquad \left(\frac{m}{x}\right)^{x/2}\leq e^{\frac{m}{2e}}.
\]
\end{lemma}
\begin{proof}
The second inequality follows from the first one by taking a square root. To prove the first inequality note that for $f(t)=\frac{\log{t}}{t}$ we have
\[
f'(t)=\frac{1-\log{t}}{t^2},
\]
which means that $f(t)$ is increasing on $(0,e)$ and decreasing on $(e,\infty)$. Thus, $f(t)\leq f(e)$ for every $t>0$ and so
\[
\left(\frac{m}{x}\right)^x=\left((m/x)^{x/m}\right)^m=\left(e^{f(m/x)}\right)^m\leq \left(e^{f(e)}\right)^m=e^{m/e}. 
\]
\end{proof}
We also deal with factorials and functions of the form $\ell^\ell$. The following lemma says that up to exponentially large factors they are the same and enjoy the property $f(x+y)\approx f(x)f(y)$.
\begin{lemma}
\label{quotient of factorials}
Let $a_1,\dots a_m$ be positive integers and $S=\sum a_i$. Let $f_1,\dots,f_m$ be functions, each of them is either $x\mapsto x^x$ or $x\mapsto x!$. Then 
\[
\left(\frac{S}{3m}\right)^S \leq \prod_{i=1}^m f_i(a_i) \leq S^S.
\]
\end{lemma}
\begin{proof}
For the upper bound note that 
\[
\prod_{i=1}^m f_i(a_i)\leq \prod_{i=1}^m a_i^{a_i}\leq \prod_{i=1}^m S^{a_i}=S^S.
\]
For the lower bound we have 
\[
\prod_{i=1}^m f_i(a_i)\geq \prod_{i=1}^m a_i!\geq \prod_{i=1}^m \left(\frac{a_i}{e}\right)^{a_i}=e^{-S}\prod_{i=1}^m a_i^{a_i}.
\]
Consider a multiset of $S$ integers where each $a_k$ appears $a_k$ times. Sum of reciprocals of numbers in this set is then $m$, which means that the harmonic mean is $S/m$. The geometric mean is at least as large as the harmonic mean so
\[
\sqrt[S]{\prod_{i=1}^m a_i^{a_i}}\geq \frac{S}{m}. 
\]
This implies that 
\[
\prod_{i=1}^m f_i(a_i)\geq e^{-S}\left(\frac{S}{m}\right)^S> \left(\frac{S}{3m}\right)^S.
\]
\end{proof}
\begin{rem}
In all the applications we have $m\leq 4$. We also frequently apply the lemma for a quotient of two products, in which case what remains (up to exponential factor) is $f(\sum a_i-\sum b_j)$ where $a_i$'s are in the numerator and $b_j$'s are on the denominator. 
\end{rem}
\begin{rem}
If applying the lemma we have expression of the form $\ell^{\ell/2}$ it should be replaced with $(\ell/2)^{\ell/2}\cdot 2^{\ell/2}=f(\ell/2)\cdot 2^{\ell/2}$.
\end{rem}

\section{Fourier spectrum of the component growing protocol from section~\ref{sec:simple-protocol}} \label{sec:htft-app}
Suppose that the players use the component growing protocol described in section~\ref{sec:simple-protocol}. For each $t=1,\ldots, T$ let $E^*_t$ denote the set of edges of the forest $F_t$. In that case the set of possible values of $X^*$ consistent with the players' knowledge at time $t$ can be defined quite easily:
\begin{equation*}
\B_t=\{x\in \bool^n: \forall e=(a, b)\in E^*_t, x_a+x_b=w_e\}.
\end{equation*}

Thus, $\B_t$ is simply a linear subspace of $\bool^n$, where the constraints are given by the edges in $E^*_t$. Recall that we denote the indicator of $\B_t$ by $h_t$. We now derive a characterization of $\wh{h}_t$.  We call a coefficient $v\in \bool^n$ {\em admissible} if it has an even intersection with every connected component in $E^*_t$ (see Fig.~\ref{fig:clusters-ft}(a), where the elements of an admissible $v$ are marked red). For each admissible $v$ let $Q(v)$ denote the unique pairing of vertices of $v$ via edge-disjoint paths in $E^*_t$ (we associate $Q_v$ with the set of edges on these paths).  This is illustrated in Fig.~\ref{fig:clusters-ft}, where the vertices of $v\in \bool^n$ are marked red, and the edges of $Q(v)$ are the green dashed edges. Note that since edges of $E_t^*$ form a forest, this pairing is indeed unique for every admissible $v$.

We show that $\wh{h}_t(v)$ has the following simple form:
\begin{equation}\label{eq:whht}
\wh{h}_t(v)=\left\lbrace
\begin{array}{cc}
\frac{|\B_t|}{2^n}\cdot (-1)^{\sum_{e\in Q(v)} w_e}&\text{if~}v\text{~is admissible}\\
0&\text{o.w.}.
\end{array}
\right.
\end{equation}
Recall that $w_e$ are the labels on the edges of the graph $G_t$ formed by first $t$ matchings.

We now prove~Eq.~\eqref{eq:whht}. We first prove that $\wh{h}_t(v)=0$ for any inadmissible $v$. By definition of the Fourier transform 
\begin{equation*}
\begin{split}
\wh{h}_t(v)&=\frac1{2^n} \sum_{w\in \bool^n}  h_t(w) (-1)^{v\cdot w}=\frac1{2^n} \sum_{w\in \B_t}  (-1)^{v\cdot w}\\
&=\frac{|\{w\in \B_t: v\cdot w~\text{is even}\}|-|\{w\in \B_t: v\cdot w~\text{is odd}\}|}{2^n}.
\end{split}
\end{equation*}

Now suppose that $v\in \bool^n$ has an odd intersection with at least one of the connected components in $E^*_t$.  Denote the set of vertices in this component by $\mathcal{C}_*\subseteq [n]$, and let $\mathbf{1}_{\mathcal{C}_*}\in \bool^n$ denote the indicator vector of vertices in $\mathcal{C}_*$. We now note that for any $w\in \B_t$ one necessarily has that $w+\mathbf{1}_{\mathcal{C}_*}\in \B_t$. Indeed, adding $1$ to every vertex in $\mathcal{C}_*$ could only violate those constraints (edges) that have exactly one endpoint in $\mathcal{C}_*$. But there are no such edges since $\mathcal{C}_*$ is a connected component by definition, so $w+\mathbf{1}_{\mathcal{C}_*}\in \B_t$ as required. On the  other hand, $v$ has an odd intersection with $\mathcal{C}_*$, we have $v\cdot \mathbf{1}_{\mathcal{C}_*}=1$, so for any $w\in \bool^n$
$$
v\cdot (w+\mathbf{1}_{\mathcal{C}_*})=v\cdot w+1.
$$
This means that $|\{w\in \B_t: v\cdot w~\text{is even}\}|=|\{w\in \B_t: v\cdot w~\text{is odd}\}$, since the map $w\to w+\mathbf{1}_{\mathcal{C}_*}$ is an involution on $\B_t$, and hence $\wh{h}_t(v)=0$ as required.

Now suppose that $v$ is admissible. To derive the equation for $\wh{h}_t$ given in~Eq.~\eqref{eq:whht}, we note that $\B_t$ can be alternatively characterized as follows. Pick any element $w^*\in \B_t$, and let $\mathcal{C}_1,\ldots, \mathcal{C}_k$ denote the connected components in $E^*_t$ (so that each singleton node is a connected component of its own). Then 
\begin{equation*}
\B_t=\{ w^*+\sum_{i=1}^k \lambda_i \mathbf{1}_{\mathcal{C}_i}: \lambda\in \bool^k\}.
\end{equation*}
Noting that $v\cdot \mathbf{1}_{\mathcal{C}_i}=0$ for all $i=1,\ldots, k$, we note that for any $w\in \bool^n$
$$
v\cdot w=v\cdot (w^*+\sum_{i=1}^k \lambda_i \mathbf{1}_{\mathcal{C}_i})=v\cdot w^*.
$$ 
We thus have
\begin{equation*}
\begin{split}
\wh{h}_t(v)&=\frac1{2^n} \sum_{w\in \B_t}  (-1)^{v\cdot w}=\frac{|\B_t|}{2^n}(-1)^{v\cdot w^*}.
\end{split}
\end{equation*}
Let $\Lambda$ be the the unique pairing of vertices of $v$ by edge disjoint paths in $F_t$. For any pair $\{i, j\}\in \Lambda$ let $\mathcal{P}_{i, j}$ denote the path from $i$ to $j$ in $F_t$. 
Note that 
$$
\sum_{e\in \mathcal{P}_{i, j}} w_e=\sum_{e=(a, b)\in \mathcal{P}_{i, j}} (X^*_a+X^*_b)=X^*_i+X^*_j.
$$
To complete the proof, it suffices to note that 
$$
v\cdot w^*=\sum_{\{i, j\}\in \Lambda} X^*_i+X^*_j=\sum_{\{i, j\}\in \Lambda} \sum_{e\in \mathcal{P}_{i, j}} w_e=\sum_{e\in Q(v)} w_e
$$
as required.

We now turn to the Fourier transform of $f_t$ for the component growing protocol. Let $M^*_t:=M_t\cap E_t^*$ denote the set of edges of $M_t$ whose bits are revealed by the $t$-th player. We have $E^*_t=E^*_{t-1}\cup M^*_t$ for all $t=1,\ldots, T$. Note that the simple rule that defines $M^*_t$ immediately specifies the set $\A_t$ for our component growing protocol (recall that $\A_t$ was defined in Definition~\ref{def:at}). For completeness, we instantiate both the definitions of $\A_{t, reduced}$ and $\A_t$ for our component growing protocol. The set $\A_{t, reduced}$ is simply
$$
\A_{reduced, t}=\{z\in \bool^{M_t}: z_e=(w_t)_e\text{~for all~}e\in M^*_t\}.
$$
Thus, $\A_{reduced, t}$ is a subcube of the boolean hypercube $\bool^{M_t}$ obtained by fixing coordinates in $M^*_t$ to their specified values. The set $\A_t$ is defined as $\A_t=\{x\in \bool^n: M_t x\in \A_{t, reduced}\}$ (see Definition~\ref{def:at}), and in our case is a linear subspace of $\bool^n$:
$$
\A_t=\{x\in \bool^n: \text{~for all~}e=(a, b)\in M^*_t\text{~}x_a+x_b=(w_t)_e\}.
$$
Recall that the function $f_t$ is defined as the indicator of $\A_t$ (see Definition~\ref{def:at}). Since $\A_t$ is a linear subspace just like $\B_t$, the Fourier transform of $f_t$ is also quite easy to understand. 
We note that the same derivation as for $h_t$ shows that the Fourier transform of $\wh{f}_t$ is supported on edges of $M_t$. We now say that $v\in \bool^n$ is admissible if it has even intersection with every connected component in $M^*_t$. But this can only happen when $v$ is a union of edges of $M^*_t$, as required.


\section{Proof of Lemmas \ref{lm:L1xorKKL} and \ref{lm:matching-fourier}}\label{app-F}

\subsection{General bound for the $L_1$ mass of the Fourier transform}  
Here we prove Lemma \ref{lm:L1xorKKL}. The starting point is the following Lemma which can be found in \cite{KKL}.
\begin{lemma}
Let $f$ be a function $f:\{0,1\}^m\rightarrow \{-1,0,1\}$ and let $\A=f^{-1}(\{-1,1\})$. Let $|s|$ denote the Hamming weight of $s\in \{0,1\}^m$. Then for every $\delta\in[0,1]$
\[
\sum_{s\in \{0,1\}^m} \delta^{|s|}\widehat{f}^2(s)\leq \left(\frac{|\A|}{2^m}\right)^{\frac{2}{1+\delta}}.
\]
\end{lemma}
\xxx[MK]{Are we ever referencing these lemmas?}
We then deduce 
\begin{lemma}\label{lm:hypercontractivity-fourier}
Let $f$ be a function on $\{0,1\}^m$ taking values in $\{-1,0,1\}$. Define a set $\A=f^{-1}(\{-1,1\})$. Then if $|\A|\geq 2^{m-d}$ and $q\leq d$ then
\[
\left(\frac{2^m}{|\A|}\right)^2\sum_{\substack{x\in \{0,1\}^m \\ |x|=q }} \widehat{f}^2(x)  \leq  \left(\frac{4d}{q}\right)^q.
\]
\end{lemma}
\begin{proof}
By the Lemma above, for every $\delta\in [0,1]$ we have
\[
\left(\frac{2^m}{|\A|}\right)^2\sum_{\substack{x\in \{0,1\}^m \\ |x|=q }} \widehat{f}^2(x)\leq \frac{2^{2m}}{|\A|^2}\delta^{-q}\left(\frac{|\A|}{2^m}\right)^{\frac{2}{1+\delta}}=\delta^{-q}\left(\frac{2^m}{|\A|}\right)^{\frac{2\delta}{1+\delta}} \leq \delta^{-q}\left(\frac{2^m}{|\A|}\right)^{2\delta}\leq \frac{2^{2\delta d}}{\delta^q}.
\]
Plugging in $\delta=\lambda q/d$ with $\lambda\in[0,1]$ (which ensures that $\delta\in [0,1]$) we obtain
\[
\left(\frac{2^m}{|\A|}\right)^2\sum_{\substack{x\in \{0,1\}^m \\ |x|=q }} \widehat{f}^2(x)\leq \frac{2^{2\delta d}}{\delta^q} = \left(\frac{2^{2\lambda} d}{\lambda q}\right)^q.
\] 
It remains to note that for $\lambda=\frac{1}{2\log{2}}$ we have $2^{2\lambda}/\lambda=2e\log{2}<4$.
\end{proof}

\begin{lemma}\label{xorKKL}
Let $\A\subset \{0,1\}^m$ be a set of cardinality at least $2^{m-d}$ with indicator function $f$. Then for every $y\in \{0,1\}^m$ and every $q\leq d$ one has 
\[
\sum_{\substack{x\in \{0,1\}^m \\ |x\oplus y|=q}} \wt{f}^2(x) \leq \left(\frac{4d}{q}\right)^{q}.
\]
\end{lemma}
\begin{proof}
Recall that $\wt{f}(x):=\frac{2^m}{|\A|}\cdot \widehat{f}(x)$, see definition \ref{def:tilde}.  So the statement is equivalent to the following inequality:
\[
\left(\frac{2^m}{|\A|}\right)^2 \sum_{\substack{x\in \{0,1\}^m \\ |x\oplus y|=q}} \widehat{f}^2(x) \leq \left(\frac{4d}{q}\right)^{q}.
\]
Consider a function $g(z):=f(z)(-1)^{z \cdot y}$. For this function we have 
\[
\widehat{g}(x)=2^{-m}\sum_{z\in \{0,1\}^m } g(z)(-1)^{z\cdot x}=2^{-m}\sum_{z\in \{0,1\}^m} f(z)(-1)^{z\cdot (x\oplus y)}=\widehat{f}(x\oplus y).
\]
It remains to apply the previous lemma to the function $g(z)$.
\end{proof}
Since there are exactly $\binom{m}{q}$ different $x\in\{0,1\}^m$ for which $|x\oplus y|=q$, the above Lemma together with Lemma \ref{ineq:CS} imply Lemma \ref{lm:L1xorKKL}.

\subsection{Structure of the Fourier transform of a single player's message}
Here we prove Lemma \ref{lm:matching-fourier}.

\begin{proof}
We compute the Fourier transform of $f(x)$. For $z\in \{0, 1\}^{\alpha n}$ let $x(z)\in\bool^n$ be defined by setting, for each edge $(u, v)\in M$,
$$
x(z)_{u}=z_e \text{~and~}x(z)_{v}=0
$$
and $x(z)_w=0$ if $w$ is not matched by $M$.  Note that $x(z)$ is a particular solution of $Mx=z$ and the set of solutions is given by 
\begin{equation}\label{eq:sol-space}
\{x(z)+Ns: s\in \{0, 1\}^{n-\alpha n}\},
\end{equation}
where $N$ is a basis for the kernel of $M$. Without loss of generality suppose that $M$ contains the edges $(2i-1, 2i), i=1,\ldots, \alpha n$. Then 
the matrix $N\in \{0, 1\}^{n\times (n-\alpha n)}$ may be taken as 
\begin{equation*}
\left(
\begin{array}{llllllll}
1&0&0&\ldots&0&0&\ldots&0\\
1&0&0&\ldots&0&0&\ldots&0\\
0&1&0&\ldots&0&0&\ldots&0\\
0&1&0&\ldots&0&0&\ldots&0\\
\vdots&\vdots&\vdots&\vdots&\vdots&\vdots&\vdots&\vdots\\
0&0&1&\ldots&0&0&0&0\\
0&0&1&\ldots&0&0&0&0\\
0&0&0&\ldots&1&0&0&0\\
0&0&0&\ldots&0&1&0&0\\
0&0&0&\ldots&0&0&\ddots&0\\
0&0&0&\ldots&0&0&0&1\\
\end{array}
\right),
\end{equation*}
where first $\alpha n$ columns form $M^T$, the bottom right submatrix is the $(n-2\alpha n)\times (n-2\alpha n)$ identity, all the other entries are zero.

The Fourier transform of $f$ at $v\in \{0, 1\}^{n}$ is given by
\begin{equation*}
\begin{split}
\hat f(v)&=\frac1{2^n} \sum_{x\in \{0, 1\}^n} f(x) \cdot (-1)^{x\cdot v}\\
&=\frac1{2^n} \sum_{z\in \A_{reduced}}\sum_{s\in \{0, 1\}^{n-\alpha n}} (-1)^{(x(z)+Ns)\cdot v}\\
&=\frac1{2^n} \sum_{z\in \A_{reduced}}(-1)^{x(z)\cdot v}\sum_{s\in \{0, 1\}^{n-\alpha n}} (-1)^{(v^T N)\cdot s}\\
\end{split}
\end{equation*}

First note that
$$
\sum_{s\in \{0, 1\}^{n-\alpha n}} (-1)^{(v^TN)\cdot s}=\mathbf{1}_{v^TN=0}\cdot 2^{n-\alpha n},
$$
so $\hat f(v)=0$ unless $v^T N =0$.  Note that all such $v$ are of the form $v=M^T r$ for some $r\in \{0, 1\}^{\alpha n}$.

Thus,
\begin{equation*}
\begin{split}
\hat f(M^Tr)&=\frac{2^{n-\alpha n}}{2^n} \sum_{z\in \A_{reduced}}(-1)^{x(z)\cdot M^T r}=\frac{2^{n-\alpha n}}{2^n} \sum_{z\in \A_{reduced}}(-1)^{z\cdot r}=\hat q(r)\\
\end{split}
\end{equation*}
and $\hat f(v)=0$ for all $v$ not of the form $M^Tr$. Here we use the fact that $x(z)\cdot M^Tr=z\cdot r$ for all $z$ and $r$.
 
Note that Fourier coefficients of $f$ only have even weight, and weight $k$ Fourier coefficients of $q$ are in direct correspondence with weight $2k$ coefficients of $f$ (since $|M^Tr|=2|r|$ for all $r\in \bool^{\alpha n}$). 
\end{proof}

\end{document}